\pdfoutput=1
\documentclass[
  12pt,
  secnumarabic,
  notitlepage,
  tightenlines,
  nofootinbib,
  superscriptaddress,
  longbibliography,
  aps,
  pra
]{revtex4-2}

\usepackage{newpxtext,newpxmath}

\usepackage[utf8]{inputenc}
\usepackage{newunicodechar}
\newunicodechar{ρ}{\rho}
\newunicodechar{β}{\beta}
\newunicodechar{σ}{\sigma}
\newunicodechar{λ}{\lambda}
\newunicodechar{Λ}{\Lambda}
\newunicodechar{μ}{\mu}
\newunicodechar{ν}{\nu}
\newunicodechar{ψ}{\psi}
\newunicodechar{ϕ}{\varphi}
\newunicodechar{Φ}{\Phi}
\newunicodechar{φ}{\phi}
\newunicodechar{π}{\pi}
\newunicodechar{α}{\alpha}
\newunicodechar{ϵ}{\epsilon}
\newunicodechar{ε}{\varepsilon}
\newunicodechar{δ}{\delta}
\newunicodechar{ω}{\omega}
\newunicodechar{Δ}{\Delta}
\newunicodechar{Σ}{\Sigma}

\usepackage[thmtools-compat]{keytheorems}
\usepackage{amsthm}

\usepackage{amssymb}
\usepackage{amsmath}
\usepackage{bm,bbm}
\usepackage{braket}
\usepackage{dsfont}
\usepackage{mathdots}
\usepackage{mathtools}
\usepackage{enumerate}
\usepackage[shortlabels]{enumitem}
\usepackage{csquotes}
\usepackage{stmaryrd}
\usepackage{graphicx}
\usepackage{stackengine}
\usepackage{scalerel}
\usepackage{tensor}       %\tensor[_n]{\braket{j|\psi}}{_{1\ldots n}}
\usepackage[dvipsnames,svgnames]{xcolor}
\usepackage{xfrac}
\usepackage{centernot}
\usepackage{comment}
\usepackage{chngcntr}
\usepackage{booktabs}
\usepackage{tabularray}
\usepackage[normalem]{ulem}
\UseTblrLibrary{booktabs}

\usepackage[pdftex]{hyperref}
\hypersetup{
    bookmarksnumbered=true, % If Acrobat bookmarks are requested, include section numbers
    breaklinks=true,
    unicode=false, % non-Latin characters in Acrobat bookmarks
    pdfstartview={FitH}, % fits the width of the page to the window
    pdfnewwindow=true, % links in new window
    colorlinks=true, % false: boxed links; true: colored links
    allcolors=MidnightBlue!70!black!70!TealBlue
}
\usepackage{zref-clever}
\zcsetup{cap}

\makeatletter
\def\thmhead@plain#1#2#3{%
  \thmname{#1}\thmnumber{\@ifnotempty{#1}{ }\@upn{#2}}%
  \thmnote{ {\the\thm@notefont#3}}}
\let\thmhead\thmhead@plain
\makeatother

\theoremstyle{plain}
\newtheorem{theorem}{Theorem}

\newtheorem*{theorem*}{Theorem}
\newtheorem*{thm*}{Theorem}

\newtheorem{lemma}{Lemma}
\newtheorem*{lemma*}{Lemma}

\newtheorem{proposition}{Proposition}

\newtheorem*{proposition*}{Proposition}
\newtheorem*{prop*}{Proposition}

\newtheorem{corollary}{Corollary}

\newtheorem*{corollary*}{Corollary}
\newtheorem*{cor*}{Corollary}

\newtheorem*{conjecture*}{Conjecture}
\newtheorem*{cj*}{Conjecture}

\theoremstyle{definition}
\newtheorem{definition}{Definition}

\newtheorem*{definition*}{Definition}
\newtheorem*{Def*}{Definition}

\newtheorem*{example*}{Example}
\newtheorem*{ex*}{Example}

\newtheorem*{question*}{Question}

\newtheorem*{problem*}{Problem}

\newtheorem*{axiom*}{Axiom}

\newtheorem*{fact*}{Fact}

\newtheorem*{claim*}{Claim}

\newtheorem*{assumption*}{Assumption}

\newtheorem*{observation*}{Observation}

\theoremstyle{remark}
\newtheorem{remark}{Remark}

\newtheorem*{remark*}{Remark}
\newtheorem*{rem*}{Remark}

\newtheorem*{note*}{Note}

\makeatletter

\renewcommand{\section}{%
  \@startsection
    {section}{1}{\z@}%
    {2.5ex \@plus 1ex \@minus .2ex}%
    {1.2ex \@plus .2ex}%
    {\normalfont\large\bfseries\raggedright}%
}

\renewcommand{\subsection}{%
  \@startsection
    {subsection}{2}{\z@}%
    {2ex \@plus 1ex \@minus .2ex}%
    {.8ex \@plus .2ex}%
    {\normalfont\normalsize\bfseries\raggedright}%
}

\renewcommand{\subsubsection}{%
  \@startsection
    {subsubsection}{3}{\z@}%
    {1.5ex \@plus 1ex \@minus .2ex}%
    {.6ex \@plus .2ex}%
    {\normalfont\normalsize\itshape\raggedright}%
}

\def\@hangfrom@section#1#2#3{\@hangfrom{#1#2}#3}
\def\@hangfroms@section#1#2{#1#2}
\def\@hangfrom@subsection#1#2#3{\@hangfrom{#1#2}#3}
\def\@hangfroms@subsection#1#2{#1#2}
\def\@hangfrom@subsubsection#1#2#3{\@hangfrom{#1#2}#3}
\def\@hangfroms@subsubsection#1#2{#1#2}

\makeatother

\theoremstyle{plain}

\newcommand{\bb}{\begin{equation}\begin{aligned}\hspace{0pt}}
\newcommand{\bbb}{\begin{equation*}\begin{aligned}}
\newcommand{\ee}{\end{aligned}\end{equation}}
\newcommand{\eee}{\end{aligned}\end{equation*}}

\renewcommand{\epsilon}{\varepsilon}

\newcommand{\ve}{\varepsilon}

\DeclareMathOperator{\Tr}{Tr}

\DeclareMathAlphabet{\pazocal}{OMS}{zplm}{m}{n}

\DeclareMathOperator{\supp}{supp}

\newcommand{\lsmatrix}{\left(\begin{smallmatrix}}
\newcommand{\rsmatrix}{\end{smallmatrix}\right)}

\stackMath

\stackMath

\makeatletter
\newcommand*\rel@kern[1]{\kern#1\dimexpr\macc@kerna}
\newcommand*\widebar[1]{%
  \begingroup
  \def\mathaccent##1##2{%
    \rel@kern{0.8}%
    \overline{\rel@kern{-0.8}\macc@nucleus\rel@kern{0.2}}%
    \rel@kern{-0.2}%
  }%
  \macc@depth\@ne
  \let\math@bgroup\@empty \let\math@egroup\macc@set@skewchar
  \mathsurround\z@ \frozen@everymath{\mathgroup\macc@group\relax}%
  \macc@set@skewchar\relax
  \let\mathaccentV\macc@nested@a
  \macc@nested@a\relax111{#1}%
  \endgroup
}

\counterwithin*{equation}{part}
\counterwithin*{thm}{part}
\counterwithin*{figure}{part}

\definecolor{Blues5seq1}{RGB}{239,243,255}
\definecolor{Blues5seq2}{RGB}{189,215,231}
\definecolor{Blues5seq3}{RGB}{107,174,214}
\definecolor{Blues5seq4}{RGB}{49,130,189}
\definecolor{Blues5seq5}{RGB}{8,81,156}

\definecolor{Greens5seq1}{RGB}{237,248,233}
\definecolor{Greens5seq2}{RGB}{186,228,179}
\definecolor{Greens5seq3}{RGB}{116,196,118}
\definecolor{Greens5seq4}{RGB}{49,163,84}
\definecolor{Greens5seq5}{RGB}{0,109,44}

\definecolor{Reds5seq1}{RGB}{254,229,217}
\definecolor{Reds5seq2}{RGB}{252,174,145}
\definecolor{Reds5seq3}{RGB}{251,106,74}
\definecolor{Reds5seq4}{RGB}{222,45,38}
\definecolor{Reds5seq5}{RGB}{165,15,21}

\allowdisplaybreaks

\let\nc\newcommand

\nc{\lset}{\left\{\left.}
\nc{\rset}{\right\}}
\nc{\lsetr}{\left\{\,}
\nc{\rsetr}{\right.\right\}}
\nc{\barr}{\;\rule{0pt}{9.5pt}\left|\;}

\usepackage[most,breakable]{tcolorbox}
  {\expandafter\ifstrequal\expandafter{#1}{filled}{\begin{tcolorbox}[colback=MidnightBlue!70!black!70!TealBlue!2!white,colframe=MidnightBlue!70!black!70!TealBlue!30!white,breakable,enhanced,left=5.75pt,right=5.75pt,grow sidewards by=10pt]}{\begin{tcolorbox}[colback=white,colframe=gray!15,breakable,enhanced,left=5.75pt,right=5.75pt,grow sidewards by=10pt]}}%
  {\end{tcolorbox}}

\nc{\R}{\mathcal{R}}
\nc{\W}{\mathcal{W}}
\nc{\T}{\mathcal{T}}
\nc{\B}{\mathcal{B}}
\nc{\C}{\mathcal{C}}
\nc{\U}{\mathcal{U}}
\nc{\E}{\mathcal{E}}
\nc{\EE}{\mathscr{E}}
\nc{\K}{\mathcal{K}}
\nc{\V}{\mathcal{V}}
\nc{\X}{\mathcal{X}}
\nc{\F}{\mathcal{F}}
\nc{\G}{\mathcal{G}}
\nc{\D}{\mathcal{D}}
\nc{\Y}{\mathcal{Y}}

\nc{\M}{\mathcal{M}}
\nc{\N}{\mathcal{N}}
\nc{\I}{\mathcal{I}}
\nc{\Q}{\mathbb{Q}}
\nc{\RR}{\mathbb{R}}
\nc{\CC}{\mathbb{C}}

\nc{\HH}{\mathbb{H}}
\nc{\MM}{\mathbb{M}}
\nc{\NN}{\mathbb{N}}
\nc{\DD}{\mathbb{D}}

\nc{\id}{\mathbbm{1}}
\nc{\idc}{\mathrm{id}}

\nc{\norm}[2]{\left\lVert#1\right\rVert_{\,#2}}
\nc{\proj}[1]{\ket{#1}\!\bra{#1}}
\nc{\lnorm}[2]{\left\lVert#1\right\rVert_{\ell_{#2}}}

\let\oldproofname\proofname
\renewcommand{\proofname}{\rm\bf{\oldproofname}}

\makeatletter
\renewenvironment{proof}[1][\proofname]{\par
\pushQED{\qed}%
\normalfont \topsep6\p@\@plus6\p@\relax
\trivlist
\item\relax
{\bfseries  %<- font shape
#1\@addpunct{.}}\hspace\labelsep\ignorespaces %<- includes punctuation code
}{%
\popQED\endtrivlist\@endpefalse
}
\makeatother

\nc{\rhos}{\rho'}

\nc{\gm}{\mathbin\#}

\newcommand{\sPnorm}{P,=}
\newcommand{\sPsub}{P,\leq}
\newcommand{\sTnorm}{T,=}
\newcommand{\sTsub}{T,\leq}
\newcommand{\smoothing}[2]{#2,\,#1}
\NewDocumentCommand{\s}{m m O{\ve}}{%
 \lowercase{\def\sdist{#1}}
 \lowercase{\def\snorm{#2}}
 \IfEqCase{\sdist}{%
  {t}{%
   \IfEqCase{\snorm}{%
    {norm}{ \smoothing{\sTnorm}{#3} }%
    {n}{ \smoothing{\sTnorm}{#3} }%
    {sub}{ \smoothing{\sTsub}{#3}}%
    {s}{ \smoothing{\sTsub}{#3}}%
   }%
  }%
  {p}{%
   \IfEqCase{\snorm}{%
    {norm}{ \smoothing{\sPnorm}{#3} }%
    {n}{ \smoothing{\sPnorm}{#3} }%
    {sub}{ \smoothing{\sPsub}{#3}}%
    {s}{ \smoothing{\sPsub}{#3}}%
   }%
  }%
 }%
}%
\NewDocumentCommand{\Dmax}{m m O{\ve}}{%
 D_{\max}^{\s{#1}{#2}[#3]}
}%
\makeatletter

\def\@sect@ltx#1#2#3#4#5#6[#7]#8{%
    \@ifnum{#2>\c@secnumdepth}{%
        \def\H@svsec{\phantomsection}%
        \let\@svsec\@empty
    }{%
        \H@refstepcounter{#1}%
        \def\H@svsec{%
            \phantomsection
        }%
        \protected@edef\@svsec{{#1}}%
        \@ifundefined{@#1cntformat}{%
            \prepdef\@svsec\@seccntformat
        }{%
            \expandafter\prepdef
            \expandafter\@svsec
            \csname @#1cntformat\endcsname
        }%
    }%
    \@tempskipa #5\relax
    \@ifdim{\@tempskipa>\z@}{%
        \begingroup
        \interlinepenalty \@M
        #6{%
            \@ifundefined{@hangfrom@#1}{\@hang@from}{\csname @hangfrom@#1\endcsname}%
            {\hskip#3\relax\H@svsec}{\@svsec}{#8}%
        }%
        \@@par
        \endgroup
        \@ifundefined{#1mark}{\@gobble}{\csname #1mark\endcsname}{#7}%
        \addcontentsline{toc}{#1}{%
            \@ifnum{#2>\c@secnumdepth}{%
                \protect\numberline{}%
            }{%
                \protect\numberline{\csname the#1\endcsname}%
            }%
            #7}% <<<<<<<<<<<<<<<<<<<<<<<< changed from #8
    }{%
        \def\@svsechd{%
            #6{%
                \@ifundefined{@runin@to@#1}{\@runin@to}{\csname @runin@to@#1\endcsname}%
                {\hskip#3\relax\H@svsec}{\@svsec}{#8}%
            }%
            \@ifundefined{#1mark}{\@gobble}{\csname #1mark\endcsname}{#7}%
            \addcontentsline{toc}{#1}{%
                \@ifnum{#2>\c@secnumdepth}{%
                    \protect\numberline{}%
                }{%
                    \protect\numberline{\csname the#1\endcsname}%
                }%
                #8}%
        }%
    }%
    \@xsect{#5}}%

\makeatother

\providecommand{\Tr}{\operatorname{Tr}}
\providecommand{\supp}{\operatorname{supp}}

\providecommand{\id}{\mathrm{id}}
\providecommand{\Ad}{\operatorname{Ad}}

\begin{document}
\count\footins=1000

\title{Robust classical and quantum communication over almost-iid channels}
\author{Liuhang Ye}
\email{ly400@cam.ac.uk}
\affiliation{Department of Applied Mathematics and Theoretical Physics, Centre for Mathematical Sciences, University of Cambridge, Cambridge CB3 0WA, United Kingdom}

\author{Bjarne Bergh}
\email{bb536@cam.ac.uk}
\affiliation{Department of Applied Mathematics and Theoretical Physics, Centre for Mathematical Sciences, University of Cambridge, Cambridge CB3 0WA, United Kingdom}

\author{Nilanjana Datta}
\email{n.datta@damtp.cam.ac.uk}
\affiliation{Department of Applied Mathematics and Theoretical Physics, Centre for Mathematical Sciences, University of Cambridge, Cambridge CB3 0WA, United Kingdom}
% \date{\today}
%%%%%%%%%%%%%%%%%%%%%%%%%%%%%%%%%%%%%%%%%%%%%%%%%%%%%%%%%%%%%%%%%%%%%%%%%%%%%%%%
\begin{abstract}
\vspace{0.8em}
{\centering\bfseries Abstract\par}
\vspace{0.4em}
\noindent

The independent and identically distributed (iid) paradigm for asymptotic
communication through quantum channels describes \(n\) uses of a channel
\(\Phi\) by \(\Phi^{\otimes n}\). Since exact tensor-power behaviour is an
idealization, we ask whether the asymptotic capacities of \(\Phi\) remain
achievable for channel sequences that are almost-iid along \(\Phi\). We
consider sequences of Choi-almost-iid channels whose normalized Choi
states are Mazzola--Sutter--Renner (MSR) almost-iid along the normalized
Choi state of \(\Phi\), with defect parameters \(r_n=o(n)\). In Renner's
exponential de Finetti theorem for quantum states, almost-iid states
appear in the approximating mixture and motivate the MSR almost-iid
class. Choi-almost-iid channels play the corresponding role in an
exponential de Finetti theorem for quantum channels established in a companion paper. In this paper, we prove that the universal unassisted classical and quantum capacities of the class of Choi-almost-iid channels along $\Phi$ equal those of the channel~$\Phi$. Here universality is within 
this class: for each rate below the relevant
capacity, a code sequence depending on \(\Phi\) and the rate, but not on
the particular channel sequence, achieves vanishing error for every
sequence in the class.
\end{abstract}

%%%%%%%%%%%%%%%%%%%%%%%%%%%%%%%%%%%%%%%%%%%%%%%%%%%%%%%%%%%%%%%%%%%%%%%%%%%%%%%%
\makeatletter
\@booleanfalse\titlepage@sw
\makeatother
\maketitle

\tableofcontents
%%%%%%%%%%%%%%%%%%%%%%%%%%%%%%%%%%%%%%%%%%%%%%%%%%%%%%%%%%%%%%%%%%%%%%%%%%%%%%%%
%%%%%%%%%%%%%%%%%%%%%%%%%%%%%%%%%%%%%%%%%%%%%%%%%%%%%%%%%%%%%%%%%%%%%%%%%%%%%%%%
%%%%%%%%%%%%%%%%%%%%%%%%%%%%%%%%%%%%%%%%%%%%%%%%%%%%%%%%%%%%%%%%%%%%%%%%%%%%%%%%

\section{Introduction}
The memoryless, or independent and identically distributed (iid),
setting has long been the standard paradigm for determining optimal
asymptotic rates in information theory. For a memoryless quantum source characterized at the single-use level
by a density operator $\rho$, the corresponding $n$-use description
is $\rho^{\otimes n}$, while $n$ independent uses of a quantum
channel $\Phi$ are described by $\Phi^{\otimes n}$. This tensor-power
structure underlies the asymptotic analysis of a broad range of
information-theoretic tasks. For quantum states, representative
examples include quantum data compression, entanglement concentration
and dilution, and distributed protocols such as quantum state merging
and state redistribution
\cite{schumacher1995quantum,bennett1996concentrating,
horodecki2007quantum,devetak2008exact}. For quantum channels,
it underlies the classical, quantum, and private capacities, as well
as entanglement-assisted and other assisted capacities
\cite{schumacher1997sending,holevo1998capacity,lloyd1997capacity,shor2002quantum,
devetak2005private,bennett2002entanglement}. The tensor-power
assumption provides concentration and typicality properties that are
central to many of the corresponding coding theorems.

Exact iid behaviour is, however, an idealization. The successive outputs of a physical source may be correlated, while
multiple uses of a communication device may be affected by slow
variations in its physical parameters, memory effects arising from
unavoidable system--environment interactions, or a small number of
faulty components.
Moreover, a single state on a growing number of systems cannot
generally be subjected to full tomography without destroying the
resource that is to be used. It is therefore difficult to justify
operationally that a source or channel has an exact tensor-power
structure for arbitrarily large blocklengths. This has motivated
several approaches to information theory beyond the iid setting. The
information-spectrum method provides general asymptotic formulas for
arbitrary sequences of sources and channels
\cite{han1993approximation,verdu1994general,hayashi2006quantum,bowen2006beyond}, while
one-shot information theory characterizes finite-resource tasks in
terms of operational quantities such as smooth entropies and
hypothesis-testing divergences
\cite{renner2008security,datta2009min,tomamichel2015quantum}. These frameworks apply far
beyond the memoryless regime, but they generally characterize
performance in terms of the full non-iid sequence or the particular
finite-blocklength resource. A complementary question is whether the
familiar iid rates, and protocols depending only on an iid reference
resource, remain valid under structured but otherwise unknown
deviations from the tensor-power model.

This question has recently led to the study of almost-iid quantum
sources. Three inequivalent notions have emerged. In increasing order
of generality, these are Mazzola--Sutter--Renner (MSR) almost-iid
sources, Wasserstein almost-iid sources, and weakly almost-iid sources
\cite{mazzola2026almost,girardi2026new}. They form a
strict hierarchy. An MSR almost-iid state along $\rho$ admits a
permutation-invariant extension whose support is contained in the
span of vectors having the tensor-power structure of a fixed
purification of $\rho$ on all but a sublinear number of systems.
Wasserstein almost-iid sources are characterized instead by a
sublinear quantum Wasserstein distance of order $1$ from $\rho^{\otimes n}$,
whereas weakly almost-iid sources require only that, for every fixed
$k$, their $k$-body marginals converge on average to
$\rho^{\otimes k}$. Thus, MSR almost-iid sources possess the strongest
global structure, while weakly almost-iid sources may retain only the
local statistics of the reference state.

{A primary motivation for the MSR almost iid class comes from Renner's
exponential quantum de Finetti theorem
\cite{renner2008security,renner2007symmetry}. The standard finite quantum de Finetti
theorem approximates the reduced state of a symmetric pure state by a
convex mixture of pure iid states, but its error vanishes
only polynomially and only when asymptotically more systems are discarded than
retained.\footnote{For a symmetric pure state on $n$ systems of
local dimension $d$, the standard bound~\cite{christandl2007one} is
$2dk/n$. For fixed $d$, this tends to zero only if
$k/n\to0$; if $k=n^\alpha$ with $\alpha<1$, it decreases only
polynomially, as $O(n^{\alpha-1})$.}
The exponential de Finetti theorem obtains a much more favourable
trade-off by replacing exact tensor-power states in the mixture by
pure almost-power states with a controlled number of
defects. Each such state is a coherent superposition of vectors that
agree with $|\theta\rangle$ on all but at most $r$
positions, with the defect positions allowed to vary between the
vectors in the superposition. In return for this relaxation, the
approximation error contains an exponentially decaying factor and can
vanish even when both the number of discarded systems and the number
of defects are sublinear in $n$; this exponential dependence is the
origin of the theorem's name.\footnote{More precisely, the error is
bounded by
$3(n-k)^d\exp[-(n-k)(r+1)/n]$. For example, $k=n-n^{3/4}$ and
$r=n^{1/2}$ give an error of order
$\exp[-\Omega(n^{1/4})]$ for fixed $d$.}
Since these almost-power states are particular pure members of the
MSR almost-iid class, the exponential de Finetti theorem may be viewed
as representing almost the entire symmetric state by a convex mixture
of MSR almost-iid states rather than exact iid states.
}

This structural motivation raises an operational question: to what
extent are the optimal rates of information-theoretic tasks robust
when iid sources are replaced by almost-iid ones? Mazzola, Sutter and
Renner proved the asymptotic stability of conditional entropy for MSR
almost-iid states and, as an application, the robustness of squashed
entanglement \cite{mazzola2026almost}. Subsequent work
investigated the robustness of hypothesis-testing exponents, generalized quantum Stein's lemma, classical and quantum data-compression rates, and classical
communication rates under different notions of almost-iid structure
\cite{mazzola2026robust,girardi2026quantum,girardi2026generalised}. Entropy concentration, universal
typicality, quantum data compression, and asymmetric hypothesis
testing have also been studied for the substantially larger class of
weakly almost-iid sources \cite{datta2026entropy}. Since MSR
sources retain considerably more global structure, stronger
conclusions can be obtained for tasks that depend sensitively on
coherence and entanglement. In particular, robustness results for
entanglement concentration of pure MSR sources and entanglement
dilution of mixed MSR sources were established in
\cite{datta2026robustness}.

These developments naturally raise the question of whether analogous
robustness results hold for quantum channels. In the present work, we call a channel
$\mathcal N_{A^n\to B^n}$ \emph{Choi-almost-iid} along a reference
channel $\Phi_{A\to B}$ if its normalized Choi state
$J(\mathcal N_{A^n\to B^n})$ is an MSR almost-iid state along $J(\Phi)$. The Choi state is obtained by acting the channel on one-half of the maximally entangled state between a reference system and the channel input; accounting for this reference is essential when studying transmission of entanglement and hence quantum communication.

Further motivation for this definition comes from a companion work \cite{companion-paper},
in which we prove an exponential de Finetti theorem for quantum
channels\footnote{A limitation of our channel exponential de Finetti theorem is that exponential decay is guaranteed only in regimes where a large fraction of the systems still need to be discarded. It does \emph{not} yet reach the characteristic regime of Renner's exponential de Finetti theorem for states, in which nearly all systems can be retained while the error remains exponentially small.}. In the standard channel de Finetti representation, permutation symmetry and no-signalling condition lead to an approximation of the reduced channel by a convex mixture of iid channels~\cite{berta2022semidefinite}. In the exponential counterpart, the iid channels occurring in this
mixture are replaced by almost-iid channels of the type
considered here. The passage from iid to Choi-almost-iid channels
therefore mirrors the passage from tensor-power states to
almost-power states in the exponential quantum de Finetti theorem for states.

{A different notion of almost-iid behaviour for sequences of quantum
channels was introduced in~\cite{girardi2026quantum}, based
on the quantum Wasserstein distance of order one \cite{de2021quantum}. The {\em{club distance}} between two channels
$\mathcal N,\mathcal M\colon A^n\to B^n$ was defined there by
\begin{equation}\label{eq:club}
    \|\mathcal N-\mathcal M\|_{\clubsuit}
    \coloneqq
    \sup_{\rho_{A^n}\in\mathcal D(A^n)}
    \left\|
        \mathcal N(\rho_{A^n})
        -
        \mathcal M(\rho_{A^n})
    \right\|_{W_1},
\end{equation}
where, for states $\rho$ and $\sigma$ on $B^n$,
$\|\rho-\sigma\|_{W_1}$ denotes their quantum Wasserstein distance
of order~$1$. Unlike the diamond distance, the club distance does not
optimize over inputs entangled with an ancillary reference system.
Girardi et al.\ do not consider its stabilized version, since their
application concerns unassisted classical communication. Consequently,
in its present form, this distance is not suited to the study of
quantum capacity or of entanglement-assisted capacities. By contrast,
our Choi-state formulation includes an input reference system by
construction, making it suitable for both classical and quantum
communication; it is also the notion naturally singled out by the
exponential de Finetti theorem for channels.
}

The purpose of this paper is to determine whether the unassisted classical and
quantum capacities of a memoryless reference channel are robust under
arbitrary Choi-almost-iid deviations with a sublinear number of
defects. Let $\mathfrak S_\Phi$ denote the class of all channel
sequences
$\mathbf N=\{\mathcal N^{(n)}\}_{n\in\mathbb N}$ for which
$\mathcal N^{(n)}$ is Choi-almost-iid along $\Phi$ with defect size
$r_n=o(n)$. We require universality: for each target rate, the same sequence of
encoding and decoding maps must work for every
$\mathbf N\in\mathfrak S_\Phi$, without knowledge of the particular
channel sequence or of the extension and defect structure witnessing
its Choi-almost-iid property. We denote the corresponding universal classical and quantum
capacities by $C_{\mathrm{univ}}(\mathfrak S_\Phi)$ and
$Q_{\mathrm{univ}}(\mathfrak S_\Phi)$, see \zcref{Def:universal-classical-capacity} and \zcref{Def:universal-quantum-capacity}, respectively; the subscript
$\mathrm{univ}$ records this universality requirement. Our main results establish that
\begin{equation}
    C_{\mathrm{univ}}(\mathfrak S_\Phi)
    =
    \inf_{\mathbf N\in\mathfrak S_\Phi} C(\mathbf N)
    =
    C(\Phi)
\end{equation}
and
\begin{equation}
    Q_{\mathrm{univ}}(\mathfrak S_\Phi)
    =
    \inf_{\mathbf N\in\mathfrak S_\Phi} Q(\mathbf N)
    =
    Q(\Phi).
\end{equation}
Thus, allowing arbitrary Choi-almost-iid deviations with a sublinear
number of defects causes no loss in either the universally achievable
classical rate or the universally achievable quantum rate {within this class}. We expect analogous robustness results for entanglement-assisted communication and leave their proof to future work.

The proofs concatenate a finite-blocklength inner code for the
reference channel $\Phi$ with an appropriate outer classical or
quantum error-correcting code. The Choi-almost-iid structure implies
that, after blocking, the deviations lie in the span of errors
supported on a sublinear number of blocks, including coherent
superpositions of different defect locations. The outer code corrects both the residual errors of the inner code
and these deviations from perfect iid-ness, with vanishing overall error and
asymptotically negligible rate loss.

The remainder of the paper is organized as follows. In \zcref{sec:math-prel}, we
introduce the notation and preliminary material used throughout. In \zcref{sec:def-almost-iid-channel}, we define Choi-almost-iid channels, establish an equivalent characterization in terms of
defect-structured channel extensions, and show that they map iid
input states to MSR almost-iid output states. We also compare our definition of almost-iid channels with that of \cite{girardi2026quantum}. In \zcref{Sec:robust-classical-communication}, we prove the universal robustness of classical communication. \zcref{Sec:robust-quantum-communication} establishes the corresponding result for quantum communication.
%%%%%%%%%%%%%%%%%%%%%%%%%%%%%%%%%%%%%%%%%%%%%%%%%%%%%%%%%%%%%%%%%%%%%%%%%%%%%%%%
\section{Mathematical preliminaries}
\label{sec:math-prel}

Let $\mathcal{H}$ denote a finite-dimensional Hilbert space, and let $\mathcal{D}(\mathcal{H})$ denote the set of density matrices on $\mathcal{H}$, i.e., the set of positive semidefinite operators of unit trace. We label different quantum systems by capital Roman letters ($A,B,C$, etc.) and often use these letters interchangeably with the corresponding Hilbert space (i.e., we write $A$ instead of $\mathcal{H}_A$). We also concatenate these letters to denote tensor products of systems, i.e., we write $AR$ for $\mathcal{H}_A \otimes \mathcal{H}_R$. We write $d_A\coloneqq\dim\mathcal H_A$. For systems $A$ and $B$, let
$\mathcal L(A,B)$ denote the vector space of linear maps from
$\mathcal H_A$ to $\mathcal H_B$, and abbreviate
$\mathcal L(A)\coloneqq\mathcal L(A,A)$. The identity operator in $\mathcal{L}(\mathcal{H})$ is denoted by $\id_{\mathcal{H}}$, whereas the identity map acting on $\mathcal{L}(\mathcal{H})$ is denoted by $\idc_{\mathcal{H}}$. We use the convention $\mathcal H^{\otimes0}=\mathbb C$ and interpret linear maps between zero-fold tensor products as scalars. We use the shorthand notation \(\phi=|\phi\rangle\!\langle\phi|\) for pure states. The set of quantum channels from $A$ to $B$ is denoted by
$\mathrm{CPTP}(A,B)$. Thus, a channel
$\Phi_{A\to B}\in\mathrm{CPTP}(A,B)$ is a completely positive (CP), trace-preserving (TP) linear map from $\mathcal L(A)$ to $\mathcal L(B)$.
Its normalized Choi state is defined by
\begin{equation}
    J(\Phi_{A\to B})
    \coloneqq
    (\Phi_{A\to B}\otimes\idc_{A'})
    (|\Omega\rangle\langle\Omega|_{AA'}), \qquad
    |\Omega\rangle_{AA'}
    \coloneqq
    \frac{1}{\sqrt{d_A}}
    \sum_{i=1}^{d_A}|i\rangle_A|i\rangle_{A'},
\end{equation}
where $A'\simeq A$, and $\{|i\rangle_A\}_{i=1}^{d_A}$ is a fixed
orthonormal basis of $A$, and
$\{|i\rangle_{A'}\}_{i=1}^{d_A}$ is the corresponding basis of
$A'$. With this normalization, $J(\Phi_{A\to B})$ is a
state satisfying
$\Tr_BJ(\Phi_{A\to B})=\id_{A'}/d_A$. Conversely, a state
$\tau_{BA'}$ is the normalized Choi state of a unique channel from
$A$ to $B$ if and only if $\tau_{A'}=\id_{A'}/d_A$.

A state $\rho_{AB}\in\mathcal D(AB)$ is called an \emph{extension} of
$\rho_A\in\mathcal D(A)$ if $\Tr_B\rho_{AB}=\rho_A$. Analogously, a
channel $\mathcal N_{A\to BE}$ is called an \emph{extension} of
$\mathcal N_{A\to B}$ if
$\Tr_E\circ\mathcal N_{A\to BE}=\mathcal N_{A\to B}$. This is
equivalent to the corresponding state-extension property for their
Choi states:
$\Tr_EJ(\mathcal N_{A\to BE})=J(\mathcal N_{A\to B})$. For integers \(n\geq 1\), let \(S_n\) denote the symmetric group of degree \(n\); for every \(\pi\in S_n\), we denote by
\(U_{\mathcal H^{\otimes n}}^\pi\) the unitary on \(\mathcal H^{\otimes n}\)
that permutes the \(n\) tensor factors according to \(\pi\). We use \([n]\) to denote the set \(\{1,\ldots,n\}\). For operators
$X,Y\in\mathcal L(A,B)$, the \emph{normalized} Hilbert--Schmidt inner product is defined by
$\langle X,Y\rangle_{\mathrm{nHS}}
\coloneqq d_A^{-1}\Tr(X^\dagger Y)$.

\subsection{MSR almost-iid states}
Here we briefly recall the MSR notion of almost-iid states \cite{mazzola2026almost}. Let $\mathcal H_A$ be a Hilbert space, let $|\theta\rangle\in\mathcal H_A$ be a unit vector, and let $n\geq1$ and $0\leq m\leq n$ be integers. Define
\begin{equation}
\label{eq:def-vector-defect-set}
    \mathcal V\!\left(
        \mathcal H_A^{\otimes n},
        |\theta\rangle^{\otimes m}
    \right)
    \coloneqq
    \left\{
         U^\pi_{A^n}
        \left(
            |\theta\rangle^{\otimes m}
            \otimes
            |\omega^{(n-m)}\rangle
        \right)
        :
        \pi\in S_n,\;
        |\omega^{(n-m)}\rangle
        \in\mathcal H_A^{\otimes(n-m)}
    \right\}.
\end{equation}

\begin{definition}[(MSR-almost-iid state~\cite{mazzola2026almost})]
\label{def:MSR almost-iid state}
Let $\sigma_A\in\mathcal D(A)$, let $n\geq1$, and let
$0\leq r\leq n$ be an integer. A state
$\rho_{A^n}\in\mathcal D(A^n)$ is called an
$\binom{n}{r}$-almost-iid state along $\sigma_A$ if there exist a
purification $|\theta\rangle_{AE}$ of $\sigma_A$ and an extension
$\rho_{A^nE^n}$ of $\rho_{A^n}$ such that
\begin{equation}
\label{def:MSR-state}
\begin{aligned}
\textnormal{(i)}\quad&
    U^\pi_{(AE)^n}\rho_{A^nE^n}
    \bigl(U^\pi_{(AE)^n}\bigr)^\dagger
    =
    \rho_{A^nE^n},
    \quad\forall\,\pi\in S_n;                                         \\
\textnormal{(ii)}\quad&
    \supp(\rho_{A^nE^n})
    \subseteq
    \operatorname{span}
    \mathcal V\!\left(
        \mathcal H_{AE}^{\otimes n},
        |\theta\rangle_{AE}^{\otimes(n-r)}
    \right).
\end{aligned}
\end{equation}
Here $ U^\pi_{(AE)^n}$ simultaneously permutes the $n$ composite
systems $A_iE_i$. The set of all
$\binom{n}{r}$-almost-iid states along $\sigma_A$ is denoted by
$\mathcal S^n(\mathcal H_A,\sigma_A^{\otimes(n-r)})$.
\end{definition}

%%%%%%%%%%%%%%%%%%%%%%%%%%%%%%%%%%%%%%%%%%%%%%%%%%%%%%%%%%%%%%%%%%%%%%%%%%%%%%%%
\section{Choi-almost-iid channels}
\label{sec:def-almost-iid-channel}

In this section, we formally introduce a new definition of almost-iid channels, called \emph{Choi-almost-iid} channels, because it is obtained by extending the MSR notion of almost-iid states to quantum channels via the Choi–Jamiołkowski isomorphism. We then discuss a few properties of this class of almost-iid channels.

\begin{definition}[(Choi-almost-iid channel)]
\label{Def:Choi-almost-iid-channel}
Let $\Phi_{A\to B}\in\mathrm{CPTP}(A,B)$, let $n\geq1$, and let
$0\leq r\leq n$ be an integer. We refer to $\Phi_{A\to B}$ as the
\emph{reference channel}. A channel
$\mathcal N_{A^n\to B^n}\in\mathrm{CPTP}(A^n,B^n)$ is called an
$\binom{n}{r}$-Choi-almost-iid channel along $\Phi_{A\to B}$ if
\begin{equation}
    J(\mathcal N_{A^n\to B^n})
    \in
    \mathcal S^n\!\left(
        \mathcal H_{BA'},
        J(\Phi_{A\to B})^{\otimes(n-r)}
    \right).
\end{equation}
Here $J(\mathcal N_{A^n\to B^n})$, initially written as a state on
$B^nA'^n$, is regarded as a state on $(BA')^{\otimes n}$ by grouping
$B_i$ with $A_i'$ for each $i\in[n]$. Thus,
$\mathcal N_{A^n\to B^n}$ is $\binom{n}{r}$-Choi-almost-iid along $\Phi_{A\to B}$
precisely when its normalized Choi state is an
$\binom{n}{r}$-almost-iid state along $J(\Phi_{A\to B})$.
\end{definition}

\subsection{Properties and an operator characterization of Choi-almost-iid channels}
Every vector in
$\mathcal V(\mathcal H_A^{\otimes n},|\theta\rangle^{\otimes m})$ in \zcref{eq:def-vector-defect-set} has $m$ factors equal to
the reference vector $|\theta\rangle$ and an arbitrary joint vector
on the remaining $n-m$ factors. We introduce an analogous construction at the operator level, which will be useful for characterizing Choi-almost-iid channels. Let $V\in\mathcal{L}(A,C)$ be an isometry, that is $V^{\dagger}V=\id_A$, and let $n\geq1$ and
$0\leq m\leq n$ be integers. Define a subset $\mathcal K$ of $\mathcal{L}(A^n,C^n)$:
\begin{equation}
    \mathcal K(A^n\to C^n,V^{\otimes m})
    \coloneqq
    \left\{
         U^\pi_{C^n}
        \bigl(V^{\otimes m}\otimes L^{(n-m)}\bigr)
        \bigl( U^\pi_{A^n}\bigr)^\dagger
        :
        \pi\in S_n,\;
        L^{(n-m)}
        \in\mathcal L(A^{n-m},C^{n-m})
    \right\}.
\end{equation}
Thus, up to the same permutation of the input and output tensor
factors, each operator in $\mathcal K(A^n\to C^n,V^{\otimes m})$ acts as $V$ on $m$ factors and as an
arbitrary joint operator on the remaining $n-m$ factors. We henceforth refer to
\begin{equation}
\begin{aligned}
    \operatorname{span}
    \mathcal V\!\left(
        \mathcal H_A^{\otimes n},
        |\theta\rangle^{\otimes m}
    \right)
    &\subseteq
    \mathcal H_A^{\otimes n},\\
    \operatorname{span}
    \mathcal K(A^n\to C^n,V^{\otimes m})
    &\subseteq
    \mathcal L(A^n,C^n)
\end{aligned}
\end{equation}
as the \emph{vector defect subspace} and the
\emph{operator defect subspace}, respectively. This terminology
refers to the generating vectors and operators: in each generator,
the reference tensor-power structure may fail on at most $n-m$
positions. A general element of either defect subspace may, however,
be a linear combination involving different choices of defect
locations and therefore need not have a definite set of defective
tensor factors.

For integers $k\geq1$ and
$X\in\mathcal L(A^k,C^k)$, we define its \emph{vectorization} by
\begin{equation}
    |X\rangle\!\rangle_{C^kA'^k}
    \coloneqq
    \bigl(X_{A^k\to C^k}\otimes \id_{A'^k}\bigr)
    |\Omega\rangle_{AA'}^{\otimes k}, \qquad
    |\Omega\rangle_{AA'}
    \coloneqq
    \frac{1}{\sqrt{d_A}}
    \sum_{i=1}^{d_A}|i\rangle_A|i\rangle_{A'}.
\end{equation}
With this normalization, vectorization is an isometric isomorphism
between $\mathcal L(A^k,C^k)$, equipped with the
normalized Hilbert--Schmidt inner product, and
$\mathcal H_{C^kA'^k}$. The following lemma identifies the vector and operator defect subspaces under this isomorphism.

\begin{lemma}[(Isomorphism between operator and vector defect subspaces)]
\label{lem:Defect subspace correspondence}
Let $V\in\mathcal{L}(A,C)$ be an isometry, and let $n\geq1$ and
$0\leq m\leq n$ be integers. Then
\begin{equation}
    \left\{
        |K\rangle\!\rangle_{C^nA'^n}
        :
        K\in\mathcal K(A^n\to C^n,V^{\otimes m})
    \right\}
    =
    \mathcal V\!\left(
        \mathcal H_{CA'}^{\otimes n},
        |V\rangle\!\rangle_{CA'}^{\otimes m}
    \right).
\end{equation}
Consequently, vectorization restricts to an isometric isomorphism
\begin{equation}
    \operatorname{span}
    \mathcal K(A^n\to C^n,V^{\otimes m})
    \simeq
    \operatorname{span}
    \mathcal V\!\left(
        \mathcal H_{CA'}^{\otimes n},
        |V\rangle\!\rangle_{CA'}^{\otimes m}
    \right).
\end{equation}
\end{lemma}

\begin{proof}
Fix $\pi\in S_n$ and $L^{(n-m)}\in
\mathcal L(A^{n-m},C^{n-m})$, and let
\begin{equation}
    K
    =
     U^\pi_{C^n}
    \bigl(V^{\otimes m}\otimes L^{(n-m)}\bigr)
    \bigl( U^\pi_{A^n}\bigr)^\dagger.
\end{equation}
For operators of compatible dimensions, vectorization satisfies
$|UXW^\dagger\rangle\!\rangle
=(U\otimes\overline W)|X\rangle\!\rangle$, where
$\overline W$ denotes entrywise complex conjugation in the basis used
to define vectorization. This follows directly from the transpose trick $(M\otimes \id)|\Omega\rangle
=(\id\otimes M^{\mathsf T})|\Omega\rangle$.

In the corresponding product bases, permutation matrices are real, so
$\overline{ U^\pi_{A^n}}= U^\pi_{A'^n}$. Moreover, vectorization
respects tensor products. Therefore,
\begin{equation}
    |K\rangle\!\rangle
    =
     U^\pi_{(CA')^n}
    \left(
        |V\rangle\!\rangle^{\otimes m}
        \otimes
        |L^{(n-m)}\rangle\!\rangle
    \right)
    \in
    \mathcal V\!\left(
        \mathcal H_{CA'}^{\otimes n},
        |V\rangle\!\rangle^{\otimes m}
    \right),
\end{equation}
where
$ U^\pi_{(CA')^n}
= U^\pi_{C^n}\otimes U^\pi_{A'^n}$
under the prescribed grouping of the tensor factors. This proves one
inclusion.

Conversely, every vector on the right-hand side has the form
\begin{equation}
     U^\pi_{(CA')^n}
    \left(
        |V\rangle\!\rangle^{\otimes m}
        \otimes|\omega^{(n-m)}\rangle
    \right)
\end{equation}
for some $\pi\in S_n$ and some
$|\omega^{(n-m)}\rangle
\in\mathcal H_{CA'}^{\otimes(n-m)}$. Since vectorization is an
isomorphism, there exists a unique
$L^{(n-m)}
\in\mathcal L(A^{n-m},C^{n-m})$
such that
$|L^{(n-m)}\rangle\!\rangle
=|\omega^{(n-m)}\rangle$. Reversing the preceding calculation shows
that the given vector is the vectorization of
\begin{equation}
     U^\pi_{C^n}
    \bigl(V^{\otimes m}\otimes L^{(n-m)}\bigr)
    \bigl( U^\pi_{A^n}\bigr)^\dagger,
\end{equation}
which belongs to
$\mathcal K(A^n\to C^n,V^{\otimes m})$. This proves the reverse
inclusion. The statement concerning the spans follows by linearity.
\end{proof}

The operator defect subspace will be used to give an
operator-level representation of Choi-almost-iid channels. It is
therefore useful to know that this subspace admits an orthonormal basis
whose elements belong to the generating set $\mathcal K$ itself. Here
and below, orthonormality of operators is understood with respect to
the normalized Hilbert--Schmidt inner product defined in \zcref{sec:math-prel}.

\begin{lemma}[(Orthonormal basis of the defect subspace)]
\label{lem:Orthonormal basis of the defect subspace}
Let $V\in\mathcal{L}(A,C)$ be an isometry, let $n\geq1$, and let
$0\leq r\leq n$ be an integer. There exists an orthonormal basis
$\{K_t\}_{t\in\mathcal T}$ of
$\operatorname{span}
\mathcal K(A^n\to C^n,V^{\otimes(n-r)})$
such that
\begin{equation}
    K_t
    \in
    \mathcal K(A^n\to C^n,V^{\otimes(n-r)})
    \quad\forall\,t\in\mathcal T,
    \qquad
    |\mathcal T|
    =
    \sum_{j=0}^r
    \binom{n}{j}(d_Ad_C-1)^j
    \leq
    \binom{n}{r}(d_Ad_C)^r.
\end{equation}
\end{lemma}

\begin{proof}
Set $D\coloneqq d_Ad_C$. Since $V^\dagger V=\id_A$, we have
$\langle V,V\rangle_{\mathrm{nHS}}
=d_A^{-1}\Tr(\id_A)=1$. We may therefore set $F_0\coloneqq V$ and
extend it to an orthonormal basis
$\{F_a\}_{a=0}^{D-1}$ of $\mathcal L(A,C)$. Define
\begin{equation}
    \mathcal B_r
    \coloneqq
    \left\{
        F_{a_1}\otimes\cdots\otimes F_{a_n}
        :
        \bigl|\{i\in[n]:a_i\neq0\}\bigr|\leq r
    \right\}.
\end{equation}
Since the normalized Hilbert--Schmidt inner product is multiplicative
under tensor products, the elements of $\mathcal B_r$ are
orthonormal.

Every element of $\mathcal B_r$ contains at least $n-r$ factors equal
to $F_0=V$. A simultaneous permutation of the input and output tensor
factors can place any chosen $n-r$ of these factors in the first
$n-r$ positions, while the remaining $r$ factors can be incorporated
into $L^{(r)}$. Hence every element of $\mathcal B_r$ belongs to
$\mathcal K(A^n\to C^n,V^{\otimes(n-r)})$.

Conversely, consider
\begin{equation}
    K
    =
     U^\pi_{C^n}
    \bigl(V^{\otimes(n-r)}\otimes L^{(r)}\bigr)
    \bigl( U^\pi_{A^n}\bigr)^\dagger
    \in
    \mathcal K(A^n\to C^n,V^{\otimes(n-r)}).
\end{equation}
Expanding $L^{(r)}$ in the product basis
$\{F_{a_1}\otimes\cdots\otimes F_{a_r}\}$ expresses $K$ as a linear
combination of product operators containing at most $r$ factors
different from $V$. Every such product operator belongs to
$\mathcal B_r$. It follows that
\begin{equation}
    \operatorname{span}\mathcal B_r
    =
    \operatorname{span}
    \mathcal K(A^n\to C^n,V^{\otimes(n-r)}).
\end{equation}
Thus, $\mathcal B_r$ is an orthonormal basis of the operator defect
subspace, and every one of its elements belongs to the generating set
$\mathcal K$.

It remains to estimate its cardinality. For each $j$, one may choose
the $j$ positions whose indices are nonzero in $\binom{n}{j}$ ways and
choose one of $D-1$ nonzero indices at each such position. Hence
\begin{equation}
\begin{aligned}
    |\mathcal B_r|
    &=
    \sum_{j=0}^r\binom{n}{j}(D-1)^j\\
    &\leq
    \binom{n}{r}
    \sum_{j=0}^r\binom{r}{j}(D-1)^j
    =
    \binom{n}{r}D^r.
\end{aligned}
\end{equation}
The inequality follows from
$\binom{n}{j}\binom{n-j}{r-j}
=\binom{n}{r}\binom{r}{j}$
and $\binom{n-j}{r-j}\geq1$. Enumerating the elements of
$\mathcal B_r$ as $\{K_t\}_{t\in\mathcal T}$ completes the proof.
\end{proof}

The preceding lemmas give an operator-level characterization of Choi-almost-iid channels, analogous to the expansion of almost-iid states in~\cite[Remark~2.4(c)]{mazzola2026almost}. It makes the action of these channels on input states explicit and will be used repeatedly in later proofs.

\begin{proposition}[(Operator characterization of Choi-almost-iid channels)]
\label{Prop:alternative-characterization}
Let $\Phi_{A\to B}\in\mathrm{CPTP}(A,B)$, let $n\geq1$ and
$0\leq r\leq n$, and let
$\mathcal N_{A^n\to B^n}\in\mathrm{CPTP}(A^n,B^n)$. Then
$\mathcal N_{A^n\to B^n}$ is
$\binom{n}{r}$-Choi-almost-iid along $\Phi_{A\to B}$ if and only if
there exist a system $E$, a Stinespring isometry
$V_\Phi\in\mathcal{L}(A,BE)$ for $\Phi_{A\to B}$, and an extension channel
\begin{equation}
    {\mathcal N}_{A^n\to B^nE^n}
    \in\mathrm{CPTP}(A^n,B^nE^n)
\end{equation}
of $\mathcal N_{A^n\to B^n}$ satisfying the following conditions:

\begin{enumerate}

\item[(i)]
The channel ${\mathcal N}_{A^n\to B^nE^n}$ is permutation
covariant, i.e.,
\begin{equation}
    {\mathcal N}_{A^n\to B^nE^n}
    \left(
         U^\pi_{A^n}X( U^\pi_{A^n})^\dagger
    \right)
    =
     U^\pi_{ B^nE^n}
    {\mathcal N}_{A^n\to B^nE^n}(X)
    ( U^\pi_{ B^nE^n})^\dagger
\end{equation}
for every $\pi\in S_n$ and $X\in\mathcal L(A^n)$.

\item[(ii)]
There exist a positive semidefinite matrix
$\beta=(\beta_{tt'})_{t,t'\in\mathcal T}$ with unit trace, and an orthonormal basis
$\{K_t\}_{t\in\mathcal T}$ of
\(
    \operatorname{span}
    \mathcal K\big(
        A^n\to B^nE^n,
        V_\Phi^{\otimes(n-r)}
    \big),
\)
satisfying the constraint
\begin{equation}
    \sum_{t,t'\in\mathcal T}
    \beta_{tt'}K_{t'}^\dagger K_t
    =
    \id_{A^n},
\end{equation}
such that $\forall X\in\mathcal L(A^n)$,
\begin{equation}
    {\mathcal N}_{A^n\to B^nE^n}(X)
    =
    \sum_{t,t'\in\mathcal T}
    \beta_{tt'}K_tXK_{t'}^\dagger.
\end{equation}
Moreover, the basis can be chosen with $K_t\in\mathcal K(A^n\to B^nE^n,V_\Phi^{\otimes(n-r)})$ for all $t\in\mathcal{T}$.
\end{enumerate}
\end{proposition}

\begin{proof}
Suppose first that $\mathcal N_{A^n\to B^n}$ is
$\binom{n}{r}$-Choi-almost-iid along $\Phi_{A\to B}$. By
\zcref{Def:Choi-almost-iid-channel}, there exist a purification
$|\theta\rangle_{BEA'}$ of $J(\Phi_{A\to B})$ and a
permutation-invariant state
$\tau_{ B^nE^nA'^n}$ satisfying
\begin{equation}
    \Tr_{E^n}\tau_{ B^nE^nA'^n}
    =
    J(\mathcal N_{A^n\to B^n})
\end{equation}
and
\begin{equation}
    \supp(\tau_{ B^nE^nA'^n})
    \subseteq
    \operatorname{span}
    \mathcal V\!\left(
        \mathcal H_{BEA'}^{\otimes n},
        |\theta\rangle_{BEA'}^{\otimes(n-r)}
    \right).
\end{equation}
Here we have made the fixed within-site reordering
$BA'E\simeq BEA'$.

By the vectorization isomorphism, there exists a unique operator
$V_\Phi\in\mathcal{L}(A,BE)$ such that
$|\theta\rangle=|V_\Phi\rangle\!\rangle$. Since
$|\theta\rangle$ purifies $J(\Phi_{A\to B})$, its marginal on $A'$ is
$\id_{A'}/d_A$. On the other hand,
\begin{equation}
    \Tr_{BE}
    |V_\Phi\rangle\!\rangle
    \langle\!\langle V_\Phi|
    =
    \frac{(V_\Phi^\dagger V_\Phi)^{\mathsf T}}{d_A}.
\end{equation}
It follows that $V_\Phi^\dagger V_\Phi=\id_A$, so $V_\Phi$ is an
isometry. Moreover,
\begin{equation}
    \Tr_E
    |V_\Phi\rangle\!\rangle
    \langle\!\langle V_\Phi|
    =
    J(\Phi_{A\to B}),
\end{equation}
and hence, by injectivity of the Choi--Jamio{\l}kowski correspondence \cite{choi1975completely,jamiolkowski1972linear},
\begin{equation}
    \Phi_{A\to B}(X)
    =
    \Tr_E\!\left[V_\Phi XV_\Phi^\dagger\right].
\end{equation}
Thus, $V_\Phi$ is a Stinespring isometry for $\Phi_{A\to B}$.

Since $\tau_{ B^nE^nA'^n}$ extends
$J(\mathcal N_{A^n\to B^n})$, it has marginal
$\id_{A'^n}/d_A^n$ on $A'^n$. It is therefore the normalized Choi state
of a unique channel
${\mathcal N}_{A^n\to B^nE^n}$. Furthermore,
\begin{equation}
    \Tr_{E^n}J({\mathcal N}_{A^n\to B^nE^n})
    =
    J(\mathcal N_{A^n\to B^n}),
\end{equation}
so ${\mathcal N}_{A^n\to B^nE^n}$ is an extension of $\mathcal N_{A^n\to B^n}$.

The permutation invariance of $\tau$ is invariance under
$ U^\pi_{ B^nE^n}\otimes U^\pi_{A'^n}$. Since permutation matrices
are real in the chosen product bases, this invariance is equivalent,
under the Choi--Jamio{\l}kowski correspondence, to
\begin{equation}
    {\mathcal N}_{A^n\to B^nE^n}
    \left(
         U^\pi_{A^n}X( U^\pi_{A^n})^\dagger
    \right)
    =
     U^\pi_{ B^nE^n}
    {\mathcal N}_{A^n\to B^nE^n}(X)
    ( U^\pi_{ B^nE^n})^\dagger.
\end{equation}
This proves~(i).

By \zcref{lem:Defect subspace correspondence} and
$|\theta\rangle=|V_\Phi\rangle\!\rangle$, the support condition becomes
\begin{equation}
    \supp(\tau)
    \subseteq
    \operatorname{span}
    \left\{
        |K\rangle\!\rangle
        :
        K\in
        \mathcal K\!\left(
            A^n\to B^nE^n,
            V_\Phi^{\otimes(n-r)}
        \right)
    \right\}.
\end{equation}
Let $\{K_t\}_{t\in\mathcal T}$ be the orthonormal basis supplied by
\zcref{lem:Orthonormal basis of the defect subspace}. Since
vectorization is an isometry,
$\{|K_t\rangle\!\rangle\}_{t\in\mathcal T}$ is an orthonormal basis
of the corresponding vector defect subspace. Consequently,
\begin{equation}
    \tau
    =
    \sum_{t,t'\in\mathcal T}
    \beta_{tt'}
    |K_t\rangle\!\rangle
    \langle\!\langle K_{t'}|
\end{equation}
for some positive semidefinite matrix
$\beta=(\beta_{tt'})_{t,t'\in\mathcal T}$. Since $\tau$ is normalized
and the vectors $|K_t\rangle\!\rangle$ are orthonormal,
$\operatorname{Tr}\beta=1$.

For every $X\in\mathcal L(A^n)$, normalized vectorization gives
\begin{equation}
    d_A^n\Tr_{A'^n}
    \left[
        (\id_{B^nE^n}\otimes X^{\mathsf T})
        |K_t\rangle\!\rangle
        \langle\!\langle K_{t'}|
    \right]
    =
    K_tXK_{t'}^\dagger.
\end{equation}
Applying the inverse Choi--Jamio{\l}kowski correspondence therefore
yields
\begin{equation}
    {\mathcal N}_{A^n\to B^nE^n}(X)
    =
    \sum_{t,t'\in\mathcal T}
    \beta_{tt'}K_tXK_{t'}^\dagger.
\end{equation}
Since ${\mathcal N}_{A^n\to B^nE^n}$ is trace preserving, the coefficients necessarily satisfy
\begin{equation}
    \sum_{t,t'\in\mathcal T}
    \beta_{tt'}K_{t'}^\dagger K_t
    =
    \id_{A^n}.
\end{equation}
This proves~(ii).

Conversely, suppose that $V_\Phi$ and
${\mathcal N}_{A^n\to B^nE^n}$ satisfying~(i)--(ii) exist.
Since $V_\Phi$ is a Stinespring isometry for $\Phi$,
$|V_\Phi\rangle\!\rangle_{BEA'}$ is a purification of
$J(\Phi_{A\to B})$.

Let $\tau\coloneqq J({\mathcal N}_{A^n\to B^nE^n})$. Since
$ \mathcal N_{A^n\to B^nE^n}$ extends $\mathcal N_{A^n\to B^n}$, we have
$\Tr_{E^n}\tau=J(\mathcal N_{A^n\to B^n})$. Condition~(ii) and the
Choi--Jamio{\l}kowski correspondence give
\begin{equation}
    \tau
    =
    \sum_{t,t'\in\mathcal T}
    \beta_{tt'}
    |K_t\rangle\!\rangle
    \langle\!\langle K_{t'}|.
\end{equation}
It follows from
\zcref{lem:Defect subspace correspondence} that
\begin{equation}
    \supp(\tau)
    \subseteq
    \operatorname{span}
    \mathcal V\!\left(
        \mathcal H_{BEA'}^{\otimes n},
        |V_\Phi\rangle\!\rangle_{BEA'}^{\otimes(n-r)}
    \right).
\end{equation}
By~(i), $\tau$ is invariant under simultaneous permutations of the
$n$ composite systems $BEA'$. Thus, $\tau$ is a
permutation-invariant extension of $J(\mathcal N)$ satisfying the
required support condition. Hence
$J(\mathcal N_{A^n\to B^n})$ is an
$\binom{n}{r}$-almost-iid state along $J(\Phi_{A\to B})$, and therefore
$\mathcal N_{A^n\to B^n}$ is
$\binom{n}{r}$-Choi-almost-iid along $\Phi_{A\to B}$.
\end{proof}

In \zcref{Prop:alternative-characterization}, we only required the existence of a Stinespring dilation of the reference channel. The following corollary, whose proof can be found in \zcref{app:proof-of-fixed-Stinespring-representation}, shows that we can in fact pick \emph{any} minimal Stinespring isometry and choose the extension channel and decomposition using this isometry. This is analogous to~\cite[Remark~2.4(a)]{mazzola2026almost}.

\begin{corollary}[(Fixed Stinespring representation)]
\label{cor:fixed-Stinespring-representation}
Let $\Phi_{A\to B}\in\mathrm{CPTP}(A,B)$, and fix an arbitrary minimal Stinespring isometry
\begin{equation}
    V_\Phi\in\mathcal{L}(A,BE)
\end{equation}
for $\Phi$. If
$\mathcal N_{A^n\to B^n}$ is an
$\binom{n}{r}$-Choi-almost-iid channel along $\Phi$, then the
extension channel and the operator decomposition in
\zcref{Prop:alternative-characterization} may be chosen using this
fixed environment $E$ and this fixed Stinespring isometry $V_\Phi$.
\end{corollary}

The operator decomposition in \zcref{Prop:alternative-characterization} contains cross terms \(K_t X K_{t'}^\dagger\) with \(t\neq t'\), so it is \emph{not} in the standard Kraus form. The following remark makes explicit the connection with the standard Kraus decomposition.

{\begin{remark}[(Kraus representation)]
\label{Rem:Kraus-defect}
Let $ \mathcal N_{A^n\to B^nE^n}$ be an extension channel
provided by \zcref{Prop:alternative-characterization}. Its normalized
Choi state has an eigenvalue decomposition
\begin{equation}
    J( \mathcal N_{A^n\to B^nE^n})
    =
    \sum_{\mu:\lambda_\mu>0}
    \lambda_\mu
    |\psi_\mu\rangle\langle\psi_\mu|.
\end{equation}
By the support condition in the proof of
\zcref{Prop:alternative-characterization}, every eigenvector
$|\psi_\mu\rangle$ with $\lambda_\mu>0$ belongs to
\begin{equation}
    \operatorname{span}
    \left\{
        |K\rangle\!\rangle
        :
        K\in
        \mathcal K\!\left(
            A^n\to B^nE^n,
            V_\Phi^{\otimes(n-r)}
        \right)
    \right\}.
\end{equation}
The vectorization isomorphism therefore gives a unique operator
\begin{equation}
    L_\mu
    \in
    \operatorname{span}
    \mathcal K\!\left(
        A^n\to B^nE^n,
        V_\Phi^{\otimes(n-r)}
    \right)
\end{equation}
such that $|\psi_\mu\rangle=|L_\mu\rangle\!\rangle$. Defining
$A_\mu\coloneqq\sqrt{\lambda_\mu}\,L_\mu$, we obtain
\begin{equation}
    J( \mathcal N_{A^n\to B^nE^n})
    =
    \sum_{\mu:\lambda_\mu>0}
    |A_\mu\rangle\!\rangle
    \langle\!\langle A_\mu|.
\end{equation}
The inverse Choi--Jamio{\l}kowski correspondence then gives the Kraus
representation
\begin{equation}
     \mathcal N_{A^n\to B^nE^n}(X)
    =
    \sum_{\mu:\lambda_\mu>0}
    A_\mu X A_\mu^\dagger,
    \qquad
    \sum_{\mu:\lambda_\mu>0}
    A_\mu^\dagger A_\mu
    =
    \id_{A^n},
\end{equation}
where every Kraus operator $A_\mu$ belongs to the operator defect
subspace
\begin{equation}
    \operatorname{span}
    \mathcal K\!\left(
        A^n\to B^nE^n,
        V_\Phi^{\otimes(n-r)}
    \right).
\end{equation}
It need not, however, belong to the generating set $\mathcal K$
itself: it is in general a linear combination of operators with different defect locations and different defect operators. Consequently, Choi-almost-iid-ness does not imply that every Kraus branch has a definite set of at most $r$ defective positions. This is the direct analogy of \emph{superposition} of defects in MSR almost-iid states.

Equivalently, one may obtain these Kraus operators by diagonalizing
the coefficient matrix $\beta$ in
\zcref{Prop:alternative-characterization}. Relative to the
defect-structured basis $\{K_t\}_{t\in\mathcal T}$, the off-diagonal
entries of $\beta$ describe coherences between different basis
elements. This interpretation is basis dependent: diagonalizing
$\beta$ removes the off-diagonal entries but generally replaces the
individual operators $K_t\in\mathcal K$ by linear combinations
belonging only to $\operatorname{span}\mathcal K$.
\end{remark}}

We conclude this part with a simple consistency observation: a Choi-almost-iid channel maps
tensor-power input states to MSR-almost-iid output states.

\begin{proposition}
\label{Prop:iid-input-almost-iid-output}
Let $\mathcal N_{A^n\to B^n}$ be an $\binom nr$-Choi-almost-iid
channel along $\Phi_{A\to B}$. Then, for every
$\rho_A\in\mathcal D(A)$,
\begin{equation}
    \mathcal N_{A^n\to B^n}(\rho_A^{\otimes n})
    \in
    \mathcal S^n\!\left(
        \mathcal H_B,
        \Phi(\rho_A)^{\otimes(n-r)}
    \right).
\end{equation}
In other words, an iid input is mapped to an
$\binom nr$-almost-iid output along the reference output.
\end{proposition}

\begin{proof}
By~\zcref{Prop:alternative-characterization}, there exist an isometric extension
$V_\Phi\in\mathcal{L}(A,BE)$ of $\Phi$ and a permutation-covariant extension channel
$\mathcal N_{A^n\to B^nE^n}$ such that
\begin{equation}
    \mathcal N_{A^n\to B^nE^n}(\cdot)
    =
    \sum_{t,t'\in\mathcal T}
    \beta_{tt'}K_t(\cdot)K_{t'}^\dagger ,
    \qquad
    K_t\in
    \mathcal K\!\left(
        A^n\to B^nE^n,
        V_\Phi^{\otimes(n-r)}
    \right).
\end{equation}
Let $|\psi\rangle_{AR}$ be a purification of $\rho_A$. Define
\begin{equation}
    |\theta\rangle_{BER}
    \coloneqq 
    (V_\Phi\otimes\id_R)|\psi\rangle_{AR}.
\end{equation}
Since $V_\Phi$ is an isometric extension of $\Phi$,
$|\theta\rangle_{BER}$ is a purification of $\Phi(\rho_A)$.

Consider the state
\begin{equation}
    \omega_{B^nE^nR^n}
    \coloneqq 
    \bigl(
        \mathcal N_{A^n\to B^nE^n}\otimes\idc_{R^n}
    \bigr)
    \left(
        |\psi\rangle\langle\psi|_{AR}^{\otimes n}
    \right).
\end{equation}
It is an extension of
$\mathcal N_{A^n\to B^n}(\rho_A^{\otimes n})$. Moreover,
$|\psi\rangle^{\otimes n}$ is invariant under simultaneous
permutations of the pairs $A_iR_i$, and permutation covariance of the
extension channel therefore implies that $\omega_{B^nE^nR^n}$ is
permutation invariant under the triples $B_iE_iR_i$.

It remains to verify the support condition in~\zcref{def:MSR-state}.
For each $t\in\mathcal T$, there exist a permutation
$\pi_t\in S_n$ and an operator
\begin{equation}
    L_t^{(r)}
    \in
    \mathcal L(A^r,(BE)^r)
\end{equation}
such that
\begin{equation}
    K_t
    =
    U_{B^nE^n}^{\pi_t}
    \bigl(
        V_\Phi^{\otimes(n-r)}
        \otimes
        L_t^{(r)}
    \bigr)
    \bigl(
        U_{A^n}^{\pi_t}
    \bigr)^\dagger .
\end{equation}
Define
\begin{equation}
    |\Psi_t\rangle_{B^nE^nR^n}
    \coloneqq
    \bigl(
        K_t\otimes\id_{R^n}
    \bigr)
    |\psi\rangle_{AR}^{\otimes n}.
\end{equation}
These vectors need not be normalized, and
\begin{equation}
    \omega_{B^nE^nR^n}
    =
    \sum_{t,t'\in\mathcal T}
    \beta_{tt'}
    |\Psi_t\rangle
    \langle\Psi_{t'}|.
\end{equation}
Since $|\psi\rangle_{AR}^{\otimes n}$ is invariant under
simultaneous permutations of the $AR$ pairs,
\begin{equation}
    \left(
        \bigl(U_{A^n}^{\pi_t}\bigr)^\dagger
        \otimes
        \id_{R^n}
    \right)
    |\psi\rangle_{AR}^{\otimes n}
    =
    \left(
        \id_{A^n}
        \otimes
        U_{R^n}^{\pi_t}
    \right)
    |\psi\rangle_{AR}^{\otimes n}.
\end{equation}
Therefore,
\begin{equation}
\begin{aligned}
    |\Psi_t\rangle
    &=
    \bigl(
        K_t\otimes\id_{R^n}
    \bigr)
    |\psi\rangle_{AR}^{\otimes n}
    \\
    &=
    \left(
        U_{B^nE^n}^{\pi_t}
        \otimes
        \id_{R^n}
    \right)
    \bigl(
        V_\Phi^{\otimes(n-r)}
        \otimes
        L_t^{(r)}
        \otimes
        \id_{R^n}
    \bigr)
    \left(
        \bigl(U_{A^n}^{\pi_t}\bigr)^\dagger
        \otimes
        \id_{R^n}
    \right)
    |\psi\rangle_{AR}^{\otimes n}
    \\
    &=
    \left(
        U_{B^nE^n}^{\pi_t}
        \otimes
        \id_{R^n}
    \right)
    \bigl(
        V_\Phi^{\otimes(n-r)}
        \otimes
        L_t^{(r)}
        \otimes
        \id_{R^n}
    \bigr)
    \left(
        \id_{A^n}
        \otimes
        U_{R^n}^{\pi_t}
    \right)
    |\psi\rangle_{AR}^{\otimes n}
    \\
    &=
    \left(
        U_{B^nE^n}^{\pi_t}
        \otimes
        U_{R^n}^{\pi_t}
    \right)
    \bigl(
        V_\Phi^{\otimes(n-r)}
        \otimes
        L_t^{(r)}
        \otimes
        \id_{R^n}
    \bigr)
    |\psi\rangle_{AR}^{\otimes n}
    \\
    &=
    U_{(BER)^n}^{\pi_t}
    \left(
        |\theta\rangle_{BER}^{\otimes(n-r)}
        \otimes
        |\omega_t^{(r)}\rangle
    \right),
\end{aligned}
\end{equation}
where
\begin{equation}
    |\omega_t^{(r)}\rangle
    \coloneqq
    \bigl(
        L_t^{(r)}
        \otimes
        \id_{R^r}
    \bigr)
    |\psi\rangle_{AR}^{\otimes r}.
\end{equation}
It follows that
\begin{equation}
    |\Psi_t\rangle
    \in
    \mathcal V\!\left(
        \mathcal H_{BER}^{\otimes n},
        |\theta\rangle_{BER}^{\otimes(n-r)}
    \right)
    \qquad
    \text{for every }t\in\mathcal T,
\end{equation}
which implies
\begin{equation}
    \supp\!\left(
        \omega_{B^nE^nR^n}
    \right)
    \subseteq
    \operatorname{span}
    \mathcal V\!\left(
        \mathcal H_{BER}^{\otimes n},
        |\theta\rangle_{BER}^{\otimes(n-r)}
    \right).
\end{equation}
Thus, $\omega$ satisfies both conditions
in~\zcref{def:MSR-state}, proving the claim.
\end{proof}

\subsection{Examples of Choi-almost-iid channels}

We provide simple examples of Choi-almost-iid channels, which can be viewed as the direct channel analogues of the
almost-iid-state examples in~\cite[Example 2.2]{mazzola2026almost}: tensor-power
states, convex mixtures with a small number of defects, and coherent
superpositions of defect positions.

\begin{example*}[(Choi-almost-iid channels)]
\leavevmode\par
\begin{enumerate}
\item[(a)] \emph{Tensor-power channels.}
$\mathcal N_{A^n\to B^n}
=
{\Phi}_{A\to B}^{\otimes n}$
is $\binom{n}{0}$-Choi-almost-iid\ along
${\Phi}_{A\to B}$.

\item[(b)] \emph{Convex mixtures with a fixed number of defects.}
Let $\mathcal D^{(r)}_{A^r\to B^r}$ be an arbitrary channel representing the defects and define
\begin{equation}
    \mathcal N_{A^n\to B^n}
    \coloneqq 
    \frac{1}{n!}\sum_{\pi\in S_n}
    \mathcal U^{\pi}_{B^n}
    \circ
    \bigl(
        \Phi_{A\to B}^{\otimes(n-r)}\otimes\mathcal D^{(r)}_{A^r\to B^r}
    \bigr)
    \circ
    \mathcal U^{\pi,\dagger}_{A^n},
    \label{Eq:convex-defect-channel}
\end{equation}
where $\mathcal U^\pi_{A^n}(\cdot)\coloneqq U^\pi_{A^n}(\cdot)U^{\pi,\dagger}_{A^n}$. Then $\mathcal N_{A^n\to B^n}$ is
$\binom nr$-Choi-almost-iid along $\Phi_{A\to B}$.

To see this, choose $\mathcal H_E$ large enough so that
$\Phi$ admits an isometric extension $V_\Phi\in\mathcal{L}(A,BE)$ and
$\mathcal D^{(r)}$ admits an isometric extension
$W^{(r)}:A^r\to B^rE^r$. Define
\begin{equation}
    K_\pi
    \coloneqq 
    U^{\pi}_{ B^nE^n}
    \bigl(
        V_\Phi^{\otimes(n-r)}\otimes W^{(r)}
    \bigr)
    (U^{\pi}_{A^n})^\dagger ,
    \qquad \pi\in S_n,
\end{equation}
and the extension channel
\begin{equation}
    \mathcal N_{A^n\to B^nE^n}(\rho_{A^n})
    \coloneqq 
    \frac{1}{n!}\sum_{\pi\in S_n}
    K_\pi\rho_{A^n} K_\pi^\dagger,
\end{equation}
which is manifestly permutation covariant. Each $K_\pi$ is an isometry, and tracing out $E^n$ gives~\zcref{Eq:convex-defect-channel}. By construction, $K_\pi\in
    \mathcal K\big(
        A^n\to B^nE^n,
        V_\Phi^{\otimes(n-r)}
    \big)$, and hence by~\zcref{lem:Defect subspace correspondence},
\begin{equation}
    |K_\pi\rangle\!\rangle
    \in
    \mathcal V\!\left(
        \mathcal H_{BEA'}^{\otimes n},
        |V_\Phi\rangle\!\rangle_{BEA'}^{\otimes(n-r)}
    \right)
    \qquad
    \forall\,\pi\in S_n.
\end{equation}
Thus $J(\mathcal N_{A^n\to B^nE^n})$ is a permutation-invariant
extension whose support is contained in the vector defect subspace.

\item[(c)] \emph{Coherent superposition of defects.}
Let $A\simeq\mathbb C^2$, $E\simeq\mathbb C^2$, and
define isometries $V,W:A\to A\otimes E$ by $V|\psi\rangle=|\psi\rangle|0\rangle$ and $W|\psi\rangle=|\psi\rangle|1\rangle$ for any $|\psi\rangle\in A$. Define
\begin{equation}
    K
    \coloneqq 
    \frac{1}{\sqrt2}
    \bigl(
        W\otimes V-V\otimes W
    \bigr)
    :A^2\to (A\otimes E)^2 .
\end{equation}
Then $K(\cdot)K^\dagger$ is $\binom21$-Choi-almost-iid along $V(\cdot)V^\dagger$.

To see this, notice $V^\dagger W=W^\dagger V=0$ so $K^\dagger K=\mathbbm 1_{A^2}$ is an isometry. Moreover $K(\cdot)K^\dagger$ is permutation covariant, and obviously $K\in
    \operatorname{span}\mathcal K(A^2\to (A\otimes E)^2,V)$.
\end{enumerate}
\end{example*}

\subsection{Wasserstein distance between outputs of Choi-almost-iid and iid channels}
\label{sec:Wasserstein}

For MSR almost-iid state sequences $(\rho^{(n)})_{n\geq1}$ along $\sigma$ with defect sizes \(r_n=o(n)\), the quantum Wasserstein distance of order $1$ \cite{de2021quantum} from the corresponding tensor-power states is \(o(n)\) \cite{mazzola2026robust,girardi2026new}:
\begin{equation}
    \frac{1}{n}
    \left\|
        \rho^{(n)}-\sigma^{\otimes n}
    \right\|_{W_1}\longrightarrow 0.
\end{equation}
It is natural to ask whether an analogous statement holds for
channel sequences: if
\(\bm{\mathcal N}=(\mathcal N^{(n)}_{A^n\to B^n})_{n\geq1}\)
is Choi-almost-iid along \(\Phi_{A\to B}\), with \(r_n=o(n)\), does
\begin{equation}
    \frac{1}{n}
    \bigl\|\mathcal N^{(n)}-\Phi^{\otimes n}\bigr\|_{\clubsuit}
    \longrightarrow 0?
\end{equation}
By~\zcref{eq:club}, this asks whether the output \(W_1\) distance
vanishes after normalization, uniformly over all input states
\(\rho_{A^n}\), when both channels act on the same input.

For arbitrary, possibly highly entangled inputs, we do not obtain
such a vanishing bound. However, we do obtain one when the inputs have limited
{\em{entanglement depth}}. An \(n\)-partite state has entanglement depth at
most \(\ell\) if it is a convex combination of pure states, each of
which factorizes across a partition into groups of at most \(\ell\)
subsystems~\cite{sorensen2001entanglement}. The following proposition
bounds the normalized output \(W_1\) distance uniformly over these
inputs by \(4\sqrt{\ell r/n}\). Applied at each blocklength with
\(\ell=\ell_n\) and \(r=r_n\), this bound vanishes whenever
\(\ell_n r_n=o(n)\).

\begin{proposition}[(Output Wasserstein distance for bounded entanglement depth)]
\label{Prop:output-Wasserstein-depth}
Let \(n\geq1\), \(0\leq r\leq n\), and \(1\leq\ell\leq n\)
be integers. Let \(\mathrm{SEP}_\ell(A_1:\cdots:A_n)\) denote the
set of states of entanglement depth at most \(\ell\): the convex
hull of pure states \(\psi_{A^n}\) for which there exist a partition
\([n]=I_1\mathbin{\sqcup}\cdots\mathbin{\sqcup}I_M\), with
\(|I_j|\leq\ell\) for every \(j\), and pure states
\(|\psi_j\rangle_{A_{I_j}}\) such that
\begin{equation}
    |\psi\rangle_{A^n}
    =
    \bigotimes_{j=1}^{M}|\psi_j\rangle_{A_{I_j}}.
\end{equation}

If \(\mathcal N_{A^n\to B^n}\) is an
\(\binom nr\)-Choi-almost-iid channel along \(\Phi_{A\to B}\), then
\begin{equation}
    \frac{1}{n}
    \sup_{\rho_{A^n}\in\mathrm{SEP}_\ell(A_1:\cdots:A_n)}
    \left\|
        \mathcal N_{A^n\to B^n}(\rho_{A^n})
        -
        \Phi_{A\to B}^{\otimes n}(\rho_{A^n})
    \right\|_{W_1}
    \leq
    4\sqrt{\frac{\ell r}{n}},
    \label{Eq:output-Wasserstein-depth}
\end{equation}
where \(\|\cdot\|_{W_1}\) is defined with respect to the output systems
\(B_1,\ldots,B_n\).
\end{proposition}

\begin{proof}
    See~\zcref{app:proof-of-Wasserstein-Prop}.
\end{proof}

%%%%%%%%%%%%%%%%%%%%%%%%%%%%%%%%%%%%%%%%%%%%%%%%%%%%%%%%%%%%%%%%%%%%%%%%%%%%%%%%
{%%%%%%%%%%%%%%%%%%%%%%%%%%%%%%%%%%%%%%%%%%%%%%%%%%%%%%%%%%%%%%%%%%%%%%%%%%%%%%%%
\section{Robustness of classical communication}
\label{Sec:robust-classical-communication}

We first study classical communication over Choi-almost-iid channels.
By \emph{robustness}, we mean that {the classical capacity of the memoryless reference channel is retained}: every rate strictly below
$C(\Phi)$ can be achieved by a single sequence of codes for every
Choi-almost-iid channel sequence along $\Phi$. The codes may depend on
the reference channel $\Phi$ and on the target rate, but not on the
particular channel sequence or on any choice of witness in the
almost-iid definition. We begin by fixing the coding conventions for a
general, not necessarily memoryless, sequence of channels. Throughout
this section, logarithms are taken to be base two.

\begin{definition}[(Classical communication codes and capacity)]
\label{Def:classical-communication}
Let
\(
    \mathbf N
    \coloneqq
    \{\mathcal N^{(n)}_{A^n\to B^n}\}_{n\in\mathbb N}
\)
be a sequence of quantum channels. For each $n\in\mathbb N$, an
$(n,M_n)$ classical communication code
\(
    \mathcal C_n
\)
consists of a message set $\mathcal M_n$ with
$|\mathcal M_n|=M_n$, encoding states
\(
    \{\rho_m^{(n)}\}_{m\in\mathcal M_n}
    \subseteq\mathcal D(A^n)
\),
and a decoding POVM
\(
    \{D_m^{(n)}\}_{m\in\mathcal M_n}
\)
on $B^n$. Thus,
\(
    D_m^{(n)}\geq0
\)
for every $m\in\mathcal M_n$ and
\(
    \sum_{m\in\mathcal M_n}D_m^{(n)}=\id_{B^n}.
\)

The maximal error probability of $\mathcal C_n$ when used over
$\mathcal N^{(n)}$ is
\begin{equation}
    \varepsilon_{\max}
    \bigl(\mathcal C_n,\mathcal N^{(n)}\bigr)
    \coloneqq
    \max_{m\in\mathcal M_n}
    \left\{
        1-
        \Tr\!\left[
            D_m^{(n)}
            \mathcal N^{(n)}\!\left(\rho_m^{(n)}\right)
        \right]
    \right\}.
\end{equation}

A rate $R\geq0$ is \emph{achievable} for $\mathbf N$ if there exists a
sequence of codes $\{\mathcal C_n\}_{n\in\mathbb N}$ such that
\begin{equation}
    \lim_{n\to\infty}
    \varepsilon_{\max}
    \bigl(\mathcal C_n,\mathcal N^{(n)}\bigr)
    =0,
    \qquad
    \liminf_{n\to\infty}\frac{1}{n}\log M_n
    \geq R.
\end{equation}
The classical capacity of $\mathbf N$ is
\begin{equation}
    C(\mathbf N)
    \coloneqq
    \sup\left\{
        R:
        R\text{ is achievable for }\mathbf N
    \right\}.
\end{equation}
\end{definition}

For the memoryless sequence
\(
    \{\Phi_{A\to B}^{\otimes n}\}_{n\in\mathbb N},
\)
this definition reduces to the usual unassisted classical capacity
$C(\Phi)$. By the Holevo--Schumacher--Westmoreland theorem
and regularization over blocks~\cite{holevo1998capacity,
schumacher1997sending},
\begin{equation}
    C(\Phi)
    =
    \lim_{k\to\infty}
    \frac{1}{k}\chi\!\left(\Phi^{\otimes k}\right)
    =
    \sup_{k\geq1}
    \frac{1}{k}\chi\!\left(\Phi^{\otimes k}\right),
    \label{Eq:HSW-capacity}
\end{equation}
where, for a channel $\Psi_{A\to B}$,
\begin{equation}
    \chi(\Psi)
    \coloneqq
    \max_{\{p_x,\rho_x\}}
    I(X:B)_\sigma,
    \qquad
    \sigma_{XB}
    \coloneqq
    \sum_x
    p_x|x\rangle\langle x|_X
    \otimes\Psi(\rho_x).
\end{equation}
The maximum is over finite ensembles
$\{p_x,\rho_x\}$ of states on $A$. The maximal-error convention used
above gives the same capacity formula as the usual average-error
formulation, by expurgation.

Let $\mathfrak S_\Phi$ denote the class of all channel sequences
\(
    \mathbf N
    =
    \{\mathcal N^{(n)}_{A^n\to B^n}\}_{n\in\mathbb N}
\)
for which there exists a sequence of integers
$r_n\in\{0,\ldots,n\}$ satisfying
\(
    \lim_{n\to\infty}r_n/n=0
\)
such that, for every $n$, the channel $\mathcal N^{(n)}$ is an
$\binom{n}{r_n}$-Choi-almost-iid channel along $\Phi$.

\begin{definition}[(Universal classical capacity of $\mathfrak S_\Phi$)]
\label{Def:universal-classical-capacity}
A rate $R\geq0$ is \emph{universally achievable} for
$\mathfrak S_\Phi$ if there exists a sequence of classical
communication codes
\(
    \{\mathcal C_n\}_{n\in\mathbb N},
\)
chosen independently of the particular sequence
$\mathbf N\in\mathfrak S_\Phi$, such that
\begin{equation}
    \liminf_{n\to\infty}\frac{1}{n}\log M_n
    \geq R
\end{equation}
and, for every $\mathbf N\in\mathfrak S_\Phi$,
\begin{equation}
    \lim_{n\to\infty}
    \varepsilon_{\max}
    \bigl(\mathcal C_n,\mathcal N^{(n)}\bigr)
    =0.
\end{equation}
Such a sequence of codes is called a \emph{universally reliable
classical communication scheme} for $\mathfrak S_\Phi$.

The \emph{universal classical capacity} of $\mathfrak S_\Phi$ is
\begin{equation}
    C_{\mathrm{univ}}(\mathfrak S_\Phi)
    \coloneqq
    \sup\left\{
        R:
        R\text{ is universally achievable for }\mathfrak S_\Phi
    \right\}.
    \label{Eq:universal-classical-capacity-def}
\end{equation}
\end{definition}

The order of quantifiers in the preceding definition is important:
the sequence of codes is common to all
$\mathbf N\in\mathfrak S_\Phi$, whereas the rate at which the error
converges to zero may depend on the particular channel sequence.
Since the memoryless sequence
$\{\Phi^{\otimes n}\}_{n\in\mathbb N}$ belongs to
$\mathfrak S_\Phi$ by taking $r_n=0$, one immediately has
\begin{equation}
    C_{\mathrm{univ}}(\mathfrak S_\Phi)
    \leq C(\Phi).
\end{equation}

{Every universally achievable rate for $\mathfrak S_\Phi$ is achievable
for each of its individual members. Consequently,
\begin{equation}
    C_{\mathrm{univ}}(\mathfrak S_\Phi)
    \leq
    \inf_{\mathbf N\in\mathfrak S_\Phi}C(\mathbf N).
\end{equation}
For a general class of channel sequences, this inequality can be
strict, since the codes achieving the individual capacities may depend
on the particular channel sequence. The robustness theorem below will show that this inequality is in fact
an equality for $\mathfrak S_\Phi$, and that both sides are equal to
$C(\Phi)$.

The proof is based on concatenating a finite-blocklength transmission
code for the reference channel $\Phi$ with an outer classical
error-correcting code. The outer code absorbs both ordinary decoding
errors of the inner transmission code and errors caused by the
deviation from the reference tensor-power channel.

More precisely, we first choose a sufficiently large but fixed
blocklength $k$ and an inner code for $\Phi^{\otimes k}$ with message
alphabet $[L]$, pure encoding states
$\{|\psi_a\rangle\}_{a\in[L]}$ on $A^k$, and a decoding POVM
$\{D_a\}_{a\in[L]}$ on $B^k$, such that
\begin{equation}
    \Tr\!\left[
        D_a\,
        \Phi^{\otimes k}
        \bigl(|\psi_a\rangle\langle\psi_a|\bigr)
    \right]
    \geq 1-\eta
    \qquad\forall\,a\in[L].
\end{equation}
Pure encoding states may be assumed without loss of generality: for
each mixed encoding state, one may choose a pure state in one of its
decompositions whose decoding success probability is at least the
corresponding average.

We then apply this inner code independently to $m$ blocks. Consider
one of the defect-structured operators $K_t$ appearing in the
non-diagonal decomposition supplied by
\zcref{Prop:alternative-characterization}. After grouping the channel
uses into blocks of size $k$, at most $r_n$ blocks contain a position
on which $K_t$ does not act as the reference Stinespring isometry.
On every remaining block, its action is
$V_\Phi^{\otimes k}$. Thus, when $K_t$ is applied to a tensor product
of pure inner-code states, the resulting vector agrees with the ideal
Stinespring output on all but at most $r_n$ blocks, {while its joint output on the remaining blocks may be arbitrary}.

Because the decomposition in
\zcref{Prop:alternative-characterization} need not be diagonal, the
extended output state may contain coherent cross terms between
different defect operators and different defect locations. We control
these terms by applying the pinching inequality in a product basis
adapted to the ideal block outputs. The resulting multiplicative
factor is bounded by the dimension of the corresponding block-defect
subspace. For fixed $k$ and $r_n=o(n)$, the logarithm of this dimension
is sublinear in the number of blocks, so the pinching factor is
subexponential and does not affect the achievable communication rate.

Finally, we choose an outer classical code
$\mathcal C_m^{\rm{out}}\subseteq[L]^m$ with minimum Hamming distance at least
$2t_m+1$. Thus, any two codewords differ in at least $2t_m+1$
positions, and nearest-neighbour decoding recovers the transmitted
codeword whenever at most $t_m$ inner blocks are decoded incorrectly.
The Gilbert--Varshamov bound~\cite{gilbert1952comparison,varshamov1957estimate} guarantees the existence of such codes
with a suitable asymptotic rate.

We first establish the following estimate on the probability that
more than $t$ of the block tests fail.

\begin{proposition}[(Tail bound for product tests on a defect subspace)]
\label{Prop:classical-outer}
Let $\mathcal H_1,\ldots,\mathcal H_m$ be Hilbert spaces of the same
dimension $d$, and let
$|\theta_j\rangle\in\mathcal H_j$ be a unit vector for every
$j\in[m]$. For each $j\in[m]$, let $T_j$ be an effect on
$\mathcal H_j$ satisfying
\begin{equation}
    0\leq T_j\leq \id_j,
    \qquad
    \langle\theta_j|T_j|\theta_j\rangle
    \geq 1-\eta,
\end{equation}
where $\id_j$ denotes the identity operator on $\mathcal H_j$ and
$0\leq\eta\leq1$.

The product of the binary POVMs
$\{T_j,\id_j-T_j\}$ is a POVM
$\{R_S\}_{S\subseteq[m]}$ on
$\bigotimes_{j=1}^m\mathcal H_j$, where
\begin{equation}
    R_S
    \coloneqq
    \bigotimes_{j=1}^m R_{j,S},
    \qquad
    R_{j,S}
    \coloneqq
    \begin{cases}
        \id_j-T_j, & j\in S,\\
        T_j,     & j\notin S.
    \end{cases}
\end{equation}
Thus, $S$ records the blocks on which the corresponding test fails.

Fix integers $0\leq r\leq t\leq m$. For
$F\subseteq[m]$, write
$\mathcal H_F\coloneqq\bigotimes_{j\in F}\mathcal H_j$ and
$|\theta\rangle_{F^c}
\coloneqq\bigotimes_{j\notin F}|\theta_j\rangle$. Define
\begin{equation}
    \mathcal W
    \coloneqq
    \operatorname{span}
    \left\{
        |\theta\rangle_{F^c}\otimes|\omega_F\rangle
        :
        F\subseteq[m],\ |F|\leq r,\ 
        |\omega_F\rangle\in\mathcal H_F
    \right\},
\end{equation}
where the tensor factors are understood to be placed in their original
order. Then every state $\rho$ satisfying
$\supp(\rho)\subseteq\mathcal W$ obeys
\begin{equation}
    \Tr\!\left[
        \rho
        \sum_{\substack{S\subseteq[m]\\|S|>t}}R_S
    \right]
    \leq
    2^m\binom mr d^r\,\eta^{t-r+1}.
    \label{Eq:many-rejections}
\end{equation}
\end{proposition}

\begin{proof}
For each $j\in[m]$, extend $|\theta_j\rangle$ to an orthonormal basis
$\{|e_{j,\nu}\rangle\}_{\nu=0}^{d-1}$ of $\mathcal H_j$, with
$|e_{j,0}\rangle\coloneqq|\theta_j\rangle$. The resulting product
basis is
\begin{equation}
    |\Psi_{\boldsymbol{\nu}}\rangle
    \coloneqq
    \bigotimes_{j=1}^m|e_{j,\nu_j}\rangle,
    \qquad
    \boldsymbol{\nu}
    =
    (\nu_1,\ldots,\nu_m)
    \in\{0,\ldots,d-1\}^m.
\end{equation}
For each $\boldsymbol{\nu}$, let
\(
    F_{\boldsymbol{\nu}}
    \coloneqq
    \{j\in[m]:\nu_j\neq0\}
\)
and define
\(
    \mathcal I_r
    \coloneqq
    \{\boldsymbol{\nu}:|F_{\boldsymbol{\nu}}|\leq r\}.
\)
Then
\(
    \{|\Psi_{\boldsymbol{\nu}}\rangle\}_{\boldsymbol{\nu}\in\mathcal I_r}
\)
is an orthonormal basis of $\mathcal W$. Its cardinality is
\begin{equation}
    \Gamma_{m,r}
    \coloneqq
    |\mathcal I_r|
    =
    \sum_{q=0}^r\binom mq(d-1)^q
    \leq
    \binom mr d^r.
    \label{Eq:defect-subspace-dim}
\end{equation}
Indeed,
\(
    \binom mq
    \leq
    \binom mr\binom rq
\)
for $q\leq r$, which follows from
\(
    \binom mq\binom{m-q}{r-q}
    =
    \binom mr\binom rq.
\)

Since $\supp(\rho)\subseteq\mathcal W$, we may write
\begin{equation}
    \rho
    =
    \sum_{\boldsymbol{\nu},\boldsymbol{\mu}\in\mathcal I_r}
    \beta_{\boldsymbol{\nu}\boldsymbol{\mu}}
    |\Psi_{\boldsymbol{\nu}}\rangle
    \langle\Psi_{\boldsymbol{\mu}}|,
    \qquad
    \sum_{\boldsymbol{\nu}\in\mathcal I_r}
    \beta_{\boldsymbol{\nu}\boldsymbol{\nu}}
    =1.
\end{equation}
Applying the pinching inequality on $\mathcal W$ in the above
orthonormal basis gives~\cite{hayashi2001optimal}
\begin{equation}
    \rho
    \leq
    \Gamma_{m,r}
    \sum_{\boldsymbol{\nu}\in\mathcal I_r}
    \beta_{\boldsymbol{\nu}\boldsymbol{\nu}}
    |\Psi_{\boldsymbol{\nu}}\rangle
    \langle\Psi_{\boldsymbol{\nu}}|.
    \label{Eq:classical-pinching}
\end{equation}

Fix $S\subseteq[m]$ with $|S|>t$. For every
$\boldsymbol{\nu}\in\mathcal I_r$, we have
\(
    |S\setminus F_{\boldsymbol{\nu}}|
    \geq |S|-r.
\)
Moreover,
\begin{align}
    \langle\Psi_{\boldsymbol{\nu}}|
        R_S
    |\Psi_{\boldsymbol{\nu}}\rangle
    &=
    \prod_{j\in S}
    \langle e_{j,\nu_j}|
        (\id_j-T_j)
    |e_{j,\nu_j}\rangle
    \prod_{j\notin S}
    \langle e_{j,\nu_j}|
        T_j
    |e_{j,\nu_j}\rangle
    \nonumber\\
    &\leq
    \prod_{j\in S\setminus F_{\boldsymbol{\nu}}}
    \langle\theta_j|
        (\id_j-T_j)
    |\theta_j\rangle
    \leq
    \eta^{|S\setminus F_{\boldsymbol{\nu}}|}
    \leq
    \eta^{|S|-r}.
\end{align}
Here all factors not indexed by
$S\setminus F_{\boldsymbol{\nu}}$ have simply been bounded by one.

Since $R_S\geq0$, \zcref{Eq:classical-pinching} therefore implies
\begin{equation}
    \Tr(\rho R_S)
    \leq
    \Gamma_{m,r}
    \sum_{\boldsymbol{\nu}\in\mathcal I_r}
    \beta_{\boldsymbol{\nu}\boldsymbol{\nu}}
    \langle\Psi_{\boldsymbol{\nu}}|
        R_S
    |\Psi_{\boldsymbol{\nu}}\rangle
    \leq
    \Gamma_{m,r}\eta^{|S|-r}.
\end{equation}
Summing over all $S$ with $|S|>t$ gives
\begin{align}
    \Tr\!\left[
        \rho
        \sum_{\substack{S\subseteq[m]\\|S|>t}}R_S
    \right]
    &\leq
    \Gamma_{m,r}
    \sum_{s=t+1}^m
    \binom ms\eta^{s-r}
    \nonumber\\
    &\leq
    2^m\Gamma_{m,r}\eta^{t-r+1}
    \nonumber\\
    &\leq
    2^m\binom mr d^r\eta^{t-r+1},
\end{align}
where the penultimate inequality uses \(0\leq\eta\leq1\), and the final inequality follows from \zcref{Eq:defect-subspace-dim}.
\end{proof}
}
We can now combine \zcref{Prop:classical-outer} with outer
classical error correction to prove the robustness of classical
communication.

\begin{theorem}[(Universal robustness of classical communication)]
\label{Thm:classical-robustness}
Let $\Phi\in\operatorname{CPTP}(A,B)$. Then
{\begin{equation}
    C_{\mathrm{univ}}(\mathfrak S_\Phi)
    =
    \inf_{\mathbf N\in\mathfrak S_\Phi}
    C(\mathbf N)
    =
    C(\Phi).
\end{equation}}
\end{theorem}

\begin{proof}
We first prove the converse. The memoryless sequence
$\{\Phi^{\otimes n}\}_{n\in\mathbb N}$ belongs to
$\mathfrak S_\Phi$, by taking $r_n=0$. Therefore, every universally
achievable rate for $\mathfrak S_\Phi$ must, in particular, be
achievable for this memoryless sequence. It follows that
\begin{equation}
    C_{\mathrm{univ}}(\mathfrak S_\Phi)
    \leq
    C\!\left(\{\Phi^{\otimes n}\}_{n\in\mathbb N}\right)
    =
    C(\Phi).
    \label{eq:classical-universal-converse}
\end{equation}

We now prove achievability. {If $C(\Phi)=0$, then the converse above and non-negativity imply
\begin{equation}
    C_{\mathrm{univ}}(\mathfrak S_\Phi)=0.
\end{equation}
Moreover, since the memoryless sequence
$\{\Phi^{\otimes n}\}_{n\in\mathbb N}$ belongs to
$\mathfrak S_\Phi$,
\begin{equation}
    0
    \leq
    \inf_{\mathbf N\in\mathfrak S_\Phi} C(\mathbf N)
    \leq
    C(\Phi)
    =
    0.
\end{equation}
Thus the claimed three-way equality holds in this case. Assume
henceforth that $C(\Phi)>0$, and fix
\begin{equation}
    0\leq R<R_0<C(\Phi).
\end{equation}} 
Choose $0<\tau<1/4$ sufficiently small that
\begin{equation}
    (1-2\tau)R_0-h(2\tau)>R,
    \label{Eq:classical-rate-choice}
\end{equation}
where $h$ denotes the binary entropy. Next, choose $0<\eta<1$
sufficiently small that
\begin{equation}
    1+\tau\log\eta<0.
    \label{Eq:classical-eta-choice}
\end{equation}

{By the achievability part of the HSW theorem and expurgation}, there exist
$k,L\in\mathbb N$, pure states
$\{|\psi_a\rangle\}_{a\in[L]}$ on $A^k$, and a POVM
$\{D_a\}_{a\in[L]}$ on $B^k$ such that
\begin{equation}
    \frac{1}{k}\log L\geq R_0,
    \qquad
    \Tr\!\left[
        D_a\,
        \Phi^{\otimes k}(\psi_a)
    \right]
    \geq1-\eta
    \quad\forall\,a\in[L],
    \label{Eq:inner-HSW-code}
\end{equation}
where
\(
    \psi_a\coloneqq|\psi_a\rangle\langle\psi_a|.
\)
The encoding states may be taken to be pure. Indeed, if a mixed
encoding state
\(
    \rho_a=\sum_xp_x\psi_{a,x}
\)
has decoding success probability at least $1-\eta$, then at least one
pure state $\psi_{a,x}$ in this decomposition has decoding success
probability at least $1-\eta$.

%%% INSERT HERE
{Fix once and for all a minimal Stinespring isometry
\begin{equation}
    V_\Phi\in\mathcal{L}(A,BE)
\end{equation}
for $\Phi$. The environment $E$, and hence its dimension $d_E$,
depends only on $\Phi$. By
\zcref{cor:fixed-Stinespring-representation}, every
Choi-almost-iid channel along $\Phi$ admits the extension and
defect-operator representation in
\zcref{Prop:alternative-characterization} using this same
Stinespring isometry $V_\Phi$. In particular, neither $E$ nor
$V_\Phi$ depends on the particular channel sequence or on its
witnessing defect sequence.
}
For each $n\geq k$, let
\begin{equation}
    m\coloneqq\left\lfloor\frac{n}{k}\right\rfloor,
    \qquad
    n=mk+s,
    \qquad
    0\leq s<k,
\end{equation}
and set
\(
    t_m\coloneqq\lfloor\tau m\rfloor.
\)
Codes for the finitely many blocklengths $n<k$ may be chosen
arbitrarily. Since $m/n\to1/k$ and $t_m/m\to\tau$, every sequence
$r_n=o(n)$ satisfies
\(
    r_n=o(m)
\);
in particular, $r_n\leq t_m\leq m$ for all sufficiently large $n$.

We next construct an outer classical code
$ \mathcal C_m^{\mathrm{out}}$. For
$\boldsymbol a,\boldsymbol b\in[L]^m$, define their Hamming distance by
\begin{equation}
    d_{\mathrm H}(\boldsymbol a,\boldsymbol b)
    \coloneqq
    \bigl|\{j\in[m]:a_j\neq b_j\}\bigr|.
\end{equation}
The Gilbert--Varshamov bound~\cite{gilbert1952comparison,varshamov1957estimate} guarantees the existence of a code
$\mathcal C_m^{\mathrm{out}}$ with minimum Hamming distance at least
$2t_m+1$ and satisfying
\begin{equation}
    |\mathcal C_m^{\mathrm{out}}|
    \geq
    \frac{L^m}{
        \displaystyle
        \sum_{j=0}^{2t_m}
        \binom mj(L-1)^j
    }.
    \label{Eq:classical-GV}
\end{equation}
Thus,
\begin{equation}
    d_{\mathrm H}(\boldsymbol a,\boldsymbol a')
    \geq2t_m+1
    \qquad
    \forall\,
    \boldsymbol a\neq\boldsymbol a'
    \text{ in }\mathcal C_m^{\mathrm{out}}.
\end{equation}
In particular, the Hamming balls of radius $t_m$ centred at the
codewords in $\mathcal C_m^{\mathrm{out}}$ are pairwise disjoint.

Fix a pure state
\(
    \varphi_A=|\varphi\rangle\langle\varphi|_A
\)
for padding the remaining $s$ channel inputs. A codeword
$\boldsymbol a=(a_1,\ldots,a_m)\in \mathcal C_m^{\mathrm{out}}$ is encoded as
\begin{equation}
    \rho_{\boldsymbol a}^{(n)}
    \coloneqq
    \psi_{a_1}\otimes\cdots\otimes\psi_{a_m}
    \otimes\varphi_A^{\otimes s}
    \in\mathcal D(A^n).
    \label{eq:classical-full-enc}
\end{equation}

The receiver first applies the product POVM
\begin{equation}
    \left\{
        \Lambda_{\boldsymbol b}^{(n)}
        \coloneqq
        D_{b_1}\otimes\cdots\otimes D_{b_m}
        \otimes \id_{B^{\otimes s}}
    \right\}_{\boldsymbol b\in[L]^m}.
    \label{eq:classical-full-dec}
\end{equation}
Define a map
\(
    g_m:[L]^m\to \mathcal C_m^{\mathrm{out}}
\)
as follows. If there exists a codeword
$\boldsymbol a\in\mathcal C_m^{\mathrm{out}}$ satisfying
\(
    d_{\mathrm H}(\boldsymbol a,\boldsymbol b)\leq t_m,
\)
then $g_m(\boldsymbol b)$ is this unique codeword; otherwise,
$g_m(\boldsymbol b)$ is chosen arbitrarily. The final decoding POVM is
therefore
\begin{equation}
    \widehat D_{\boldsymbol a}^{(n)}
    \coloneqq
    \sum_{\substack{\boldsymbol b\in[L]^m\\
                    g_m(\boldsymbol b)=\boldsymbol a}}
    \Lambda_{\boldsymbol b}^{(n)},
    \qquad
    \boldsymbol a\in\mathcal C_m^{\mathrm{out}}.
\end{equation}
If $\boldsymbol a$ is transmitted, a decoding error can occur only if
the outcome $\boldsymbol b$ of the inner decoder satisfies
\(
    d_{\mathrm H}(\boldsymbol a,\boldsymbol b)>t_m.
\)

The construction up to this point depends only on $R$ and $\Phi$.
Now fix an arbitrary sequence
\(
    \mathbf N
    =
    \{\mathcal N^{(n)}_{A^n\to B^n}\}_{n\in\mathbb N}
    \in\mathfrak S_\Phi
\)
and a witnessing sequence $r_n=o(n)$. By
\zcref{Prop:alternative-characterization}, for each $n$ there exists
an extension channel
\({\mathcal N}^{(n)}_{A^n\to B^nE^n}
\)
of $\mathcal N^{(n)}$, based on the fixed Stinespring isometry
$V_\Phi$, with a decomposition
\begin{equation}
   {\mathcal N}^{(n)}_{A^n\to B^nE^n}(X)
    =
    \sum_{u,u'\in\mathcal T_n}
    \beta_{uu'}K_uXK_{u'}^\dagger,
    \qquad
    K_u\in
    \mathcal K\!\left(
        A^n\to B^nE^n,
        V_\Phi^{\otimes(n-r_n)}
    \right).
\end{equation}

For each $a\in[L]$, define the ideal extended output
\begin{equation}
    |\theta_a\rangle
    \coloneqq
    V_\Phi^{\otimes k}|\psi_a\rangle
    \in(BE)^{\otimes k}.
\end{equation}
For an input codeword
$\boldsymbol a=(a_1,\ldots,a_m)\in \mathcal C_m^{\mathrm{out}}$, let
\begin{equation}
    \omega_{\boldsymbol a}
    \coloneqq
    \Tr_{(BE)^{\otimes s}}
    {\mathcal N}^{(n)}_{A^n\to B^nE^n}
    \left(
        \rho_{\boldsymbol a}^{(n)}
    \right)
\end{equation}
be the extended output state on the first $mk$ output systems, after
discarding the $s$ output systems corresponding to the padding.

Let
\(
    \mathcal H_j\simeq(BE)^{\otimes k}
\)
be the extended output space of the $j$th block. We claim that
$\omega_{\boldsymbol a}$ is supported on
\begin{equation}
    \mathcal W_{\boldsymbol a}
    \coloneqq
    \operatorname{span}
    \left\{
        |\theta_{\boldsymbol a}\rangle_{F^c}
        \otimes|\xi_F\rangle
        :
        F\subseteq[m],\ |F|\leq r_n,\ 
        |\xi_F\rangle\in
        \bigotimes_{j\in F}\mathcal H_j
    \right\},
    \label{Eq:classical-block-defect-space}
\end{equation}
where
\(
    |\theta_{\boldsymbol a}\rangle_{F^c}
    \coloneqq
    \bigotimes_{j\notin F}|\theta_{a_j}\rangle
\),
with all tensor factors placed in their original order.

To prove the claim, define
\begin{equation}
    |\Psi_{\boldsymbol a}\rangle
    \coloneqq
    |\psi_{a_1}\rangle\otimes\cdots\otimes
    |\psi_{a_m}\rangle\otimes|\varphi\rangle^{\otimes s}.
\end{equation}
For every $u\in\mathcal T_n$, the definition of
$\mathcal K$ provides a set
$G_u\subseteq[n]$ with $|G_u|\leq r_n$ such that $K_u$ acts as
$V_\Phi$ on all positions in $G_u^c$, while its action on the
positions in $G_u$ may be arbitrary and correlated. Let
\begin{equation}
    I_j
    \coloneqq
    \{(j-1)k+1,\ldots,jk\},
    \qquad
    F_u
    \coloneqq
    \{j\in[m]:I_j\cap G_u\neq\varnothing\}.
\end{equation}
Then
\(
    |F_u|\leq|G_u|\leq r_n.
\)
For every $j\notin F_u$, the operator $K_u$ acts as
$V_\Phi^{\otimes k}$ on the entire $j$th block. Consequently,
\begin{equation}
    K_u|\Psi_{\boldsymbol a}\rangle
    =
    |\theta_{\boldsymbol a}\rangle_{F_u^c}
    \otimes
    |\xi_{u,\boldsymbol a}\rangle_{F_u,\mathrm{rem}}
\end{equation}
for some, possibly unnormalized, vector on the blocks in $F_u$ and
the remaining $s$ output systems. Hence
\(
    K_u|\Psi_{\boldsymbol a}\rangle
\)
belongs to
\(
    \mathcal W_{\boldsymbol a}
    \otimes(BE)^{\otimes s}.
\)
{The decomposition of
$\mathcal N^{(n)}_{A^n\to B^nE^n}
(\rho_{\mathbf a}^{(n)})$
therefore implies that the resulting extended output state is supported
on $\mathcal W_{\mathbf a}\otimes(BE)^{\otimes s}$.}
 Taking
the partial trace over the last $s$ output systems gives
\begin{equation}
    \supp(\omega_{\boldsymbol a})
    \subseteq
    \mathcal W_{\boldsymbol a},
\end{equation}
as claimed.

For a transmitted codeword $\boldsymbol a$, let
$\boldsymbol b\in[L]^m$ denote the outcome of the product inner
decoder. Its conditional distribution can be evaluated using the
extended output as
\begin{align}
    \Pr(\boldsymbol b\mid\boldsymbol a)
    &=
    \Tr\!\left[
        \Lambda_{\boldsymbol b}^{(n)}
        \mathcal N^{(n)}
        \bigl(\rho_{\boldsymbol a}^{(n)}\bigr)
    \right]
    \nonumber\\
    &=
    \Tr\!\left[
        \left(
            \bigotimes_{j=1}^m
            (D_{b_j}\otimes \id_{E^{\otimes k}})
        \right)
        \omega_{\boldsymbol a}
    \right].
    \label{Eq:inner-decoding-distribution}
\end{align}
For each $j\in[m]$, define
\begin{equation}
    T_j
    \coloneqq
    D_{a_j}\otimes \id_{E^{\otimes k}}.
\end{equation}
By \zcref{Eq:inner-HSW-code},
\begin{equation}
    \langle\theta_{a_j}|T_j|\theta_{a_j}\rangle
    =
    \Tr\!\left[
        D_{a_j}\Phi^{\otimes k}(\psi_{a_j})
    \right]
    \geq1-\eta.
\end{equation}

Let
\(
    S(\boldsymbol a,\boldsymbol b)
    \coloneqq
    \{j\in[m]:a_j\neq b_j\}.
\)
Since
\(
    \sum_{b_j\neq a_j}D_{b_j}=\id_{B^{\otimes k}}-D_{a_j},
\)
coarse-graining the product decoding POVM gives
\begin{equation}
    \Pr\!\left(
        S(\boldsymbol a,\boldsymbol b)=S
        \,\middle|\,
        \boldsymbol a
    \right)
    =
    \Tr(\omega_{\boldsymbol a}R_S),
\end{equation}
where $\{R_S\}_{S\subseteq[m]}$ is the product binary POVM constructed
from the effects $\{T_j\}_{j=1}^m$ as in
\zcref{Prop:classical-outer}. Consequently,
\begin{equation}
    \Pr\!\left(
        d_{\mathrm H}(\boldsymbol a,\boldsymbol b)>t_m
        \,\middle|\,
        \boldsymbol a
    \right)
    =
    \sum_{\substack{S\subseteq[m]\\|S|>t_m}}
    \Tr(\omega_{\boldsymbol a}R_S).
\end{equation}

Applying \zcref{Prop:classical-outer} with
\(
    d=(d_Bd_E)^k
\),
$r=r_n$, and $t=t_m$ therefore gives, for every
$\boldsymbol a\in \mathcal C_m^{\mathrm{out}}$ and all sufficiently large $n$,
\begin{equation}
    \Pr\!\left(
        \textnormal{decoding error}
        \,\middle|\,
        \boldsymbol a
    \right)
    \leq
    2^m
    \binom{m}{r_n}
    (d_Bd_E)^{kr_n}
    \eta^{t_m-r_n+1}.
    \label{Eq:classical-error-bound}
\end{equation}

%% Next from here
It remains to show that the maximal decoding error tends to zero.
For the fixed sequence
$\mathbf N\in\mathfrak S_\Phi$, let $\varepsilon_n(\mathbf N)$
denote the maximal error probability of the code constructed above.
By \zcref{Eq:classical-error-bound}, for all sufficiently large $n$,
\begin{equation}
    \varepsilon_n(\mathbf N)
    \leq
    2^m
    \binom{m}{r_n}
    (d_Bd_E)^{kr_n}
    \eta^{t_m-r_n+1}.
\end{equation}
Using
\(
    \binom{m}{r_n}
    \leq
    2^{m h(r_n/m)},
\)
the base-two logarithm of the right-hand side, divided by $m$, is
bounded above by
\begin{equation}
    1
    +
    h\!\left(\frac{r_n}{m}\right)
    +
    k\frac{r_n}{m}\log(d_Bd_E)
    +
    \left(
        \frac{t_m}{m}
        -
        \frac{r_n}{m}
        +
        \frac{1}{m}
    \right)
    \log\eta.
    \label{Eq:classical-normalized-error-bound}
\end{equation}
Since
\(
    r_n/m\to0
\)
and
\(
    t_m/m\to\tau,
\)
the expression in
\zcref{Eq:classical-normalized-error-bound} converges to
\begin{equation}
    1+\tau\log\eta<0,
\end{equation}
where the inequality follows from
\zcref{Eq:classical-eta-choice}. In particular, if
\(
    \gamma\coloneqq-(1+\tau\log\eta)>0,
\)
then, for all sufficiently large $n$,
\begin{equation}
    \varepsilon_n(\mathbf N)
    \leq
    2^{-\gamma m/2}.
\end{equation}
Thus the maximal decoding error tends to zero exponentially in $m$
and hence also exponentially in $n$.

It remains to estimate the communication rate. The number of messages
in the blocklength-$n$ code is
\(
    M_n=|\mathcal C_m^{\mathrm{out}}|.
\)
Since $\tau<1/4$, we have $2t_m/m<1/2$ for all sufficiently large
$m$. Consequently,
\begin{equation}
    \sum_{j=0}^{2t_m}
    \binom mj(L-1)^j
    \leq
    (2t_m+1)\,
    2^{m h(2t_m/m)}
    (L-1)^{2t_m}.
\end{equation}
Combining this estimate with the Gilbert--Varshamov bound
\zcref{Eq:classical-GV} gives
\begin{align}
    \frac{1}{n}\log|\mathcal C_m^{\mathrm{out}}|
    &\geq
    \frac{m}{n}\log L
    -
    \frac{1}{n}\log(2t_m+1)
    -
    \frac{m}{n}h\!\left(\frac{2t_m}{m}\right)
    -
    \frac{2t_m}{n}\log(L-1)
    \nonumber\\
    &=
    \frac{m}{n}
    \left[
        \log L
        -
        h\!\left(\frac{2t_m}{m}\right)
        -
        \frac{2t_m}{m}\log(L-1)
    \right]
    -
    \frac{1}{n}\log(2t_m+1).
\end{align}
Since $n=mk+s$ with $0\leq s<k$, while
$t_m/m\to\tau$, we have
\begin{equation}
    \frac{m}{n}\longrightarrow\frac{1}{k},
    \qquad
    \frac{2t_m}{m}\longrightarrow2\tau,
    \qquad
    \frac{1}{n}\log(2t_m+1)\longrightarrow0.
\end{equation}
The continuity of the binary entropy therefore implies
\begin{equation}
    \liminf_{n\to\infty}
    \frac{1}{n}\log|\mathcal C_m^{\mathrm{out}}|
    \geq
    \frac{1}{k}
    \left[
        \log L
        -
        h(2\tau)
        -
        2\tau\log(L-1)
    \right].
    \label{Eq:classical-final-rate}
\end{equation}
Using
\(
    \log(L-1)\leq\log L
\),
\(
    k\geq1
\),
and
\(
    k^{-1}\log L\geq R_0
\)
from \zcref{Eq:inner-HSW-code}, we obtain
\begin{align}
    \liminf_{n\to\infty}
    \frac{1}{n}\log|\mathcal C_m^{\mathrm{out}}|
    &\geq
    (1-2\tau)\frac{1}{k}\log L
    -
    \frac{1}{k}h(2\tau)
    \nonumber\\
    &\geq
    (1-2\tau)R_0-h(2\tau)
    \nonumber\\
    &>R,
\end{align}
where the last inequality follows from
\zcref{Eq:classical-rate-choice}.

The encoding states, the outer codebooks, and the decoding POVMs were
chosen independently of the particular sequence
$\mathbf N\in\mathfrak S_\Phi$ and of its witnessing defect sequence
$r_n=o(n)$. For every fixed $\mathbf N\in\mathfrak S_\Phi$, the above
argument shows that the maximal decoding error tends to zero.
Therefore, the same sequence of codes is universally reliable over
$\mathfrak S_\Phi$.

Since every $R<C(\Phi)$ is universally achievable, we conclude that
\begin{equation}
    C_{\mathrm{univ}}(\mathfrak S_\Phi)
    \geq
    C(\Phi).
\end{equation}
{Together with the converse \zcref{eq:classical-universal-converse}, this proves
\begin{equation}
    C_{\mathrm{univ}}(\mathfrak S_\Phi)
    =
    C(\Phi).
\end{equation}
Moreover, every universally achievable rate is achievable for each
individual $\mathbf N\in\mathfrak S_\Phi$, and hence
\begin{equation}
    C(\Phi)
    =
    C_{\mathrm{univ}}(\mathfrak S_\Phi)
    \leq
    \inf_{\mathbf N\in\mathfrak S_\Phi}
    C(\mathbf N).
\end{equation}
Conversely, the memoryless sequence
$\{\Phi^{\otimes n}\}_{n\in\mathbb N}$ belongs to
$\mathfrak S_\Phi$, so
\begin{equation}
    \inf_{\mathbf N\in\mathfrak S_\Phi}
    C(\mathbf N)
    \leq
    C(\Phi).
\end{equation}
Therefore,
\begin{equation}
    C_{\mathrm{univ}}(\mathfrak S_\Phi)
    =
    \inf_{\mathbf N\in\mathfrak S_\Phi}
    C(\mathbf N)
    =
    C(\Phi),
\end{equation}
which proves the theorem.}
\end{proof}

%%%%%%%%%%%%%%%%%%%%%%%%%%%%%%%%%%%%%%%%%%%%%%%%%%%%%%%%%%%%%%%%%%%%%%%%%%%%%%%%
}

{\section{Robustness of quantum communication}
\label{Sec:robust-quantum-communication}

We now turn to quantum communication over Choi-almost-iid channels. By
\emph{robustness}, we mean that every rate strictly below the quantum
capacity of the memoryless reference channel $\Phi$ remains achievable
over Choi-almost-iid channel sequences along $\Phi$. As in the classical
setting, we seek a \emph{universal} coding scheme, meaning that the same
sequence of encoders and decoders works for every such channel sequence.
Thus, for each target rate $R<Q(\Phi)$, the code sequence may depend on
$\Phi$ and on $R$, but is independent of the particular channel sequence
in $\mathfrak S_\Phi$ and of the defect parameters and extensions used to
establish its Choi-almost-iid property. Throughout this section,
logarithms are taken to be base two.

\begin{definition}[(Quantum communication codes)]
\label{Def:quantum-communication-sequence}
Let
\(
    \mathbf N
    \coloneqq
    \{\mathcal N^{(n)}_{A^n\to B^n}\}_{n\in\mathbb N}
\)
be a sequence of quantum channels. For each $n\in\mathbb N$, an
$(n,|M_n|)$ quantum communication code $\mathcal Q_n$ consists of an
encoding channel
\(
    \mathcal E^{(n)}\in\mathrm{CPTP}(M_n,A^n)
\)
and a decoding channel
\(
    \mathcal D^{(n)}\in\mathrm{CPTP}(B^n,M_n),
\)
where $M_n$ is a finite-dimensional quantum system and
\(
    |M_n|\coloneqq d_{M_n}=\dim\mathcal H_{M_n}.
\)

The transmission error of $\mathcal Q_n$ when used over
$\mathcal N^{(n)}$ is
\begin{equation}
    \varepsilon_\diamond
    \bigl(
        \mathcal Q_n,\mathcal N^{(n)}
    \bigr)
    \coloneqq
    \frac{1}{2}
    \left\|
        \mathcal D^{(n)}
        \circ\mathcal N^{(n)}
        \circ\mathcal E^{(n)}
        -
        \idc_{M_n}
    \right\|_\diamond .
\end{equation}

A rate $R\geq0$ is \emph{achievable} for $\mathbf N$ if there exists a
sequence of quantum communication codes
$\{\mathcal Q_n\}_{n\in\mathbb N}$ such that
\begin{equation}
    \lim_{n\to\infty}
    \varepsilon_\diamond
    \bigl(
        \mathcal Q_n,\mathcal N^{(n)}
    \bigr)
    =0,
    \qquad
    \liminf_{n\to\infty}
    \frac{1}{n}\log |M_n|
    \geq R.
\end{equation}
The quantum capacity of $\mathbf N$ is
\begin{equation}
    Q(\mathbf N)
    \coloneqq
    \sup\left\{
        R:
        R\text{ is achievable for }\mathbf N
    \right\}.
\end{equation}
\end{definition}

For the memoryless sequence
$\{\Phi_{A\to B}^{\otimes n}\}_{n\in\mathbb N}$, the preceding
definition reduces to the usual quantum capacity $Q(\Phi)$. By the
Lloyd--Shor--Devetak theorem~\cite{lloyd1997capacity,
shor2002quantum,devetak2005private},
\begin{equation}
    Q(\Phi)
    =
    \lim_{k\to\infty}
    \frac{1}{k}
    \max_{\phi_{RA^k}}
    I(R\rangle B^k)_\sigma
    =
    \sup_{k\geq1}
    \frac{1}{k}
    \max_{\phi_{RA^k}}
    I(R\rangle B^k)_\sigma,
    \label{Eq:quantum-capacity}
\end{equation}
where the maximum is over pure states
$\phi_{RA^k}\in\mathcal D(RA^k)$, with $R\simeq A^k$, and
\begin{equation}
    \sigma_{RB^k}
    \coloneqq
    \bigl(
        \idc_R\otimes\Phi_{A\to B}^{\otimes k}
    \bigr)
    (\phi_{RA^k}).
\end{equation}
Here
\(
    I(R\rangle B^k)_\sigma
    \coloneqq
    -H(R|B^k)_\sigma
\)
denotes the coherent information.

\begin{definition}[(Universal quantum capacity for $\mathfrak S_\Phi$)]
\label{Def:universal-quantum-capacity}
A rate $R\geq0$ is \emph{universally achievable} for
$\mathfrak S_\Phi$ if there exists a sequence of quantum communication
codes $\{\mathcal Q_n\}_{n\in\mathbb N}$, which may depend on $R$ and
$\Phi$ but is chosen independently of the particular sequence
$\mathbf N\in\mathfrak S_\Phi$, such that
\begin{equation}
    \liminf_{n\to\infty}
    \frac{1}{n}\log |M_n|
    \geq R
\end{equation}
and, for every
\(
    \mathbf N
    =
    \{\mathcal N^{(n)}_{A^n\to B^n}\}_{n\in\mathbb N}
    \in\mathfrak S_\Phi,
\)
\begin{equation}
    \lim_{n\to\infty}
    \varepsilon_\diamond
    \bigl(
        \mathcal Q_n,\mathcal N^{(n)}
    \bigr)
    =0.
\end{equation}
Any such sequence of codes is called a \emph{universally reliable
quantum communication scheme} for $\mathfrak S_\Phi$. The
\emph{universal quantum capacity} of $\mathfrak S_\Phi$ is
\begin{equation}
    Q_{\mathrm{univ}}(\mathfrak S_\Phi)
    \coloneqq
    \sup\left\{
        R:
        R\text{ is universally achievable for }\mathfrak S_\Phi
    \right\}.
    \label{Eq:universal-quantum-capacity-def}
\end{equation}
\end{definition}

The order of quantifiers in the preceding definition is important: the
sequence of codes is common to all $\mathbf N\in\mathfrak S_\Phi$,
whereas the rate at which the error converges to zero may depend on the
particular channel sequence. Every universally achievable rate for
$\mathfrak S_\Phi$ is achievable for each of its individual members.
Consequently,
\begin{equation}
    Q_{\mathrm{univ}}(\mathfrak S_\Phi)
    \leq
    \inf_{\mathbf N\in\mathfrak S_\Phi}Q(\mathbf N)
    \leq
    Q(\Phi),
\end{equation}
where the second inequality follows because the memoryless sequence
$\{\Phi^{\otimes n}\}_{n\in\mathbb N}$ belongs to
$\mathfrak S_\Phi$. For a general class of channel sequences, the first
inequality can be strict: a rate may be achievable for every individual
member using a code sequence that depends on that member, even though no
single code sequence achieves the rate for all members simultaneously.
The theorem below shows that this obstruction does not occur for
$\mathfrak S_\Phi$.

The robustness proof concatenates a finite-blocklength quantum
communication code for the reference channel with an outer quantum
error-correcting code. We first choose a blocklength-$k$ code for
$\Phi^{\otimes k}$ such that the effective channel $\Lambda$, obtained
after encoding and decoding, is close to the identity channel on a
finite-dimensional Hilbert space $Q$. Grouping a Choi-almost-iid
channel $\mathcal N^{(n)}$ into $m$ blocks of size $k$ then induces a
Choi-almost-iid channel on $Q^{\otimes m}$ along $\Lambda$. The outer
quantum error-correcting code must handle both the residual noise of
$\Lambda$ and the additional almost-iid deviations, including coherent
superpositions of different defect structures. Its universality follows
from the fact that it corrects all errors of sufficiently small block
weight, independently of their locations and operator structure. The
quantum Gilbert--Varshamov bound guarantees the existence of outer codes
whose rate loss can be made arbitrarily small. The required facts from
quantum error correction are collected in~\zcref{App:QEC}.

The first step is to divide the $n$ channel inputs into $m$ blocks of
size $k$, with a remainder of fewer than $k$ systems discarded. The
following lemma shows that this blocking procedure preserves the
Choi-almost-iid structure: if the original channel has defect parameter
$r$, then the induced block channel has defect parameter at most $r$.

\begin{lemma}[(Stability under block coding)]
\label{lem:blocking}
Let $\Phi\in\mathrm{CPTP}(A,B)$,
$\mathcal E\in\mathrm{CPTP}(Q,A^k)$, and
$\mathcal D\in\mathrm{CPTP}(B^k,Q)$, and define the effective
reference channel
\begin{equation}
    \Lambda
    \coloneqq
    \mathcal D\circ\Phi^{\otimes k}\circ\mathcal E
    \in\mathrm{CPTP}(Q,Q).
\end{equation}
Let $n=mk+s$, where $0\leq s<k$, and let
$\mathcal N^{(n)}_{A^n\to B^n}$ be an
$\binom{n}{r}$-Choi-almost-iid channel along $\Phi$, with $r\leq m$.
Fix a pure state
$\varphi_A\coloneqq|\varphi\rangle\langle\varphi|_A$, and define the
induced channel $\mathcal M^{(m)}\in\mathrm{CPTP}(Q^m,Q^m)$ by
\begin{equation}
    \mathcal M^{(m)}(\rho)
    \coloneqq
    \mathcal D^{\otimes m}
    \left[
        \Tr_{B^s}
        \mathcal N^{(n)}
        \left(
            \mathcal E^{\otimes m}(\rho)
            \otimes
            \varphi_A^{\otimes s}
        \right)
    \right],
    \label{Eq:induced-block-channel}
\end{equation}
where $B^s$ denotes the final $s$ output systems, which are discarded.
Then $\mathcal M^{(m)}$ is an
$\binom{m}{r}$-Choi-almost-iid channel along $\Lambda$.
\end{lemma}

The condition $r\leq m$ is natural for the blocking argument. In the
asymptotic application, it is satisfied for all sufficiently large
$n$, since $r_n=o(n)$ whereas $m=\lfloor n/k\rfloor$ for fixed $k$.

\begin{proof}
See \zcref{app:proof-of-lem-blocking}.
\end{proof}

Having reduced the original channel to a Choi-almost-iid channel along
an effective reference channel $\Lambda$, we next show that, whenever
$\Lambda$ is sufficiently close to the identity, an outer quantum
error-correcting code converts the induced channel into one that is
close to the identity on the encoded space.
\begin{proposition}[(Outer quantum error correction)]
\label{Prop:outer-QEC}
Let $\ell\geq1$ be an integer, let
$Q\simeq(\mathbb C^2)^{\otimes\ell}$, and let
$0\leq r\leq t\leq m$ be integers. Suppose that:
\begin{enumerate}
\item[(i)] $\Lambda\in\mathrm{CPTP}(Q,Q)$ satisfies
\(
    \|\Lambda-\idc_Q\|_\diamond
    \leq
    \varepsilon
    <1/2;
\)

\item[(ii)] $\mathcal M^{(m)}\in\mathrm{CPTP}(Q^m,Q^m)$ is an
$\binom{m}{r}$-Choi-almost-iid channel along $\Lambda$;

\item[(iii)] there exists an
$((\ell m,|M|,2\ell t+1))$ quantum error-correcting code with encoding
isometry $U:M\to Q^m$ and recovery channel
$\mathcal R\in\mathrm{CPTP}(Q^m,M)$.
\end{enumerate}
Define $\mathcal U\coloneqq\Ad_U$. Then
\begin{equation}
    \frac12
    \left\|
        \mathcal R\circ\mathcal M^{(m)}\circ\mathcal U
        -
        \idc_M
    \right\|_\diamond
    \leq
    2^{\,m+1+\frac{5\ell-2}{2}r}
    \binom{m}{r}^{1/2}
    (2\varepsilon)^{\frac{t-r+1}{2}}.
    \label{Eq:outer-QEC-bound}
\end{equation}
\end{proposition}

\begin{proof}
See~\zcref{app:proof-of-QEC-Prop}.
\end{proof}

We can now combine the block-coding argument with the outer quantum
error-correcting code to obtain the main robustness result.

\begin{theorem}[(Universal robustness of quantum communication)]
\label{Thm:quantum-robustness}
Let $\Phi\in\mathrm{CPTP}(A,B)$. Then
\begin{equation}
    Q_{\mathrm{univ}}(\mathfrak S_\Phi)
    =
    \inf_{\mathbf N\in\mathfrak S_\Phi}
    Q(\mathbf N)
    =
    Q(\Phi).
\end{equation}
\end{theorem}

\begin{proof}
For the converse, observe that the memoryless sequence
$\{\Phi^{\otimes n}\}_{n\in\mathbb N}$ belongs to
$\mathfrak S_\Phi$, by taking $r_n=0$ for every $n$. Hence every rate
that is universally achievable for $\mathfrak S_\Phi$ must be
achievable for the memoryless sequence. Therefore,
\begin{equation}
    Q_{\mathrm{univ}}(\mathfrak S_\Phi)
    \leq
    Q\!\left(\{\Phi^{\otimes n}\}_{n\in\mathbb N}\right)
    =
    Q(\Phi).
    \label{eq:quantum-universal-converse}
\end{equation}

We now prove achievability. {If $Q(\Phi)=0$, then the converse above and non-negativity imply
\begin{equation}
    Q_{\mathrm{univ}}(\mathfrak S_\Phi)=0.
\end{equation}
Moreover, since the memoryless sequence
$\{\Phi^{\otimes n}\}_{n\in\mathbb N}$ belongs to
$\mathfrak S_\Phi$,
\begin{equation}
    0
    \leq
    \inf_{\mathbf N\in\mathfrak S_\Phi} Q(\mathbf N)
    \leq
    Q(\Phi)
    =
    0.
\end{equation}
Thus the claimed three-way equality holds in this case. Assume
henceforth that $Q(\Phi)>0$, and fix
\begin{equation}
    0\leq R<R_0<R_1<Q(\Phi).
\end{equation}}
Let
\begin{equation}
    h(x)
    \coloneqq
    -x\log x-(1-x)\log(1-x)
\end{equation}
denote the binary entropy, and, for $0\leq\tau<1/2$, define
\begin{equation}
    g(\tau)
    \coloneqq
    1-2\tau\log 3-h(2\tau).
\end{equation}
Since $g(\tau)\to1$ as $\tau\to0^+$, we may choose
$0<\tau<1/4$ sufficiently small that
\begin{equation}
    g(\tau)
    >
    \frac{R}{R_0}.
\end{equation}
We may then choose $\gamma$ such that
\begin{equation}
    \frac{R}{R_0}
    <
    \gamma
    <
    g(\tau).
\end{equation}
Finally, choose $0<\varepsilon<1/2$ sufficiently small that
\begin{equation}
    1+\frac{\tau}{2}\log2\varepsilon
    <
    0.
    \label{Eq:quantum-epsilon-choice}
\end{equation}
{Since the set of achievable rates is downward closed,
$R_1<Q(\Phi)$ is achievable. Hence there exists a sequence of
quantum communication codes for $\Phi^{\otimes k}$ with message
systems $Q'_k$, encoding and decoding channels
\begin{equation}
    \mathcal E'_k\in\mathrm{CPTP}(Q'_k,A^k),
    \qquad
    \mathcal D'_k\in\mathrm{CPTP}(B^k,Q'_k),
\end{equation}
and effective channels
\begin{equation}
    \Lambda'_k
    \coloneqq
    \mathcal D'_k
    \circ
    \Phi^{\otimes k}
    \circ
    \mathcal E'_k,
\end{equation}
such that, writing
$d'_k\coloneqq\dim\mathcal H_{Q'_k}$,
\begin{equation}
    \liminf_{k\to\infty}
    \frac{1}{k}\log d'_k
    \geq
    R_1,
    \qquad
    \lim_{k\to\infty}
    \left\|
        \Lambda'_k-\idc_{Q'_k}
    \right\|_\diamond
    =
    0.
\end{equation}

Since $R_0<R_1$, we may choose $k$ sufficiently large that
\begin{equation}
    d'_k
    \geq
    2^{\lceil kR_0\rceil},
    \qquad
    \left\|
        \Lambda'_k-\idc_{Q'_k}
    \right\|_\diamond
    \leq
    \varepsilon.
\end{equation}
Set
\begin{equation}
    \ell
    \coloneqq
    \lceil kR_0\rceil
\end{equation}
and let
\begin{equation}
    Q\simeq(\mathbb C^2)^{\otimes\ell}.
\end{equation}
Choose an isometry
$\iota:Q\to Q'_k$ and define the embedding channel
$\mathcal I\in\mathrm{CPTP}(Q,Q'_k)$ by
\begin{equation}
    \mathcal I(\rho)
    \coloneqq
    \iota\rho\iota^\dagger.
\end{equation}
Fix an arbitrary state $\tau_Q\in\mathcal D(Q)$ and define a channel
$\mathcal P\in\mathrm{CPTP}(Q'_k,Q)$ by
\begin{equation}
    \mathcal P(X)
    \coloneqq
    \iota^\dagger X\iota
    +
    \Tr\!\left[
        \bigl(
            \id_{Q'_k}-\iota\iota^\dagger
        \bigr)X
    \right]
    \tau_Q,
    \qquad
    X\in\mathcal L(Q'_k).
\end{equation}
This channel is a retraction of $\mathcal I$, in the sense that
\begin{equation}
    \mathcal P\circ\mathcal I
    =
    \idc_Q.
\end{equation}

Define the restricted encoding and decoding channels by
\begin{equation}
    \mathcal E
    \coloneqq
    \mathcal E'_k\circ\mathcal I,
    \qquad
    \mathcal D
    \coloneqq
    \mathcal P\circ\mathcal D'_k.
\end{equation}
The corresponding effective channel is
\begin{equation}
    \Lambda
    \coloneqq
    \mathcal D
    \circ
    \Phi^{\otimes k}
    \circ
    \mathcal E
    =
    \mathcal P
    \circ
    \Lambda'_k
    \circ
    \mathcal I.
\end{equation}
Since $\mathcal P\circ\mathcal I=\idc_Q$ and the diamond norm is
contractive under pre- and postcomposition with quantum channels,
\begin{align}
    \left\|
        \Lambda-\idc_Q
    \right\|_\diamond
    &=
    \left\|
        \mathcal P
        \circ
        \bigl(
            \Lambda'_k-\idc_{Q'_k}
        \bigr)
        \circ
        \mathcal I
    \right\|_\diamond
    \\
    &\leq
    \left\|
        \Lambda'_k-\idc_{Q'_k}
    \right\|_\diamond
    \leq
    \varepsilon.
\end{align}
Moreover, $\ell/k\geq R_0$. We have therefore obtained
$\ell\in\mathbb N$, a system
$Q\simeq(\mathbb C^2)^{\otimes\ell}$, and channels
$\mathcal E\in\mathrm{CPTP}(Q,A^k)$ and
$\mathcal D\in\mathrm{CPTP}(B^k,Q)$ such that
\begin{equation}
    \frac{\ell}{k}
    \geq
    R_0,
    \qquad
    \left\|
        \Lambda-\idc_Q
    \right\|_\diamond
    \leq
    \varepsilon,
    \qquad
    \Lambda
    \coloneqq
    \mathcal D
    \circ
    \Phi^{\otimes k}
    \circ
    \mathcal E.
    \label{Eq:robust-inner-code}
\end{equation}}

For each sufficiently large $n$, let
\begin{equation}
    m
    \coloneqq
    \left\lfloor\frac{n}{k}\right\rfloor,
    \qquad
    n=mk+s,
    \qquad
    0\leq s<k,
\end{equation}
and set
\begin{equation}
    t_m
    \coloneqq
    \lfloor\tau m\rfloor,
    \qquad
    q_n
    \coloneqq
    \lfloor\gamma\ell m\rfloor.
\end{equation}
By the quantum Gilbert--Varshamov bound
(\zcref{Thm:quantum-GV}), an
$((\ell m,2^{q_n},2\ell t_m+1))$ quantum error-correcting code exists
whenever
\begin{equation}
    \frac{q_n}{\ell m}
    \leq
    1-
    \frac{2\ell t_m+1}{\ell m}\log3
    -
    h\!\left(
        \frac{2\ell t_m+1}{\ell m}
    \right)
    =
    g\!\left(
        \frac{2\ell t_m+1}{2\ell m}
    \right).
\end{equation}
The right-hand side converges to $g(\tau)$, whereas
$q_n/(\ell m)\to\gamma<g(\tau)$. Thus such a code exists for every
sufficiently large $n$. Fix one such code for each of these
blocklengths. For the finitely many remaining values of $n$, choose an
arbitrary one-dimensional code. The resulting message dimensions
satisfy
\begin{equation}
    \lim_{n\to\infty}
    \frac{1}{n}\log|M_n|
    =
    \lim_{n\to\infty}
    \frac{q_n}{n}
    =
    \gamma\frac{\ell}{k}
    \geq
    \gamma R_0
    >
    R.
    \label{Eq:robust-quantum-rate}
\end{equation}

For every sufficiently large $n$, let
$U_m:M_n\to Q^m$ be the encoding isometry and
$\mathcal R_m\in\mathrm{CPTP}(Q^m,M_n)$ the recovery channel of the
chosen $((\ell m,2^{q_n},2\ell t_m+1))$ code, and define
$\mathcal U_m\coloneqq\Ad_{U_m}$. Fix a pure state
$\varphi_A\coloneqq|\varphi\rangle\langle\varphi|_A$. Define the
concatenated encoding and decoding channels by
\begin{align}
    \mathcal E^{(n)}_{M_n\to A^n}(\rho)
    &\coloneqq
    \mathcal E^{\otimes m}
    \bigl(
        U_m\rho U_m^\dagger
    \bigr)
    \otimes
    \varphi_A^{\otimes s},
    \label{eq:full-enc}\\
    \mathcal D^{(n)}_{B^n\to M_n}(\sigma)
    &\coloneqq
    \mathcal R_m
    \left[
        \mathcal D^{\otimes m}
        \bigl(
            \Tr_{B^s}\sigma
        \bigr)
    \right].
    \label{eq:full-dec}
\end{align}

Now fix an arbitrary
$\mathbf N=\{\mathcal N^{(n)}\}_{n\in\mathbb N}
\in\mathfrak S_\Phi$, and let $r_n=o(n)$ be a witnessing sequence such
that $\mathcal N^{(n)}$ is an
$\binom{n}{r_n}$-Choi-almost-iid channel along $\Phi$. Since
\begin{equation}
    \frac{m}{n}
    \longrightarrow
    \frac{1}{k},
    \qquad
    \frac{t_m}{m}
    \longrightarrow
    \tau,
    \qquad
    \frac{r_n}{m}
    \longrightarrow
    0,
\end{equation}
we have $r_n\leq t_m\leq m$ for all sufficiently large $n$.

Define the induced channel
\begin{equation}
    \mathcal M^{(m)}_{Q^m\to Q^m}(\rho)
    \coloneqq
    \mathcal D^{\otimes m}
    \left[
        \Tr_{B^s}
        \mathcal N^{(n)}
        \left(
            \mathcal E^{\otimes m}(\rho)
            \otimes
            \varphi_A^{\otimes s}
        \right)
    \right].
\end{equation}
By~\zcref{lem:blocking}, $\mathcal M^{(m)}$ is an
$\binom{m}{r_n}$-Choi-almost-iid channel along $\Lambda$. The three
conditions of~\zcref{Prop:outer-QEC} are therefore satisfied for all
sufficiently large $n$, and that proposition gives
\begin{align}
    \frac12
    \left\|
        \mathcal D^{(n)}
        \circ\mathcal N^{(n)}
        \circ\mathcal E^{(n)}
        -
        \idc_{M_n}
    \right\|_\diamond
    &=
    \frac12
    \left\|
        \mathcal R_m
        \circ\mathcal M^{(m)}
        \circ\mathcal U_m
        -
        \idc_{M_n}
    \right\|_\diamond
    \nonumber\\
    &\leq
    2^{\,m+1+\frac{5\ell-2}{2}r_n}
    \binom{m}{r_n}^{1/2}
    (2\varepsilon)^{\frac{t_m-r_n+1}{2}}.
    \label{Eq:robust-final-bound}
\end{align}

Let $B_n$ denote the right-hand side of
\zcref{Eq:robust-final-bound}. Using
\begin{equation}
    \binom{m}{r_n}
    \leq
    2^{m h(r_n/m)},
\end{equation}
we obtain
\begin{align}
    \frac{1}{m}\log B_n
    \leq{}&
    1+
    \frac{1}{m}
    +
    \frac{5\ell-2}{2}\frac{r_n}{m}
    +
    \frac12 h\!\left(\frac{r_n}{m}\right)
    \nonumber\\
    &+
    \frac12
    \left(
        \frac{t_m}{m}
        -
        \frac{r_n}{m}
        +
        \frac{1}{m}
    \right)
    \log2\varepsilon.
\end{align}
Consequently,
\begin{equation}
    \limsup_{n\to\infty}
    \frac{1}{m}\log B_n
    \leq
    1+
    \frac{\tau}{2}\log2\varepsilon
    <
    0,
\end{equation}
where the last inequality is
\zcref{Eq:quantum-epsilon-choice}. Hence $B_n$ decays exponentially in
$m$ and, in particular, tends to zero.

Thus, for every fixed $\mathbf N\in\mathfrak S_\Phi$, the transmission
error tends to zero. The blocklength at which the preceding estimates
become effective may depend on $\mathbf N$, through its witnessing
sequence $r_n$, but the encoders and decoders do not. Together
with~\zcref{Eq:robust-quantum-rate}, this proves that every
$R<Q(\Phi)$ is universally achievable. Therefore,
\begin{equation}
    Q_{\mathrm{univ}}(\mathfrak S_\Phi)
    \geq
    Q(\Phi).
\end{equation}
Combining this with the converse proves
\begin{equation}
    Q_{\mathrm{univ}}(\mathfrak S_\Phi)
    =
    Q(\Phi).
\end{equation}
Moreover, universal achievability implies individual achievability,
and therefore
\begin{equation}
    Q(\Phi)
    =
    Q_{\mathrm{univ}}(\mathfrak S_\Phi)
    \leq
    \inf_{\mathbf N\in\mathfrak S_\Phi}
    Q(\mathbf N).
    \end{equation}
Conversely, since the memoryless sequence
$\{\Phi^{\otimes n}\}_{n\in\mathbb N}$ belongs to
$\mathfrak S_\Phi$,
\begin{equation}
    \inf_{\mathbf N\in\mathfrak S_\Phi}
    Q(\mathbf N)
    \leq
    Q(\Phi).
\end{equation}
Consequently,
\begin{equation}
    Q_{\mathrm{univ}}(\mathfrak S_\Phi)
    =
    \inf_{\mathbf N\in\mathfrak S_\Phi}
    Q(\mathbf N)
    =
    Q(\Phi),
\end{equation}
which proves the theorem.
\end{proof}}
\newpage
\section*{Acknowledgements}
BB and ND are supported by the Engineering and Physical Sciences Research Council [Grant Ref: EP/Y028732/1]. LY is also supported by the Engineering and Physical Sciences Research Council. The authors thank Filippo Girardi for initial discussions.

\textbf{Statement on AI usage.} OpenAI's GPT-5.6 Sol was used to assist with brainstorming and exploring proof strategies during the development of this work. It was also used to transcribe handwritten proofs into \LaTeX. The organization and exposition of the results are due to the authors, who take full responsibility for the contents in the manuscript.
%%%%%%%%%%%%%%%%%%%%%%%%%%%%%%%%%%%%%%%%%%%%%%%%%%%%%%%%%%%%%%%%%%%%%%%%%%%%%%%%
\bibliography{main}

@article{mazzola2026almost,
  title={Almost-iid information theory},
  author={Mazzola, Giulia and Sutter, David and Renner, Renato},
  journal={arXiv preprint arXiv:2603.15792},
  year={2026}
}

@unpublished{gottesman2024surviving,
  title = {Surviving as a Quantum Computer in a Classical World},
  author = {Gottesman, Daniel},
  note = {Textbook manuscript, 2024 draft},
  year = {2024},
  url = {https://www.cs.umd.edu/class/spring2024/cmsc858G/QECCbook-2024-ch1-15.pdf}
}

@article{lloyd1997capacity,
  title={Capacity of the noisy quantum channel},
  author={Lloyd, Seth},
  journal={Physical Review A},
  volume={55},
  number={3},
  pages={1613},
  year={1997},
  publisher={APS}
}

@unpublished{shor2002quantum,
  title = {The quantum channel capacity and coherent information},
  author = {Shor, Peter W.},
  note = {Lecture notes, MSRI Workshop on Quantum Computation},
  year = {2002}
}

@article{devetak2005private,
  title={The private classical capacity and quantum capacity of a quantum channel},
  author={Devetak, Igor},
  journal={IEEE Transactions on Information Theory},
  volume={51},
  number={1},
  pages={44--55},
  year={2005},
  publisher={IEEE}
}

@article{kretschmann2008information,
  title={The information-disturbance tradeoff and the continuity of Stinespring's representation},
  author={Kretschmann, Dennis and Schlingemann, Dirk and Werner, Reinhard F},
  journal={IEEE transactions on information theory},
  volume={54},
  number={4},
  pages={1708--1717},
  year={2008},
  publisher={IEEE}
}

@article{mazzola2026robust,
  title={Robust generalized quantum Stein's lemma},
  author={Mazzola, Giulia and Sutter, David and Renner, Renato},
  journal={arXiv preprint arXiv:2605.16500},
  year={2026}
}

@article{girardi2026new,
  title={New approaches to almost iid information theory},
  author={Girardi, Filippo and De Palma, Giacomo and Lami, Ludovico},
  journal={arXiv preprint arXiv:2605.15114},
  year={2026}
}

@article{sorensen2001entanglement,
  title={Entanglement and extreme spin squeezing},
  author={S{\o}rensen, Anders S and M{\o}lmer, Klaus},
  journal={Physical review letters},
  volume={86},
  number={20},
  pages={4431},
  year={2001},
  publisher={APS}
}

@article{de2021quantum,
  title={The quantum Wasserstein distance of order 1},
  author={De Palma, Giacomo and Marvian, Milad and Trevisan, Dario and Lloyd, Seth},
  journal={IEEE Transactions on Information Theory},
  volume={67},
  number={10},
  pages={6627--6643},
  year={2021},
  publisher={IEEE}
}

@article{holevo1998capacity,
  title={The capacity of the quantum channel with general signal states},
  author={Holevo, Alexander S},
  journal={IEEE Transactions on Information Theory},
  volume={44},
  number={1},
  pages={269--273},
  year={1998},
  publisher={IEEE}
}

@article{schumacher1997sending,
  title={Sending classical information via noisy quantum channels},
  author={Schumacher, Benjamin and Westmoreland, Michael D},
  journal={Physical Review A},
  volume={56},
  number={1},
  pages={131},
  year={1997},
  publisher={APS}
}

@article{hayashi2001optimal,
  title={Optimal sequence of POVMs in the sense of Stein's lemma in quantum hypothesis testing},
  author={Hayashi, Masahito},
  journal={arXiv preprint quant-ph/0107004},
  year={2001}
}

@article{christandl2007one,
  title={One-and-a-half quantum de Finetti theorems},
  author={Christandl, Matthias and K{\"o}nig, Robert and Mitchison, Graeme and Renner, Renato},
  journal={Communications in mathematical physics},
  volume={273},
  number={2},
  pages={473--498},
  year={2007},
  publisher={Springer}
}

@article{schumacher1995quantum,
  title={Quantum coding},
  author={Schumacher, Benjamin},
  journal={Physical Review A},
  volume={51},
  number={4},
  pages={2738},
  year={1995},
  publisher={APS}
}

@article{bennett1996concentrating,
  title={Concentrating partial entanglement by local operations},
  author={Bennett, Charles H and Bernstein, Herbert J and Popescu, Sandu and Schumacher, Benjamin},
  journal={Physical Review A},
  volume={53},
  number={4},
  pages={2046},
  year={1996},
  publisher={APS}
}

@article{horodecki2007quantum,
  title={Quantum state merging and negative information},
  author={Horodecki, Micha{\l} and Oppenheim, Jonathan and Winter, Andreas},
  journal={Communications in Mathematical Physics},
  volume={269},
  number={1},
  pages={107--136},
  year={2007},
  publisher={Springer}
}

@article{devetak2008exact,
  title={Exact cost of redistributing multipartite quantum states},
  author={Devetak, Igor and Yard, Jon},
  journal={Physical Review Letters},
  volume={100},
  number={23},
  pages={230501},
  year={2008},
  publisher={APS}
}

@article{bennett2002entanglement,
  title={Entanglement-assisted capacity of a quantum channel and the reverse Shannon theorem},
  author={Bennett, Charles H and Shor, Peter W and Smolin, John A and Thapliyal, Ashish V},
  journal={IEEE transactions on Information Theory},
  volume={48},
  number={10},
  pages={2637--2655},
  year={2002},
  publisher={IEEE}
}

@article{han1993approximation,
  title={Approximation theory of output statistics},
  author={Han, Te Sun and Verd{\'u}, Sergio},
  journal={IEEE Transactions on Information Theory},
  volume={39},
  number={3},
  pages={752--772},
  year={1993},
  publisher={IEEE}
}

@article{verdu1994general,
  title={A general formula for channel capacity},
  author={Verd{\'u}, Sergio and others},
  journal={IEEE Transactions on Information Theory},
  volume={40},
  number={4},
  pages={1147--1157},
  year={1994},
  publisher={IEEE}
}

@book{hayashi2006quantum,
  title={Quantum information: an introduction},
  author={Hayashi, Masahito},
  year={2006},
  publisher={Springer}
}

@inproceedings{bowen2006beyond,
  title={Beyond iid in quantum information theory},
  author={Bowen, Garry and Datta, Nilanjana},
  booktitle={2006 IEEE International Symposium on Information Theory},
  pages={451--455},
  year={2006},
  organization={IEEE}
}

@article{renner2008security,
  title={Security of quantum key distribution},
  author={Renner, Renato},
  journal={International Journal of Quantum Information},
  volume={6},
  number={01},
  pages={1--127},
  year={2008},
  publisher={World Scientific}
}

@article{datta2009min,
  title={Min-and max-relative entropies and a new entanglement monotone},
  author={Datta, Nilanjana},
  journal={IEEE Transactions on Information Theory},
  volume={55},
  number={6},
  pages={2816--2826},
  year={2009},
  publisher={IEEE}
}

@book{tomamichel2015quantum,
  title={Quantum information processing with finite resources: mathematical foundations},
  author={Tomamichel, Marco},
  year={2015},
  publisher={Springer}
}

@article{renner2007symmetry,
  title={Symmetry of large physical systems implies independence of subsystems},
  author={Renner, Renato},
  journal={Nature Physics},
  volume={3},
  number={9},
  pages={645--649},
  year={2007},
  publisher={Nature Publishing Group UK London}
}

@article{girardi2026quantum,
  title={Quantum Shannon theory made robust: a tale of three protocols for almost iid sources},
  author={Girardi, Filippo and Datta, Nilanjana and De Palma, Giacomo and Lami, Ludovico},
  journal={arXiv preprint arXiv:2605.18726},
  year={2026}
}

@article{datta2026entropy,
  title={Entropy Concentration and Universal Typicality for Weakly Almost iid Quantum Sources},
  author={Datta, Nilanjana},
  journal={arXiv preprint arXiv:2605.20092},
  year={2026}
}

@article{datta2026robustness,
  title={Robustness of Entanglement Manipulation for almost iid sources},
  author={Datta, Nilanjana},
  journal={arXiv preprint arXiv:2606.06392},
  year={2026}
}

@unpublished{companion-paper,
    author = {Ye, Liuhang and Bergh, Bjarne and Datta, Nilanjana},
    title  = {Quantum de Finetti theorems for states and channels in any distance measure},
    note   = {Companion paper to appear on arXiv concurrently with the present work},
    year   = {2026}
}

@article{gilbert1952comparison,
  title={A comparison of signalling alphabets},
  author={Gilbert, Edgar N},
  journal={The Bell system technical journal},
  volume={31},
  number={3},
  pages={504--522},
  year={1952},
  publisher={Nokia Bell Labs}
}

@article{varshamov1957estimate,
  title={Estimate of the number of signals in error correcting codes},
  author={Varshamov, Rom Rubenovich},
  journal={Docklady Akad. Nauk, SSSR},
  volume={117},
  pages={739--741},
  year={1957}
}

@article{berta2022semidefinite,
  title = {Semidefinite programming hierarchies for constrained bilinear optimization},
  author = {Berta, Mario and Borderi, Francesco and Fawzi, Omar and Scholz, Volkher B},
  journal = {Mathematical Programming},
  volume = {194},
  number = {1},
  pages = {781--829},
  year = {2022},
  publisher = {Springer},
}

@article{choi1975completely,
  title={Completely positive linear maps on complex matrices},
  author={Choi, Man-Duen},
  journal={Linear algebra and its applications},
  volume={10},
  number={3},
  pages={285--290},
  year={1975},
  publisher={Elsevier}
}

@article{jamiolkowski1972linear,
  title={Linear transformations which preserve trace and positive semidefiniteness of operators},
  author={Jamio{\l}kowski, Andrzej},
  journal={Reports on mathematical physics},
  volume={3},
  number={4},
  pages={275--278},
  year={1972},
  publisher={Elsevier}
}

@article{girardi2026generalised,
  title={Generalised quantum Stein's lemma more robust than ever},
  author={Girardi, Filippo and Lee, Kuan-Yi and Hayashi, Masahito and Lami, Ludovico},
  journal={arXiv preprint arXiv:2609.17309},
  year={2026}
}

@article{singh2026short,
  title={A short proof of the modified Kretschmann-Schlingemann-Werner conjecture},
  author={Singh, Satvik},
  journal={arXiv preprint arXiv:2606.16418},
  year={2026}
}
%%%%%%%%%%%%%%%%%%%%%%%%%%%%%%%%%%%%%%%%%%%%%%%%%%%%%%%%%%%%%%%%%%%%%%%%%%%%%%%%
\appendix
%%%%%%%%%%%%%%%%%%%%%%%%%%%%%%%%%%%%%%%%%%%%%%%%%%%%%%%%%%%%%%%%%%%%%%%%%%%%%%%%
\section{Proof of~\zcref{cor:fixed-Stinespring-representation}}
\label{app:proof-of-fixed-Stinespring-representation}

\begin{proof}
By \zcref{Prop:alternative-characterization}, there exist a system
$E'$, a Stinespring isometry
\begin{equation}
    V'_\Phi\in\mathcal{L}(A,BE')
\end{equation}
for $\Phi$, and a permutation-covariant extension channel
\begin{equation}
    \mathcal N'_{A^n\to B^nE'^n}
\end{equation}
of $\mathcal N_{A^n\to B^n}$ satisfying the required defect-support
condition with respect to $V'_\Phi$.

Since $V_\Phi$ is minimal, the uniqueness of minimal Stinespring
dilations implies that there exists an isometry
$W:E\to E'$ such that
\begin{equation}
    V'_\Phi
    =
    \bigl(
        \id_B\otimes W
    \bigr)
    V_\Phi .
\end{equation}
Fix an arbitrary state $\tau_E\in\mathcal D(E)$ and define
$\mathcal R\in\mathrm{CPTP}(E',E)$ by
\begin{equation}
    \mathcal R(Y)
    \coloneqq
    W^\dagger YW
    +
    \Tr\!\left[
        \bigl(
            \id_{E'}-WW^\dagger
        \bigr)Y
    \right]
    \tau_E ,
    \qquad
    Y\in\mathcal L(E').
\end{equation}
This channel satisfies
\begin{equation}
    \mathcal R
    \bigl(
        WXW^\dagger
    \bigr)
    =
    X
    \qquad
    \text{for every }X\in\mathcal L(E).
\end{equation}

Define
\begin{equation}
    \mathcal N_{A^n\to B^nE^n}
    \coloneqq
    \bigl(
        \idc_{B^n}\otimes\mathcal R^{\otimes n}
    \bigr)
    \circ
    \mathcal N'_{A^n\to B^nE'^n}.
\end{equation}
Since $\mathcal R$ is trace preserving,
$\mathcal N_{A^n\to B^nE^n}$ is an extension of
$\mathcal N_{A^n\to B^n}$. Moreover,
$\mathcal R^{\otimes n}$ commutes with permutations of the $n$
environment registers, so permutation covariance is preserved.

It remains to verify the defect-support condition. Let
\begin{equation}
    \tau_E
    =
    \sum_a
    \lambda_a
    |a\rangle\langle a|_E
\end{equation}
be a spectral decomposition, and let
$\{|b\rangle_{E'}\}_b$ be an orthonormal basis of
$\operatorname{ran}(W)^\perp$. The channel $\mathcal R$ admits the
Kraus representation
\begin{equation}
    R_0
    \coloneqq
    W^\dagger,
    \qquad
    R_{a,b}
    \coloneqq
    \sqrt{\lambda_a}\,
    |a\rangle_E\langle b|_{E'}.
\end{equation}
These Kraus operators satisfy
\begin{equation}
    R_0W
    =
    \id_E,
    \qquad
    R_{a,b}W
    =
    0.
\end{equation}
Consequently,
{\begin{equation}
    \bigl(
        \id_B\otimes R_0\otimes\id_{A'}
    \bigr)
    |V_\Phi'\rangle\!\rangle_{BE'A'}
    =
    |V_\Phi\rangle\!\rangle_{BEA'},
\end{equation}}
whereas
\begin{equation}
    \bigl(
        \id_B\otimes R_{a,b}\otimes\id_{A'}
    \bigr)
    |V_\Phi'\rangle\!\rangle_{BE'A'}
    =
    0.
\end{equation}

Consider any generator containing $n-r$ reference factors
$|V'_\Phi\rangle\!\rangle$. Under any product Kraus operator of
$\mathcal R^{\otimes n}$, a nonzero image is possible only if $R_0$
acts on every reference factor. Each such factor is then mapped to
$|V_\Phi\rangle\!\rangle$, while the remaining $r$ factors are mapped
to an arbitrary joint vector. Hence every nonzero Kraus image belongs
to
\begin{equation}
    \mathcal V\!\left(
        \mathcal H_{BEA'}^{\otimes n},
        |V_\Phi\rangle\!\rangle_{BEA'}^{\otimes(n-r)}
    \right).
\end{equation}
It follows that
\begin{equation}
    \supp
    J\!\left(
        \mathcal N_{A^n\to B^nE^n}
    \right)
    \subseteq
    \operatorname{span}
    \mathcal V\!\left(
        \mathcal H_{BEA'}^{\otimes n},
        |V_\Phi\rangle\!\rangle_{BEA'}^{\otimes(n-r)}
    \right).
\end{equation}
The vectorization and basis arguments used in the proof of
\zcref{Prop:alternative-characterization} now give the required
operator decomposition based on the fixed isometry $V_\Phi$.
\end{proof}
%%%%%%%%%%%%%%%%%%%%%%%%%%%%%%%%%%%%%%%%%%%%%%%%%%%%%%%%%%%%%%%%%%%%%%%%%%%%%%%%
\section{Proof of~\zcref{Prop:output-Wasserstein-depth}}
\label{app:proof-of-Wasserstein-Prop}

\begin{proof}
By convexity of the \(W_1\) norm and linearity of the channels, it
suffices to consider a pure input state
\begin{equation}
    \psi_{A^n}=|\psi\rangle\langle\psi|_{A^n},
    \qquad
    |\psi\rangle_{A^n}
    =\bigotimes_{j=1}^{M}|\psi_j\rangle_{A_{I_j}},
    \qquad |I_j|\leq\ell,
\end{equation}
for one of the partitions in the definition of
\(\mathrm{SEP}_\ell(A_1:\cdots:A_n)\). We adapt the argument for
almost-iid states in~\cite[Lemma~2.3]{mazzola2026robust}.
Throughout this proof, \(W_1\) on \(B^nE^n\) is defined with respect
to the \(n\) composite sites \(B_1E_1,\ldots,B_nE_n\).

Let \(V_\Phi\in\mathcal{L}(A,BE)\) and
\(\mathcal N_{A^n\to B^nE^n}\) be supplied
by~\zcref{Prop:alternative-characterization}. Set
\begin{equation}
    \omega_{B^nE^n}
    \coloneqq
    \mathcal N_{A^n\to B^nE^n}(\psi_{A^n})
    =
    \sum_{t,t'\in\mathcal T}
    \beta_{tt'}K_t|\psi\rangle\langle\psi|K_{t'}^\dagger,
\end{equation}
where each \(K_t\) belongs to
\(\mathcal K(A^n\to B^nE^n,V_\Phi^{\otimes(n-r)})\).
For every \(j\in[M]\), define
\begin{equation}
    |\theta_j\rangle_{(BE)_{I_j}}
    \coloneqq
    V_\Phi^{\otimes I_j}|\psi_j\rangle_{A_{I_j}}.
\end{equation}
The output of the reference Stinespring isometry is
\begin{equation}
    |\theta\rangle
    \coloneqq V_\Phi^{\otimes n}|\psi\rangle
    =\bigotimes_{j=1}^{M}|\theta_j\rangle .
\end{equation}

For each \(t\), the form of \(K_t\) gives a set
\(G_t\subseteq[n]\), with \(|G_t|\leq r\), such that \(K_t\)
acts as \(V_\Phi\) on every position outside \(G_t\).
Let
\begin{equation}
    F_t\coloneqq\{j\in[M]:I_j\cap G_t\neq\varnothing\},
    \qquad
    I_{F_t}\coloneqq\bigcup_{j\in F_t}I_j.
\end{equation}
Since the blocks \(I_j\) are disjoint, \(|F_t|\leq|G_t|\leq r\).
Moreover, the input factorizes across these blocks. Thus every block
\(j\notin F_t\) contributes the factor \(|\theta_j\rangle\) to
\(K_t|\psi\rangle\), while the joint output on the remaining blocks
can be collected into a possibly unnormalized vector
\(|\xi_t\rangle_{(BE)_{I_{F_t}}}\). Consequently,
\begin{equation}
    K_t|\psi\rangle
    =
    \bigotimes_{j\notin F_t}
    |\theta_j\rangle_{(BE)_{I_j}}
    \otimes
    |\xi_t\rangle_{(BE)_{I_{F_t}}}.
    \label{Eq:block-defect-output}
\end{equation}
It follows, including for the cross terms with \(t\neq t'\),
that \(\supp\omega\subseteq\mathcal W\), where
\begin{equation}
    \mathcal W
    \coloneqq
    \operatorname{span}
    \left\{
        \bigotimes_{j\notin F}|\theta_j\rangle
        \otimes|\xi_F\rangle:
        F\subseteq[M],\ |F|\leq r,\ 
        |\xi_F\rangle\in
        \bigotimes_{j\in F}\mathcal H_{(BE)_{I_j}}
    \right\}.
    \label{Eq:depth-defect-subspace}
\end{equation}

For each \(j\), extend \(|\theta_j\rangle\) to an orthonormal
basis \(\{|e_{j,\nu}\rangle\}_{\nu=0}^{d_j-1}\) of
\(\mathcal H_{(BE)_{I_j}}\), where
\(d_j=\dim\mathcal H_{(BE)_{I_j}}\) and
\(|e_{j,0}\rangle=|\theta_j\rangle\). For a multi-index
\(\boldsymbol{\nu}=(\nu_1,\ldots,\nu_M)\), define
\begin{equation}
    |\Psi_{\boldsymbol{\nu}}\rangle
    \coloneqq\bigotimes_{j=1}^{M}|e_{j,\nu_j}\rangle,
    \qquad
    H_{\boldsymbol{\nu}}
    \coloneqq\{j\in[M]:\nu_j\neq0\}.
\end{equation}
Then
\(\{|\Psi_{\boldsymbol{\nu}}\rangle:
|H_{\boldsymbol{\nu}}|\leq r\}\)
is an orthonormal basis of \(\mathcal W\). Indeed, a vector in the
displayed set belongs to \(\mathcal W\) by taking
\(F=H_{\boldsymbol{\nu}}\). Conversely, expanding the arbitrary
vector \(|\xi_F\rangle\) in the product basis on the blocks in \(F\)
expresses every generator of \(\mathcal W\) as a linear combination
of vectors with \(H_{\boldsymbol{\nu}}\subseteq F\).

Let
\(\mathcal S\coloneqq
\{\boldsymbol{\nu}:|H_{\boldsymbol{\nu}}|\leq r\}\).
Since \(\omega\) is a state supported on \(\mathcal W\), it has an
expansion
\begin{equation}
    \omega
    =
    \sum_{\boldsymbol{\nu},\boldsymbol{\mu}\in\mathcal S}
    \gamma_{\boldsymbol{\nu}\boldsymbol{\mu}}
    |\Psi_{\boldsymbol{\nu}}\rangle
    \langle\Psi_{\boldsymbol{\mu}}|,
    \qquad
    \gamma\geq0,\quad
    \sum_{\boldsymbol{\nu}\in\mathcal S}
    \gamma_{\boldsymbol{\nu}\boldsymbol{\nu}}=1.
    \label{Eq:depth-omega-expansion}
\end{equation}
For \(j\in[M]\), set
\begin{equation}
    w_j
    \coloneqq
    \sum_{\boldsymbol{\nu}\in\mathcal S}
    \gamma_{\boldsymbol{\nu}\boldsymbol{\nu}}
    \mathbf 1_{\{j\in H_{\boldsymbol{\nu}}\}}.
    \label{Eq:depth-wj}
\end{equation}
Every \(\boldsymbol{\nu}\in\mathcal S\) has
\(|H_{\boldsymbol{\nu}}|\leq r\), so
\begin{equation}
    \sum_{j=1}^{M}w_j
    =
    \sum_{\boldsymbol{\nu}\in\mathcal S}
    \gamma_{\boldsymbol{\nu}\boldsymbol{\nu}}
    |H_{\boldsymbol{\nu}}|
    \leq r.
    \label{Eq:depth-weight-sum}
\end{equation}

Define the projection
\begin{equation}
    \Pi_j
    \coloneqq
    \sum_{\substack{\boldsymbol{\nu}\in\mathcal S\\
                    j\notin H_{\boldsymbol{\nu}}}}
    |\Psi_{\boldsymbol{\nu}}\rangle
    \langle\Psi_{\boldsymbol{\nu}}|,
    \qquad
    p_j\coloneqq\Tr[\Pi_j\omega]=1-w_j.
\end{equation}
If \(p_j>0\), set
\(\widetilde{\omega}^{(j)}
\coloneqq\Pi_j\omega\Pi_j/p_j\).
If \(p_j=0\), choose any state \(\eta^{(j)}\) on the blocks other
than \(j\) and set
\(\widetilde{\omega}^{(j)}
\coloneqq|\theta_j\rangle\langle\theta_j|\otimes\eta^{(j)}\).
Since every vector in the range of \(\Pi_j\) has the factor
\(|\theta_j\rangle\) on block \(j\), in either case there is a state
\(\eta^{(j)}\) on the complementary blocks such that
\begin{equation}
    \widetilde{\omega}^{(j)}
    =
    |\theta_j\rangle\langle\theta_j|\otimes\eta^{(j)}.
    \label{Eq:depth-projected-factor}
\end{equation}
Furthermore,
\begin{equation}
    \|\omega-\widetilde{\omega}^{(j)}\|_1
    \leq 2\sqrt{w_j}.
    \label{eq:gentle}
\end{equation}
To see this when \(p_j>0\), purify \(\omega\) and apply
\(\Pi_j\) to its \(B^nE^n\) register. The normalized projected
purification has overlap \(\sqrt{p_j}\) with the original one, so
contractivity under partial trace gives the bound
\(2\sqrt{1-p_j}=2\sqrt{w_j}\). When \(p_j=0\), the bound follows
from
\(\|\omega-\widetilde{\omega}^{(j)}\|_1\leq2=2\sqrt{w_j}\).

For \(j=0,\ldots,M\), let
\(L_j\coloneqq I_1\sqcup\cdots\sqcup I_j\), with
\(L_0=\varnothing\), and define
\begin{equation}
    \tau^{(j)}
    \coloneqq
    \omega_{L_j}\otimes
    \bigotimes_{k=j+1}^{M}|\theta_k\rangle\langle\theta_k|,
\end{equation}
where \(\omega_{L_j}\) is the marginal of \(\omega\) on
\((BE)_{L_j}\), and \(\omega_{L_0}=1\). Then
\(\tau^{(0)}=|\theta\rangle\langle\theta|\) and
\(\tau^{(M)}=\omega\). Thus
\begin{equation}
    \omega-|\theta\rangle\langle\theta|
    =
    \sum_{j=1}^{M}
    \bigl(\tau^{(j)}-\tau^{(j-1)}\bigr).
    \label{eq:telescope}
\end{equation}
For every \(j\),
\(\Tr_{(BE)_{I_j}}
[\tau^{(j)}-\tau^{(j-1)}]=0\).
Applying~\cite[Proposition~5]{de2021quantum} to the composite
sites \(B_iE_i\), and dropping its dimension-dependent factor
smaller than one, gives
\begin{equation}
    \|\tau^{(j)}-\tau^{(j-1)}\|_{W_1}
    \leq
    |I_j|\,\|\tau^{(j)}-\tau^{(j-1)}\|_1.
    \label{Eq:depth-local-W1}
\end{equation}

The trace norm on the right satisfies
\begin{align}
    \|\tau^{(j)}-\tau^{(j-1)}\|_1
    &=
    \bigl\|
        \omega_{L_j}
        -
        \omega_{L_{j-1}}
        \otimes|\theta_j\rangle\langle\theta_j|
    \bigr\|_1 \notag\\
    &\leq
    \bigl\|
        \omega_{L_j}
        -
        \widetilde{\omega}^{(j)}_{L_j}
    \bigr\|_1
    +
    \bigl\|
        \widetilde{\omega}^{(j)}_{L_{j-1}}
        -
        \omega_{L_{j-1}}
    \bigr\|_1 \notag\\
    &\leq
    2\|\omega-\widetilde{\omega}^{(j)}\|_1
    \leq4\sqrt{w_j}.
\end{align}
Here the first inequality uses
\(\widetilde{\omega}^{(j)}_{L_j}
=\widetilde{\omega}^{(j)}_{L_{j-1}}
 \otimes|\theta_j\rangle\langle\theta_j|\), and the second uses
contractivity of the trace norm under partial trace.
Combining this estimate with
\zcref{eq:telescope} and \zcref{Eq:depth-local-W1}, we obtain
\begin{align}
    \|\omega-|\theta\rangle\langle\theta|\|_{W_1}
    &\leq
    4\sum_{j=1}^{M}|I_j|\sqrt{w_j} \notag\\
    &\leq
    4\sqrt{\sum_{j=1}^{M}|I_j|^2}
     \sqrt{\sum_{j=1}^{M}w_j}
    \leq4\sqrt{\ell nr},
\end{align}
where the last step uses \zcref{Eq:depth-weight-sum} and
\(\sum_j|I_j|^2\leq\ell\sum_j|I_j|=\ell n\).

Finally,
\begin{equation}
    \Tr_{E^n}\omega
    =\mathcal N_{A^n\to B^n}(\psi_{A^n}),
    \qquad
    \Tr_{E^n}|\theta\rangle\langle\theta|
    =\Phi_{A\to B}^{\otimes n}(\psi_{A^n}).
\end{equation}
The map
\(\Tr_{E^n}=\bigotimes_{i=1}^{n}\Tr_{E_i}\)
acts locally on the composite sites \(B_iE_i\) and cannot
increase \(W_1\). Hence
\begin{equation}
    \frac1n
    \left\|
        \mathcal N_{A^n\to B^n}(\psi_{A^n})
        -
        \Phi_{A\to B}^{\otimes n}(\psi_{A^n})
    \right\|_{W_1}
    \leq
    \frac1n
    \|\omega-|\theta\rangle\langle\theta|\|_{W_1}
    \leq4\sqrt{\frac{\ell r}{n}}.
\end{equation}
Convexity of \(W_1\) extends the same bound to convex combinations
of these pure inputs, including combinations whose pure states use
different partitions. These are precisely the states in
\(\mathrm{SEP}_\ell(A_1:\cdots:A_n)\).
\end{proof}

%%%%%%%%%%%%%%%%%%%%%%%%%%%%%%%%%%%%%%%%%%%%%%%%%%%%%%%%%%%%%%%%%%%%%%%%%%%%%%%%
\section{Quantum error-correction facts}
\label{App:QEC}

For completeness, we collect here the basic facts about quantum
error-correcting codes that are used in
\zcref{Sec:robust-quantum-communication}. We follow the terminology and
results of~\cite{gottesman2024surviving}, in particular Definitions~1.9,
2.1, and 2.3, Theorem~2.4, Theorem~2.7 on the Knill-Laflamme conditions, Corollary~2.8, and Theorem~7.1 on the quantum Gilbert--Varshamov bound.

\begin{definition}[(Weight and $t$-qubit errors)]
\label{Def:t-qubit-error}
Let $\mathcal H_N\coloneqq (\mathbb C^2)^{\otimes N}$ be an $N$-qubit Hilbert space. The \emph{weight} $\mathrm{wt}(E)$ of an operator $E\in\mathcal L(\mathcal H_N)$ is the minimum cardinality of a set $S\subseteq[N]$ such that $E=E_S\otimes\id_{S^c}$ for some $E_S\in\mathcal L((\mathbb C^2)^{\otimes S})$.
A \emph{$t$-qubit error} is any operator of the form
\begin{equation}
    E=\sum_\alpha E_\alpha,
    \qquad
    \mathrm{wt}(E_\alpha)\leq t .
\end{equation}
The space of $t$-qubit errors is denoted
\begin{equation}
    \mathcal E_t^{(N)}
    \coloneqq 
    \operatorname{span}
    \left\{
        E\in\mathcal L(\mathcal H_N):
        \mathrm{wt}(E)\leq t
    \right\}.
    \label{Eq:t-qubit-error-space}
\end{equation}
\end{definition}

Notice that the summands in a $t$-qubit error need not have the same
support: $\mathcal E_t^{(N)}$ contains coherent superpositions of errors occurring on different sets of at most $t$
qubits. This feature will be essential in our application.

\begin{definition}[(Quantum error-correcting code)]
\label{Def:QECC}
Let $\mathcal H_M$ be a logical Hilbert space and $U:\mathcal H_M\to\mathcal H_N$ an encoding isometry. Given a set
$\mathcal E\subseteq\mathcal L(\mathcal H_N)$ of errors, we say that
$(U,\mathcal E)$ is a \emph{quantum error-correcting code} (QECC) if
there exists a recovery channel $\mathcal R\in\mathrm{CPTP}(\mathcal H_N,\mathcal H_M)$ such that, for every $E\in\mathcal E$ and every state
$|\psi\rangle\in\mathcal H_M$,
\begin{equation}
    \mathcal R
    \left(
        EU|\psi\rangle\langle\psi|U^\dagger E^\dagger
    \right)
    =
    c_E\,|\psi\rangle\langle\psi|
    \label{Eq:QECC-definition}
\end{equation}
for some constant $c_E$. The subspace $\mathcal C\coloneqq \mathrm{Image}(U)\subseteq\mathcal H_N$ is called the \emph{code space}.
\end{definition}

\begin{theorem}[(Knill-Laflamme conditions)]
\label{Thm:Knill-Laflamme}
$(U,\mathcal E)$ is a QECC if and only if, for any
$E_a,E_b\in\mathcal E$, and $P_{\mathcal C}=UU^\dagger$ the projector onto the
code space,
\begin{equation}
\label{Eq:Knill-Laflamme}
    P_{\mathcal C}E_a^\dagger E_b P_{\mathcal C}
    =
    C_{a,b}P_{\mathcal C}.
\end{equation}
\end{theorem}
\begin{proof}
    See Theorem~2.7 of~\cite{gottesman2024surviving}.
\end{proof}

\begin{definition}[(Distance of a QECC)]
\label{Def:QECC-distance}
Let $U:\mathcal H_M\to(\mathbb C^2)^{\otimes N}$ be an encoding isometry and $\mathcal C\subseteq (\mathbb C^2)^{\otimes N}$ the corresponding code space. The \emph{distance} of $\mathcal C$ is
\begin{equation}
    d(\mathcal C)
    \coloneqq 
    \min\left\{
        \mathrm{wt}(F):
        F\in\mathcal L\bigl((\mathbb C^2)^{\otimes N}\bigr),\
        P_{\mathcal C}FP_{\mathcal C}
        \notin \mathbb C P_{\mathcal C}
    \right\}.
\end{equation}
If $\dim\mathcal C=K$ and $d(\mathcal C)=d$, we call $\mathcal C$ an
$((N,K,d))$ QECC. In particular, if $K=2^s$, the code encodes $s$
logical qubits into $N$ physical qubits.
\end{definition}

\begin{corollary}
\label{Cor:distance-corrects-errors}
An $((N,K,d))$ QECC corrects any $t$-qubit error for $t\leq
    \left\lfloor\frac{d-1}{2}\right\rfloor$. Equivalently, a code space with distance $d$ at least $2t+1$ corrects all $t$-qubit errors $\mathcal E_t^{(N)}$.
\end{corollary}

\begin{proof}
This is Corollary~2.8 of~\cite{gottesman2024surviving}. For completeness,
let $E_a,E_b\in\mathcal E_t^{(N)}$ and write
\begin{equation}
    E_a=\sum_\alpha E_\alpha,
    \qquad
    E_b=\sum_\beta E_\beta,
    \qquad
    \mathrm{wt}(E_\alpha),\mathrm{wt}(E_\beta)\leq t.
\end{equation}
Then $\mathrm{wt}(E_\alpha^\dagger E_\beta)\leq2t<d,$
and hence, by the definition of distance,
\begin{equation}
    P_{\mathcal C}E_a^\dagger E_b P_{\mathcal C}=\sum_{\alpha,\beta}P_{\mathcal C}E_\alpha^\dagger E_\beta P_{\mathcal C}
    =\sum_{\alpha,\beta} C'_{\alpha,\beta} P_{\mathcal C} = C_{a,b}P_{\mathcal C}.
\end{equation}
By the Knill-Laflamme conditions (\zcref{Thm:Knill-Laflamme}), $(U,\mathcal E_t^{(N)})$ is a QECC.
\end{proof}

\begin{lemma}[(Stinespring dilation of the recovery channel)]
\label{Lem:Stinespring-recovery}
Let $U:\mathcal H_M\to\mathcal H_N= (\mathbb C^2)^{\otimes N}$ be the encoding isometry of an $((N,K,d))$ QECC with $d\geq 2t+1$. Then there exist a purifying (dilation) Hilbert space $\mathcal H_S$, an
isometry
\begin{equation}
    \widehat R:\mathcal H_N\to\mathcal H_M\otimes\mathcal H_S,
\end{equation}
and for every $E\in\mathcal E_t^{(N)}$, a vector
$|\xi_E\rangle_S\in\mathcal H_S$, such that
\begin{equation}
    \widehat R E U
    =
    \id_M\otimes|\xi_E\rangle_S .
    \label{Eq:Stinespring-recovery}
\end{equation}
Consequently,
\begin{equation}
    \mathcal R(\cdot)=\Tr_S\left(\widehat R(\cdot)\widehat R^\dagger\right)
\end{equation}
is a recovery channel for all $t$-qubit errors.
\end{lemma}

\begin{proof}
This is essentially Proposition~2.2 in \cite{gottesman2024surviving}. We briefly reproduce the proof here. Since $d\geq2t+1$, \zcref{Cor:distance-corrects-errors} implies that the code corrects every
$E\in\mathcal E_t^{(N)}$. Hence there exists a recovery channel
$\mathcal R\in\mathrm{CPTP}(\mathcal H_N,\mathcal H_M)$ such that,
for every pure state $|\psi\rangle\in\mathcal H_M$,
\begin{equation}
    \mathcal R
    \left(
        EU|\psi\rangle\langle\psi|U^\dagger E^\dagger
    \right)
    =
    c_E|\psi\rangle\langle\psi|.
\end{equation}
Let $\widehat R:\mathcal H_N\to\mathcal H_M\otimes\mathcal H_S$
be a Stinespring dilation of $\mathcal R$ so that
\begin{equation}
    \Tr_S
    \left[
        \widehat R EU|\psi\rangle\langle\psi|
        U^\dagger E^\dagger\widehat R^\dagger
    \right]
    =
    c_E|\psi\rangle\langle\psi|.
\end{equation}
Since the purification of a pure state has to be a product state, we have $\widehat R EU|\psi\rangle
    =
    |\psi\rangle\otimes|\xi_{E,\psi}\rangle$ for some vector $|\xi_{E,\psi}\rangle\in\mathcal H_S$. We now show that this vector is independent of $|\psi\rangle$. Let $\{|j\rangle\}_j$ be an orthonormal basis of $\mathcal H_M$, and apply the map
\begin{equation}
    \widehat R EU|j\rangle
    =
    |j\rangle\otimes|\xi_{E,j}\rangle, \quad \widehat R EU|k\rangle
    =
    |k\rangle\otimes|\xi_{E,k}\rangle.
\end{equation}
Applying the map to the superposed input state gives
\begin{equation}
    \widehat R EU (|j\rangle+|k\rangle)
    =
    (|j\rangle+|k\rangle)\otimes|\xi_{E,jk}\rangle,
\end{equation}
so that by linearity
\begin{equation}
    |j\rangle\otimes|\xi_{E,j}\rangle
    +
    |k\rangle\otimes|\xi_{E,k}\rangle
    =
    (|j\rangle+|k\rangle)\otimes|\xi_{E,jk}\rangle .
\end{equation}
Taking inner products with $\langle j|$ and $\langle k|$ reveals that (for $j\neq k$) $|\xi_{E,j}\rangle=|\xi_{E,k}\rangle=|\xi_{E,jk}\rangle$, so the vector on the purifying space is independent of $|\psi\rangle$.
\end{proof}

\begin{theorem}[(Quantum Gilbert--Varshamov bound)]
\label{Thm:quantum-GV}
Suppose that $n,k,d\in\mathbb N$ satisfy
\begin{equation}
    2^k
    \sum_{j=0}^{d-1}
    3^j\binom{n}{j}
    \leq
    2^n .
    \label{Eq:quantum-GV}
\end{equation}
Then there exists an $((n,2^k,d))$ QECC. For large $n$, such a code
exists if
\begin{equation}
    \frac{k}{n}
    \leq
    1-\frac{d}{n}\log 3
    -h\!\left(\frac{d}{n}\right),
    \label{Eq:quantum-GV-asymptotic}
\end{equation}
where $h(x)
    \coloneqq 
    -x\log x-(1-x)\log (1-x)$ is the binary entropy.
\end{theorem}

\begin{proof}
    See~Theorem 7.1 of \cite{gottesman2024surviving}.
\end{proof}

%%%%%%%%%%%%%%%%%%%%%%%%%%%%%%%%%%%%%%%%%%%%%%%%%%%%%%%%%%%%%%%%%%%%%%%%%%%%%%%%
\section{Proof of~\zcref{lem:blocking}}
\label{app:proof-of-lem-blocking}

\begin{proof}
By~\zcref{Prop:alternative-characterization}, there exist a system $E$,
a Stinespring isometry $V_\Phi\in\mathcal{L}(A,BE)$ for $\Phi$, and a
permutation-covariant extension channel
$\mathcal N^{(n)}_{A^n\to B^nE^n}$ of
$\mathcal N^{(n)}_{A^n\to B^n}$ such that
\begin{equation}
    \mathcal N^{(n)}_{A^n\to B^nE^n}(\rho)
    =
    \sum_{t,t'\in\mathcal T}
    \beta_{tt'}K_t\rho K_{t'}^\dagger,
    \label{Eq:blocking-original-extension}
\end{equation}
where $\beta=(\beta_{tt'})_{t,t'\in\mathcal T}$ is positive
semidefinite and
\begin{equation}
    K_t
    \in
    \mathcal K\!\left(
        A^n\to(BE)^n,
        V_\Phi^{\otimes(n-r)}
    \right)
    \qquad
    \text{for every }t\in\mathcal T.
\end{equation}

Let $V_{\mathcal E}:Q\to A^kE_{\mathcal E}$ and
$V_{\mathcal D}:B^k\to QE_{\mathcal D}$ be Stinespring isometries for
$\mathcal E$ and $\mathcal D$, respectively, and define
\(
    F\coloneqq E_{\mathcal D}E^kE_{\mathcal E}.
\)
Then
\begin{equation}
    W
    \coloneqq
    \bigl(
        V_{\mathcal D}\otimes\id_{E^kE_{\mathcal E}}
    \bigr)
    \bigl(
        V_\Phi^{\otimes k}\otimes\id_{E_{\mathcal E}}
    \bigr)
    V_{\mathcal E}
    :
    Q\longrightarrow QF
    \label{Eq:effective-dilation}
\end{equation}
is a Stinespring isometry for $\Lambda$; in particular,
\begin{equation}
    \Lambda(\rho)
    =
    \Tr_F\!\left[W\rho W^\dagger\right].
\end{equation}

We now construct an extension channel
$\mathcal M^{(m)}_{Q^m\to Q^mF^m}$ of $\mathcal M^{(m)}$. We identify
\begin{equation}
    F^m
    =
    E_{\mathcal D}^mE^{mk}E_{\mathcal E}^m
\end{equation}
through the natural blockwise regrouping of tensor factors and use the
same convention in the expression below. Define
\begin{align}
    \mathcal M^{(m)}_{Q^m\to Q^mF^m}(\rho)
    \coloneqq{}&
    \left(
        \Ad_{V_{\mathcal D}}^{\otimes m}
        \otimes
        \idc_{E^{mk}E_{\mathcal E}^m}
    \right)
    \Bigg[
        \Tr_{B^sE^s}
        \Bigg(
            \left(
                \mathcal N^{(n)}_{A^n\to B^nE^n}
                \otimes
                \idc_{E_{\mathcal E}^m}
            \right)
            \left(
                \Ad_{V_{\mathcal E}}^{\otimes m}(\rho)
                \otimes
                \varphi_A^{\otimes s}
            \right)
        \Bigg)
    \Bigg],
    \label{Eq:block-extension}
\end{align}
where $\Ad_V(\cdot)\coloneqq V(\cdot)V^\dagger$. Here
$V_{\mathcal D}^{\otimes m}$ acts on $B^{mk}$, grouped into $m$
blocks of size $k$. Tracing out
$F^m=E_{\mathcal D}^mE^{mk}E_{\mathcal E}^m$ and using
\begin{equation}
    \Tr_{E_{\mathcal E}}\circ\Ad_{V_{\mathcal E}}
    =\mathcal E,
    \qquad
    \Tr_{E^n}\circ
    \mathcal N^{(n)}_{A^n\to B^nE^n}
    =
    \mathcal N^{(n)}_{A^n\to B^n},
    \qquad
    \Tr_{E_{\mathcal D}}\circ\Ad_{V_{\mathcal D}}
    =\mathcal D,
\end{equation}
we obtain
\begin{equation}
    \Tr_{F^m}
    \mathcal M^{(m)}_{Q^m\to Q^mF^m}(\rho)
    =
    \mathcal M^{(m)}_{Q^m\to Q^m}(\rho).
\end{equation}
Thus $\mathcal M^{(m)}_{Q^m\to Q^mF^m}$ is an extension of
$\mathcal M^{(m)}$.

We next verify permutation covariance. Let $\sigma\in S_m$, and let
$\sigma^{[k]}\in S_n$ be the permutation that permutes the first $m$
blocks of size $k$ according to $\sigma$, preserves the order within
each block, and fixes the final $s$ systems. The partial trace over the
final $s$ output systems commutes with these block permutations because
$\sigma^{[k]}$ acts trivially on the tail. Moreover, the tensor-power
isometries $V_{\mathcal E}^{\otimes m}$ and
$V_{\mathcal D}^{\otimes m}$ intertwine the corresponding block
permutations, while $\varphi_A^{\otimes s}$ is unchanged. Therefore,
permutation covariance of
$\mathcal N^{(n)}_{A^n\to B^nE^n}$ under $\sigma^{[k]}$ implies
\begin{equation}
    \mathcal M^{(m)}_{Q^m\to Q^mF^m}
    \left(
        U_{Q^m}^{\sigma}\rho
        \bigl(U_{Q^m}^{\sigma}\bigr)^\dagger
    \right)
    =
    U_{(QF)^m}^{\sigma}
    \mathcal M^{(m)}_{Q^m\to Q^mF^m}(\rho)
    \bigl(U_{(QF)^m}^{\sigma}\bigr)^\dagger.
\end{equation}
Hence $\mathcal M^{(m)}_{Q^m\to Q^mF^m}$ is permutation covariant.

It remains to verify the defect-support condition. Since each $K_t$
belongs to the generating set above, there exist a permutation
$\pi_t\in S_n$ and an operator
$L_t\in\mathcal L(A^r,(BE)^r)$ such that
\begin{equation}
    K_t
    =
    U_{(BE)^n}^{\pi_t}
    \left(
        V_\Phi^{\otimes(n-r)}
        \otimes L_t
    \right)
    \bigl(
        U_{A^n}^{\pi_t}
    \bigr)^\dagger.
\end{equation}
Let $S_t\subseteq[n]$ be the corresponding set of at most $r$ defect
positions. Thus $K_t$ acts as $V_\Phi$ on every tensor factor indexed
by $S_t^c$, while its action on the factors indexed by $S_t$ may be
arbitrary and correlated across those factors.

For each $j\in[m]$, let
\begin{equation}
    I_j
    \coloneqq
    \{(j-1)k+1,\ldots,jk\}
\end{equation}
be the set of channel positions in the $j$th block, and define
\begin{equation}
    G_t
    \coloneqq
    \left\{
        j\in[m]:
        I_j\cap S_t\neq\varnothing
    \right\}.
\end{equation}
Since every affected block contains at least one element of $S_t$,
\begin{equation}
    |G_t|
    \leq
    |S_t|
    \leq
    r.
\end{equation}
Consequently, $K_t$ acts as $V_\Phi^{\otimes k}$ on every block
$I_j$ with $j\notin G_t$.

Let $Z\coloneqq B^sE^s$ denote the discarded output system, and choose
a product orthonormal basis $\{|z\rangle_Z\}_z$ of $Z$. For every
$t\in\mathcal T$ and $z$, define
\begin{align}
    \widehat K_{t,z}
    \coloneqq{}&
    \bigl(
        V_{\mathcal D}^{\otimes m}
        \otimes
        \id_{E^{mk}E_{\mathcal E}^m}
    \bigr)
    \langle z|_Z
    \bigl(
        K_t\otimes\id_{E_{\mathcal E}^m}
    \bigr)
    \bigl(
        V_{\mathcal E}^{\otimes m}
        \otimes|\varphi\rangle_A^{\otimes s}
    \bigr)
    \nonumber\\
    &:
    Q^m\longrightarrow(QF)^m.
    \label{Eq:induced-Kraus}
\end{align}
Substituting~\zcref{Eq:blocking-original-extension} into
\zcref{Eq:block-extension} and evaluating the partial trace over $Z$
in the basis $\{|z\rangle_Z\}_z$ gives
\begin{equation}
    \mathcal M^{(m)}_{Q^m\to Q^mF^m}(\rho)
    =
    \sum_{t,t'\in\mathcal T}\sum_z
    \beta_{tt'}
    \widehat K_{t,z}\rho\widehat K_{t',z}^\dagger.
    \label{Eq:blocking-decomposition}
\end{equation}

For every block $j\notin G_t$, the operator $\widehat K_{t,z}$ acts
as
\begin{equation}
    W
    =
    \bigl(
        V_{\mathcal D}\otimes\id_{E^kE_{\mathcal E}}
    \bigr)
    \bigl(
        V_\Phi^{\otimes k}\otimes\id_{E_{\mathcal E}}
    \bigr)
    V_{\mathcal E}.
\end{equation}
Since $|G_t|\leq r$, there are at least $m-r$ such blocks. Choose any
$m-r$ of them. The remaining $r-|G_t|$ unaffected blocks, together
with the $|G_t|$ affected blocks, can be absorbed into an arbitrary
operator on the remaining $r$ blocks. Therefore, there exist a
permutation $\eta_t\in S_m$ and an operator
$L_{t,z}^{(r)}\in\mathcal L(Q^r,(QF)^r)$ such that
\begin{equation}
    \widehat K_{t,z}
    =
    U_{(QF)^m}^{\eta_t}
    \left(
        W^{\otimes(m-r)}
        \otimes
        L_{t,z}^{(r)}
    \right)
    \bigl(
        U_{Q^m}^{\eta_t}
    \bigr)^\dagger.
\end{equation}
Hence
\begin{equation}
    \widehat K_{t,z}
    \in
    \mathcal K\!\left(
        Q^m\to(QF)^m,
        W^{\otimes(m-r)}
    \right)
    \qquad
    \text{for every }t,z.
    \label{Eq:Khat-defect-space}
\end{equation}

Taking the normalized Choi state in
\zcref{Eq:blocking-decomposition} gives
\begin{equation}
    J\!\left(
        \mathcal M^{(m)}_{Q^m\to Q^mF^m}
    \right)
    =
    \sum_{t,t'\in\mathcal T}\sum_z
    \beta_{tt'}
    |\widehat K_{t,z}\rangle\!\rangle
    \langle\!\langle\widehat K_{t',z}|.
\end{equation}
By~\zcref{Eq:Khat-defect-space} and
\zcref{lem:Defect subspace correspondence}, every vector
$|\widehat K_{t,z}\rangle\!\rangle$ belongs to
\begin{equation}
    \mathcal V\!\left(
        \mathcal H_{QFQ'}^{\otimes m},
        |W\rangle\!\rangle_{QFQ'}^{\otimes(m-r)}
    \right),
\end{equation}
where $Q'\simeq Q$ is the reference system used in the Choi
construction. It follows that
\begin{equation}
    \supp
    J\!\left(
        \mathcal M^{(m)}_{Q^m\to Q^mF^m}
    \right)
    \subseteq
    \operatorname{span}
    \mathcal V\!\left(
        \mathcal H_{QFQ'}^{\otimes m},
        |W\rangle\!\rangle_{QFQ'}^{\otimes(m-r)}
    \right).
\end{equation}

Permutation covariance of
$\mathcal M^{(m)}_{Q^m\to Q^mF^m}$ implies that its Choi state is
invariant under simultaneous permutations of the $m$ composite systems
$Q_iF_iQ_i'$. Moreover, $|W\rangle\!\rangle_{QFQ'}$ is a purification
of $J(\Lambda)$, and
\begin{equation}
    \Tr_{F^m}
    J\!\left(
        \mathcal M^{(m)}_{Q^m\to Q^mF^m}
    \right)
    =
    J\!\left(
        \mathcal M^{(m)}_{Q^m\to Q^m}
    \right).
\end{equation}
Thus the Choi state of $\mathcal M^{(m)}$ admits a
permutation-invariant extension satisfying the required defect-support
condition. Hence $\mathcal M^{(m)}$ is an
$\binom{m}{r}$-Choi-almost-iid channel along $\Lambda$.
\end{proof}
%%%%%%%%%%%%%%%%%%%%%%%%%%%%%%%%%%%%%%%%%%%%%%%%%%%%%%%%%%%%%%%%%%%%%%%%%%%%%%%%
\section{Proof of~\zcref{Prop:outer-QEC}}
\label{app:proof-of-QEC-Prop}
\begin{proof}
Set $d\coloneqq |Q|=2^\ell$. The continuity of Stinespring dilations~\cite{kretschmann2008information,singh2026short} states that if
$\Phi_1,\Phi_2\in\mathrm{CPTP}(A,B)$ have Stinespring isometries $V_1,V_2\in\mathcal{L}(A,BE)$ with any common dilation space $E$, then
\begin{equation}
    \inf_{U_E}
    \left\|
        (\id_B\otimes U_E)V_1-V_2
    \right\|_\infty^2
    \leq
    2\|\Phi_1-\Phi_2\|_\diamond
    \leq
    4\inf_{U_E}
    \left\|
        (\id_B\otimes U_E)V_1-V_2
    \right\|_\infty ,
    \label{Eq:Stinespring-continuity}
\end{equation}
where the infimum is over unitaries on $E$. Choose $E$ to be the minimal Stinespring dilation space for $\Lambda\in\mathrm{CPTP}(Q,Q)$, and apply \zcref{Eq:Stinespring-continuity} to $\Phi_1=\Lambda$ and $\Phi_2=\mathrm{id}_Q$. Let
\begin{equation}
    V_0\in\mathcal{L}(Q,QE),
    \quad
    V_0|\psi\rangle
    \coloneqq
    |\psi\rangle_Q|0\rangle_E,
\end{equation}
be a fixed Stinespring dilation of $\mathrm{id}_Q$. Since $\|\Lambda-\mathrm{id}_Q\|_\diamond\leq\varepsilon$, we may therefore
choose a minimal Stinespring isometry $V\in\mathcal{L}(Q,QE)$ of $\Lambda$ satisfying $\|V-V_0\|_\infty\leq\sqrt{2\varepsilon}$. Equivalently, let $\Delta\coloneqq V-V_0$, then
\begin{equation}
    V=V_0+\Delta,
    \qquad
    \|\Delta\|_\infty\leq\sqrt{2\varepsilon}.
    \label{Eq:V-perturbation}
\end{equation}
Since the Choi rank of $\Lambda$ is at most $d^2$, the minimal dilation has dimension $d_E\leq d^2$.

By \zcref{cor:fixed-Stinespring-representation} and condition (ii), $\mathcal{M}^{(m)}$ admits the extension and defect-operator representation in \zcref{Prop:alternative-characterization} in terms of the minimal dilation $V$: there exists a
permutation-covariant extension $\mathcal M^{(m)}_{Q^m\to Q^mE^m}$ of $\mathcal M^{(m)}$ such that
\begin{equation}
    \mathcal M^{(m)}_{Q^m\to Q^mE^m}(\rho)
    =
    \sum_{a,b\in\mathcal T}
    \beta_{ab}K_a\rho K_b^\dagger ,
    \label{Eq:outer-extension}
\end{equation}
where $\beta\geq0$, $\Tr\beta=1$, and
$\{K_a\}_{a\in\mathcal T}$ is an orthonormal basis of $\operatorname{span}
    \mathcal K\!\left(
        Q^m\to(QE)^m,V^{\otimes(m-r)}
    \right)$. Moreover, by~\zcref{lem:Orthonormal basis of the defect subspace},
\begin{equation}
    |\mathcal T|
    \leq
    \binom mr(d^2d_E)^r,
    \label{Eq:T-cardinality}
\end{equation}
and each basis element $K_a$ can be chosen to take the form
\begin{equation}
    K_a
    =U^{\pi_a}_{(QE)^m}
    \left(V^{\otimes(m-r)}\otimes L_a\right)
    (U^{\pi_a}_{Q^m})^\dagger,
\end{equation}
for some permutation $\pi_a\in S_m$ and some arbitrary joint operator $L_a:Q^r\to(QE)^r$. Since the $K_a$ are orthonormal with respect to
the normalized Hilbert--Schmidt inner product,
\begin{equation}
    1
    =
    d^{-m}\Tr(K_a^\dagger K_a)
    =
    d^{-r}\Tr(L_a^\dagger L_a),
\end{equation}
and hence
\begin{equation}
    \|L_a\|_\infty
    \leq\|L_a\|_2
    =
    d^{r/2}.
    \label{Eq:La-bound}
\end{equation}
Expand the $m-r$ tensor power of $V=V_0+\Delta$ in $K_a$. For each subset
$J\subseteq[m-r]$, define
\begin{equation}
    T_J
    \coloneqq 
    \bigotimes_{i=1}^{m-r} X_i^{(J)},
    \qquad
    X_i^{(J)}
    \coloneqq 
    \begin{cases}
        \Delta, & i\in J,\\
        V_0,    & i\notin J.
    \end{cases}
\end{equation}
Then
\begin{equation}
    V^{\otimes(m-r)}
    =
    (V_0+\Delta)^{\otimes(m-r)}
    =
    \sum_{J\subseteq[m-r]}T_J,
\end{equation}
and hence
\begin{equation}
    K_a
    =
    \sum_{J\subseteq[m-r]}
    U^{\pi_a}_{(QE)^m}
    \bigl(
        T_J\otimes L_a
    \bigr)
    (U^{\pi_a}_{Q^m})^\dagger .
\end{equation}
Split this expansion into $K_a=K_a^{\leq t}+K_a^{>t}$ with
\begin{equation}
    K_a^{\leq t}
    \coloneqq 
    \sum_{\substack{J\subseteq[m-r]\\ |J|\leq t-r}}
    U^{\pi_a}_{(QE)^m}
    \bigl(
        T_J\otimes L_a
    \bigr)
    (U^{\pi_a}_{Q^m})^\dagger ,\qquad
    K_a^{>t}
    \coloneqq 
    \sum_{\substack{J\subseteq[m-r]\\ |J|\geq t-r+1}}
    U^{\pi_a}_{(QE)^m}
    \bigl(
        T_J\otimes L_a
    \bigr)
    (U^{\pi_a}_{Q^m})^\dagger .
\end{equation}
For the term indexed by the subset \(J\), \(L_a\) acts on \(r\) defective blocks, while, among the remaining \(m-r\) blocks, \(\Delta\) acts on precisely \(|J|\) blocks and \(V_0\) acts on all others. Since \(|J|\leq t-r\), every term in \(K_a^{\leq t}\) acts as \(V_0\) on at least
\begin{equation}
    m-r-|J|\geq m-r-(t-r)=m-t
\end{equation}
blocks. Equivalently, each term in \(K_a^{\leq t}\) can differ from \(V_0\) on at most \(t\) blocks.

The strategy is to correct both the residual error and the error caused by the defects for $K_a^{\leq t}$ via quantum error-correction, and bound the operator norm of $K_a^{>t}$, which consists of branches with too much, and hence not correctable, residual error. Using \zcref{Eq:La-bound} and $\sqrt{2\varepsilon}<1$, we obtain
\begin{equation}
    \|K_a^{>t}\|_\infty
    \leq
    d^{r/2}
    \sum_{j=t-r+1}^{m-r}
    \binom{m-r}{j}
    \|\Delta\|_\infty^j
    \leq
    2^{m-r}d^{r/2}
    (2\varepsilon)^{\frac{t-r+1}{2}} .
    \label{Eq:Ka-high-bound}
\end{equation}

Following the discussions in the end of \zcref{Rem:Kraus-defect}, we diagonalize the positive semidefinite coefficient matrix
$\beta=(\beta_{ab})_{a,b\in\mathcal T}$. Let
$\lambda_\mu\geq0$ be its nonzero eigenvalues and let
\begin{equation}
    u^{(\mu)}
    =
    \bigl(u_a^{(\mu)}\bigr)_{a\in\mathcal T}
    \in\mathbb C^{|\mathcal T|}
\end{equation}
be corresponding orthonormal eigenvectors. Then $\beta_{ab}
    =
    \sum_\mu
    \lambda_\mu
    u_a^{(\mu)}
    \overline{u_b^{(\mu)}}$. Define
\begin{equation}
    A_\mu
    \coloneqq 
    \sqrt{\lambda_\mu}
    \sum_{a\in\mathcal T}
    u_a^{(\mu)}K_a .
    \label{Eq:Kraus-from-beta}
\end{equation}
It follows that
\begin{equation}
    \sum_\mu A_\mu\rho A_\mu^\dagger=
    \sum_{a,b\in\mathcal T}
    \left(
        \sum_\mu
        \lambda_\mu
        u_a^{(\mu)}
        \overline{u_b^{(\mu)}}
    \right)
    K_a\rho K_b^\dagger=
    \sum_{a,b\in\mathcal T}
    \beta_{ab}K_a\rho K_b^\dagger .
\end{equation}
Thus $\{A_\mu\}_\mu$ is a Kraus representation of the extended channel. We can construct the Stinespring isometry $W\coloneqq 
    \sum_\mu A_\mu\otimes|\mu\rangle_F\in\mathcal{L}(Q^m,Q^mE^mF)$ of the extended channel, where $\{|\mu\rangle\}_\mu$ is an O.N.B. of the dilation space. Correspondingly, define
\begin{equation}
    A_\mu^{\leq t}
    \coloneqq 
    \sqrt{\lambda_\mu}
    \sum_{a\in\mathcal T}
    u_a^{(\mu)}K_a^{\leq t},
    \qquad
    A_\mu^{>t}
    \coloneqq 
    \sqrt{\lambda_\mu}
    \sum_{a\in\mathcal T}
    u_a^{(\mu)}K_a^{>t},
\end{equation}
and
\begin{equation}
    W^{\leq t}
    \coloneqq 
    \sum_\mu A_\mu^{\leq t}\otimes|\mu\rangle_F,
    \qquad
    W^{>t}
    \coloneqq 
    \sum_\mu A_\mu^{>t}\otimes|\mu\rangle_F.
\end{equation}

We first bound the operator norm of $W^{>t}$. Let $|v_a\rangle_F
    \coloneqq 
    \sum_\mu
    \sqrt{\lambda_\mu}\,
    u_a^{(\mu)}|\mu\rangle_F$, then
\begin{equation}
    W^{>t}
    =
    \sum_{a\in\mathcal T}
    K_a^{>t}\otimes|v_a\rangle_F .
    \label{Eq:W-bad-expansion}
\end{equation}
Recall that the eigenvectors $u^{(\mu)}$ are normalized and
$\Tr\beta=1$,
\begin{equation}
    \sum_{a\in\mathcal T}\|v_a\|^2
    =
    \sum_{a\in\mathcal T}\sum_\mu
    \lambda_\mu |u_a^{(\mu)}|^2
    =
    \sum_\mu\lambda_\mu
    =
    1.
\end{equation}
Therefore, by the triangle inequality and Cauchy--Schwarz,
\begin{equation}
    \|W^{>t}\|_\infty
    \leq
    \sum_{a\in\mathcal T}
    \|K_a^{>t}\|_\infty\|v_a\|
    \leq
    \max_{a\in\mathcal T}\|K_a^{>t}\|_\infty
    \left(
        |\mathcal T|
        \sum_{a\in\mathcal T}\|v_a\|^2
    \right)^{1/2}
    =
    |\mathcal T|^{1/2}
    \max_{a\in\mathcal T}\|K_a^{>t}\|_\infty .
\end{equation}
Combining with~\zcref{Eq:T-cardinality} and~\zcref{Eq:Ka-high-bound}, we obtain
\begin{equation}
    \|W^{>t}\|_\infty
    \leq
    2^{\,m+\frac{5\ell-2}{2}r}
    \binom mr^{1/2}
    (2\varepsilon)^{\frac{t-r+1}{2}}\eqqcolon \delta.
    \label{Eq:W-bad-bound-final}
\end{equation}

It remains to show that errors within $W^{\leq t}$ can be corrected exactly by
the QECC. Fix an orthonormal basis $\{|z\rangle\}_z$ of the extension space $E^m$, and define
\begin{equation}
    B_{\mu,z}
    \coloneqq 
    \bigl(
        \id_{Q^m}\otimes\langle z|_{E^m}
    \bigr)
    A_\mu^{\leq t}
    \in\mathcal L(Q^m).
\end{equation}
We claim that these are $\ell t$-qubit errors (see~\zcref{Def:t-qubit-error}):
\begin{equation}
    B_{\mu,z}\in
    \mathcal E_{\ell t}^{(\ell m)}
    \qquad
    \text{for every $\mu,z$}.
\end{equation}
To see this, recall that each term entering $K_a^{\leq t}$ acts as $V_0$ on at least $m-t$
blocks. Since $V_0|\psi\rangle
    =
    |\psi\rangle_Q|0\rangle_E$, contracting the environment $E^m$ with $\langle z|$ leaves the identity
operator on these $m-t$ blocks. Hence each term entering $B_{\mu,z}$ acts non-trivially on at most $t$ blocks, that is $\ell t$ qubits. Since $B_{\mu,z}$ is a linear combination of these terms, it is an $\ell t$-qubit error.

By assumption, the QECC has distance $2\ell t+1$ and therefore corrects all
$\ell t$-qubit errors (see~\zcref{Cor:distance-corrects-errors}). By~\zcref{Lem:Stinespring-recovery}, there exists a Stinespring isometry $\widehat R:Q^m\longrightarrow M S$ of the recovery channel
such that for every $\mu,z$ there exists a
vector $|\xi_{\mu,z}\rangle_S$ such that
\begin{equation}
    \widehat R B_{\mu,z}U
    =
    \id_M\otimes|\xi_{\mu,z}\rangle_S .
\end{equation}
Since
\begin{equation}
    A_\mu^{\leq t}
    =
    \sum_z B_{\mu,z}\otimes|z\rangle_{E^m},
\end{equation}
we obtain
\begin{equation}
\label{eq:correctable}
    \bigl(
        \widehat R\otimes\id_{E^mF}
    \bigr)
    W^{\leq t}U
    =
    \sum_{\mu,z}
    \widehat R B_{\mu,z}U
    \otimes|z\rangle_{E^m}|\mu\rangle_F
    =
    \id_M\otimes|\xi\rangle_{SE^mF},
\end{equation}
where $|\xi\rangle
    \coloneqq 
    \sum_{\mu,z}
    |\xi_{\mu,z}\rangle_S
    |z\rangle_{E^m}
    |\mu\rangle_F$. Thus the $W^{\leq t}$ part can be exactly corrected.

For the last step, define the isometry
\begin{equation}
    Y
    \coloneqq 
    \bigl(
        \widehat R\otimes\id_{E^mF}
    \bigr)WU
    :
    M\longrightarrow MSE^mF .
\end{equation}
Since $W$ is the Stinespring dilation of the extended channel $\mathcal M^{(m)}_{Q^m\to Q^mE^m}$,
\begin{equation}
    \Tr_{SE^mF}\circ\Ad_Y
    =
    \mathcal R\circ\mathcal M^{(m)}\circ\mathcal U,
\end{equation}
so $Y$ is the Stinespring dilation of the error-corrected channel. From $W=W^{\leq t}+W^{>t}$ and~\zcref{eq:correctable}, we have
\begin{equation}
    Y
    =
    \id_M\otimes|\xi\rangle
    +
    \bigl(
        \widehat R\otimes\id_{E^mF}
    \bigr)W^{>t}U,
\end{equation}
and therefore combining with~\zcref{Eq:W-bad-bound-final},
\begin{equation}
    \left\|
        Y-\id_M\otimes|\xi\rangle
    \right\|_\infty
    \leq
    \|W^{>t}\|_\infty
    \le
    \delta.
    \label{Eq:Y-xi-bound}
\end{equation}
Define the normalized state\footnote{Notice that $\|\xi\|>0$, otherwise~\zcref{Eq:Y-xi-bound} would imply $1\leq\delta$, in which case the proposition trivially holds since the diamond-norm distance is at most 2.} $|\widehat\xi\rangle
    \coloneqq 
    \frac{|\xi\rangle}{\|\xi\|}$, then
\begin{align}
    \left\|
        Y-\id_M\otimes|\widehat\xi\rangle
    \right\|_\infty
    &\leq
    \left\|
        Y-\id_M\otimes|\xi\rangle
    \right\|_\infty
    +
    \left\|
        \id_M\otimes
        \bigl(
            |\xi\rangle-|\widehat\xi\rangle
        \bigr)
    \right\|_\infty
    \nonumber\\
    &=
    \left\|
        Y-\id_M\otimes|\xi\rangle
    \right\|_\infty
    +
    \bigl|\,\|\xi\|-1\,\bigr|
    \leq2\delta .
\end{align}
Finally, since $\id_M\otimes|\widehat\xi\rangle$ is a Stinespring isometry of $\mathrm{id}_M$, the upper continuity bound in~\zcref{Eq:Stinespring-continuity} gives
\begin{equation}
    \frac12
    \left\|
        \mathcal R\circ\mathcal M^{(m)}\circ\mathcal U
        -\mathrm{id}_M
    \right\|_\diamond
    \leq
    \left\|
        Y-\id_M\otimes|\widehat\xi\rangle
    \right\|_\infty
    \leq
    2^{\,m+1+\frac{5\ell-2}{2}r}
    \binom mr^{1/2}
    (2\varepsilon)^{\frac{t-r+1}{2}} .
\end{equation}
This proves the proposition.
\end{proof}
%%%%%%%%%%%%%%%%%%%%%%%%%%%%%%%%%%%%%%%%%%%%%%%%%%%%%%%%%%%%%%%%%%%%%%%%%%%%%%%%

\end{document}